\documentclass[a4paper,11pt]{article}
\pdfoutput=1

\usepackage{tcs}
\begin{document}

\title{Structure learning of Hamiltonians from real-time evolution}
\author{
Ainesh Bakshi \\
\texttt{ainesh@mit.edu} \\
MIT
\and
Allen Liu \\
\texttt{cliu568@mit.edu} \\
MIT
\and
Ankur Moitra \\
\texttt{moitra@mit.edu} \\
MIT
\and
Ewin Tang \\
\texttt{ewin@berkeley.edu} \\
UC Berkeley
}
\date{}

\maketitle

\begin{abstract}
    We study the problem of Hamiltonian structure learning from real-time evolution: given the ability to apply $e^{-\ii Ht}$ for an unknown local Hamiltonian $H = \sum_{a = 1}^\terms \lambda_a E_a$ on $\qubits$ qubits, the goal is to recover $H$.
    This problem is already well-understood under the assumption that the interaction terms, $E_a$, are given, and only the interaction strengths, $\lambda_a$, are unknown.
    But how efficiently can we learn a local Hamiltonian without prior knowledge of its interaction structure?
    
    We present a new, general approach to Hamiltonian learning that not only solves the challenging structure learning variant, but also resolves other open questions in the area, all while achieving the gold standard of Heisenberg-limited scaling.
    In particular, our algorithm recovers the Hamiltonian to $\eps$ error with total evolution time $\bigO{\log (n)/\eps}$, and has the following appealing properties:
    \begin{enumerate}
        \item It does not need to know the Hamiltonian terms;
        \item It works beyond the short-range setting, extending to any Hamiltonian $H$ where the sum of terms interacting with a qubit has bounded norm;
        \item It evolves according to $H$ in constant time $t$ increments, thus achieving constant time resolution.
    \end{enumerate}
    As an application, we can also learn Hamiltonians exhibiting power-law decay up to accuracy $\eps$ with total evolution time beating the standard limit of $1/\eps^2$.
   
\end{abstract}

\begin{comment}
We study the problem of Hamiltonian structure learning from real-time evolution: given the ability to apply $e^{-\mathrm{i} Ht}$ for an unknown local Hamiltonian $H = \sum_{a = 1}^m \lambda_a E_a$ on $n$ qubits, the goal is to recover $H$. This problem is already well-understood under the assumption that the interaction terms, $E_a$, are given, and only the interaction strengths, $\lambda_a$, are unknown. But how efficiently can we learn a local Hamiltonian without prior knowledge of its interaction structure?
    
We present a new, general approach to Hamiltonian learning that not only solves the challenging structure learning variant, but also resolves other open questions in the area, all while achieving the gold standard of Heisenberg-limited scaling. In particular, our algorithm recovers the Hamiltonian to $\varepsilon$ error with total evolution time $O(\log (n)/\varepsilon)$, and has the following appealing properties:
(1) it does not need to know the Hamiltonian terms;
(2) it works beyond the short-range setting, extending to any Hamiltonian $H$ where the sum of terms interacting with a qubit has bounded norm;
(3) it evolves according to $H$ in constant time $t$ increments, thus achieving constant time resolution.
As an application, we can also learn Hamiltonians exhibiting power-law decay up to accuracy $\varepsilon$ with total evolution time beating the standard limit of $1/\varepsilon^2$.
\end{comment}

\thispagestyle{empty}
\setcounter{page}{0}
\clearpage
\newpage

\microtypesetup{protrusion=false}
\tableofcontents{}
\thispagestyle{empty}
\microtypesetup{protrusion=true}
\setcounter{page}{0}
\clearpage
\setcounter{page}{1}

%!TEX root = main.tex
\section{Introduction}

In this work, we study \emph{Hamiltonian learning from real-time evolution}.
This problem models a fundamental algorithmic challenge in the development of controllable quantum devices: supposing we can engineer a system which performs quantum evolution, how can we characterize its behavior accurately and efficiently?
In physics, this question has a long history in the domains of quantum metrology and quantum sensing~\cite{caves81,hb93,biwh96,drc17}, where it is studied for specific quantum devices with relatively simple evolutions.
We consider Hamiltonian learning construed more broadly, to larger and more complicated systems.
Algorithms for this more general version of the task have seen increased interest as a potentially scalable method for benchmarking quantum computers~\cite{shnbdu14,wgfc14}, since they can be applied to detecting errors or certifying correctness of implementation.

We now define the problem: consider a system of $\qubits$ qubits.
Associated with this system is a Hamiltonian, $H = \sum_{a=1}^\terms \lambda_a E_a$, which encodes the kinds of interactions occurring between the qubits and the strengths of these interactions.
Throughout, we consider Hamiltonians which are local, meaning that the terms $E_1,\dots,E_\terms$ all act on at most a constant number of qubits.
Our goal is to determine $H$, given the ability to evolve quantum states according to $H$, i.e.\ apply the unitary $e^{-\ii Ht}$ for any $t > 0$.

Let us situate this problem in a larger context.
For a Hamiltonian $H$ describing a system, its properties of interest fall into two broad classes: \emph{dynamics}, how an initial state evolves with respect to $H$ over time; and \emph{statics}, how the system behaves at equilibrium, corresponding to its Gibbs states, ground states, or more generally any state which is fixed by evolution by $H$.
There are also two broad classes of algorithmic tasks: first, we could ask the \emph{forward} problem, to simulate the quantum system given the Hamiltonian.
This is called Hamiltonian simulation~\cite{lloyd96} when we want to simulate $H$'s dynamics and Gibbs sampling~\cite{bkllsw19} or ground state preparation when we want to simulate $H$'s statics.
Second, we could ask the \emph{inverse} problem, to learn the Hamiltonian given copies of the quantum system.
In the static case, this is known as Hamiltonian learning from Gibbs states~\cite{AAKS21}.
We study the dynamic case of the inverse problem, Hamiltonian learning from real-time evolution~\cite{slp11,wgfc14}.
These two versions of Hamiltonian learning are appropriate in different experimental settings, and have different challenges.
While learning from time evolution requires a greater degree of experimental control, it also allows us to get stronger learning guarantees.
Since it is the main focus of our paper, we will refer to Hamiltonian learning from real-time evolution as Hamiltonian learning for brevity.

The primary figure of merit for Hamiltonian learning algorithms is \emph{total evolution time}, $\tet$.
This is the amount of time the unknown evolution is applied over the course of the algorithm: for applying $e^{-\ii Ht}$ we associate a cost of $t$, and track the total cost of the algorithm.

It is well-understood how to perform Hamiltonian learning in wide generality using ``derivative estimation'' techniques, achieving $\tet = \bigO{\log(\qubits) / \eps^3}$~\cite{zylb21}.\footnote{
    We are not aware of a rigorous proof in the literature, but see \cref{sec:related} for a derivation.
    Note that the $\log(\qubits)$ here is actually a $\log(\terms)$, but since $H$ is local, it can be specified by $\poly(\qubits)$ parameters, so $\log(\terms) = \bigO{\log(\qubits)}$.
}
For our current discussion, we narrow our focus to algorithms which strictly improve on the $\tet$ from derivative estimation.\footnote{
    It is known how to perform structure learning with Heisenberg-limited error scaling, but worse dependence on system size: $\tet = \poly(\qubits)/\eps$~\cite{oksm24}.
    We discuss this work in \cref{sec:related}.
}
Recent work has made advances in efficiency of Hamiltonian learning under the simplifying assumption that the interaction terms, $E_a$, are known, and only the interaction strengths, $\lambda_a$, are unknown.
In particular, as proved by Huang, Tong, Fang, and Su~\cite{htfs23}, to estimate $H$ to $\eps$ error, one can achieve the scaling $\tet = \bigTheta{\log(\qubits)/\eps}$.
This has the optimal, ``Heisenberg-limited'' scaling in the error, $1/\eps$, better than the ``standard limit'' of $1/\eps^2$ that one might expect.

However, assuming that the interaction terms are known is often not realistic. 
The aforementioned prior work fundamentally requires knowledge of the interaction terms, using it either to compute commutator expansions with respect to these terms~\cite{hkt24} or to reshape the Hamiltonian to decouple terms from the rest of the space~\cite{htfs23}.
The following basic question will be the focus of our work:

\begin{question}
\label{question:central-open-q}
\begin{center}
    \emph{How efficiently can we learn a local Hamiltonian \\ without prior knowledge about what interactions are allowed?}
\end{center}
\end{question}

We refer to this problem as Hamiltonian \emph{structure} learning.
Taking a step back, the hope is that frameworks and modes of analysis from theoretical machine learning can be useful in designing algorithms for learning about quantum systems.
This agenda has already seen important successes, relating algorithms for classical spin systems to learning from Gibbs states~\cite{AAKS21, hkt24, blmt24} and Gibbs sampling~\cite{rfa24,blmt24b}.
We view Hamiltonian structure learning as the next frontier.
It is a natural counterpart to the well-studied classical problem of structure learning in graphical models~\cite{km17}, and with a structure learning algorithm, we can characterize quantum devices without imposing an underlying locality structure.
Further, we aim to achieve the gold standard of $\tet = \bigTheta{\log(\qubits)/\eps}$ even when all interactions are unknown; existing improvements to derivative estimation cannot handle a Hamiltonian with even one unspecified long-range interaction.

Along the way to structure learning, we will also revisit other limitations of existing Hamiltonian learning technology.
Prior work assumes that the Hamiltonian has bounded-range interactions, meaning that, if we imagine the qubits of our system as the vertices of a lattice of small dimension, interaction strengths are exactly zero for interaction terms which extend beyond a certain constant range.
However, many classes of Hamiltonians do not have this structure, and a central open question posed by prior work~\cite{htfs23} is to understand what kind of structure is necessary for Hamiltonian learning.
For example, many physically relevant Hamiltonians have interactions whose strengths exhibit \emph{algebraic decay} in the range of the interaction~\cite{ddmprt23}.
Hamiltonian learning algorithms for this class have been studied prior~\cite{smdwb24}, but are sub-optimal in evolution time and time resolution, and do not work across the full range of decays.
Comparing to the classical setting, learning algorithms for Markov random fields work even under just the constraint that $\sum_{\supp(E_a) \ni i} \abs{\lambda_a}$, the sum of the interaction strengths on site $i$, is bounded for every $i \in [\qubits]$.
This leads us to ask:

\begin{question}
\label{question:learning-without-locality}
\begin{center}
    \emph{How efficiently can we learn a local Hamiltonian \\ without the assumption of strictly bounded range?}
\end{center}
\end{question}

Next, we consider another figure of merit on the quantum resources required by a Hamiltonian learning algorithm.
\emph{Time resolution}, $\tres$, is the smallest value of $t$ we need to apply over the course of the algorithm.
Accurately implementing many small time evolutions, interleaved with other operations, requires a large degree of quantum control over the evolution, which is experimentally challenging.
Thus, time resolution becomes a concern for making Hamiltonian learning feasible in practice~\cite{smdwb24,dos23}.
Ideally, we want the time resolution to be constant; above this, one runs into issues of identifiability \cite[Remark A.6]{hkt24}.

Currently, there is an algorithm with $\tet = \bigTheta{\log(\qubits)/\eps}$ but $\tres = \bigTheta{\sqrt{\eps}}$~\cite{htfs23} and an algorithm with $\tet = \bigTheta{\log(\qubits)/\eps^2}$ and $\tres = \bigTheta{1}$~\cite{hkt24}.
It has been conjectured that it is possible to get the best of both worlds~\cite{dos23}.  Thus, we ask:

\begin{question}
\label{question:time-resolution}
\begin{center}
    \emph{Can we learn a local Hamiltonian \\ with Heisenberg scaling} and \emph{constant time resolution?}
\end{center}
\end{question}

The limitations we have identified are all downstream of a lack of techniques for Hamiltonian learning.
Algorithms for structure learning perform suboptimally in evolution time, and existing approaches for improvement~\cite{htfs23, dos23, hkt24} require a large amount of information about the structure of the Hamiltonian, scale poorly with the number of candidate terms supported on a given site, and, to achieve Heisenberg scaling, interleave the evolution with quantum control at time intervals which depend on $\eps$.
Addressing these issues motivates us to investigate alternative strategies for achieving Heisenberg scaling.

\subsection{Results}

Our main result is an algorithm for structure learning from real-time evolution with $\tet = \bigTheta{\log(\qubits)/\eps}$, addressing Question~\ref{question:central-open-q}.
This algorithm also simultaneously supports Hamiltonians without bounded range and has constant time resolution, addressing Questions~\ref{question:learning-without-locality} and \ref{question:time-resolution}.

First, let us define the class of Hamiltonians under consideration; see \cref{subsec:hamiltonian-of-interacting-system} for more details.
Let $\locals = \{\id, \sigma_x, \sigma_y, \sigma_z\}^{\otimes \qubits} \in \mathbb{C}^{\dims \times \dims}$ denote the set of tensor products of Pauli matrices on $\qubits = \log_2(\dims)$ qubits, and let $\locals_k \subset \locals$ denote those Pauli matrices which are non-identity on at most $k$ qubits.
We consider a $\locality$-local Hamiltonian with $\terms$ terms, $H = \sum_{a=1}^\terms \lambda_a E_a$.
We assume no particular locality structure.
We only assume that the interaction strengths on any particular qubit are bounded:
\begin{align}\label{eq:local1norm-intro}
    \lonorm{H} \deq \max_{i \in [\qubits]} \sum_{\substack{a \in [\terms] \\ \supp(E_a) \ni i}} \abs{\lambda_a} \leq \degree.
\end{align}
Our algorithm will also depend on an \textit{effective sparsity} parameter $\sparse$, which one can think of as smoothly bounding the number of coefficients larger than $\eps$ acting on any given site:
\begin{align}
    \sparse \deq
    \max_{i \in [\qubits]} \sum_{\substack{a \in [\terms] \\ \supp(E_a) \ni i}}\min(1, \lambda_a^2/\eps^2)
    \leq \terms.
\end{align}
With this, we can state our main theorem.

\begin{theorem}[Learning a local Hamiltonian from real-time evolution, \cref{thm:main}] \label{thm:main-informal}
Let $H = \sum_{a=1}^\terms \lambda_a E_a$ be an unknown $n$-qubit, $\locality$-local Hamiltonian.  Suppose $\locality = \bigO{1}$ and $\lonorm{H} = \bigO{1}$.
Then given $0< \eps <1$, there exists a quantum algorithm $\calA$ that outputs a set of estimated terms, $\{(P, \wh{\lambda}_P)\}_{P \in \locals_\locality}$ with the following guarantees:
\begin{enumerate}
    \item \textup{(Accuracy)} With probability $0.99$, $\abs{\wh{\lambda}_{E_a} - \lambda_a} < \eps$ for all $a \in [\terms]$, and $\wh{\lambda}_{P} = 0$ otherwise;
    \item \textup{(Evolution time)} $\alg$ applies $e^{-\ii Ht}$ with a total evolution time of $\tet = \bigO{\sparse \log(\qubits) / \eps  }$;
    \item \textup{(Time resolution)} $\alg$ only applies $e^{-\ii Ht}$ with $t \geq \tres = \bigTheta{1/\sparse}$;
    \item \textup{(Experiment count)} $\alg$ runs $\bigOt{\sparse^2 \log(\qubits)\log(1/\eps )  }$ quantum circuits, all of the form in \cref{fig:the-circuit};
    \item \textup{(Classical overhead)} $\alg$ uses classical computation with $\bigOt{\qubits^2\sparse^3\log(1/\eps)}$ total running time.
\end{enumerate}
\end{theorem}

Prior work on Hamiltonian learning assumes the unknown Hamiltonian is \emph{low-intersection}, meaning that only a constant number of terms interact with each qubit (\cref{def:low-insersection-ham}).
In this case, $\sparse = \bigO{1}$ for all $\eps$ and thus our result matches \cite{htfs23} in evolution time and \cite{hkt24} in time resolution.
Further, our algorithm also works when the structure is unknown, recovering the interaction structure and the interaction strengths to $\eps$ accuracy.\footnote{
    Algorithms which assume knowledge of input terms can handle unspecified terms to a small extent, provided that these do not change the underlying locality structure of the Hamiltonian.
    This can be done by simply considering a broader class of Hamiltonians with more terms.
    For example, for a Hamiltonian on a line, $H = \sum_{i=1}^{n-1} h_{i, i+1}$ with $h_{i, i+1}$ corresponding to unknown terms acting on qubits $i$ and $i+1$, we can learn the terms by expanding every term in the Pauli basis, $h_{i, i+1} = \sum_{P, Q \text{ Pauli}} c_{i,i+1}^{(P, Q)} P_i \otimes Q_{i+1}$, and then learning the parameters $c_{i,i+1}^{(P, Q)}$ of the expanded Hamiltonian.
    However, this fails when the Hamiltonian has even one unspecified long-range interaction, e.g.\ there is an additional term between two qubits, but we are not told which ones.
}
Additionally, our result extends beyond the low-intersection setting to the setting where $H$ could have arbitrary coefficients on all interaction terms.
Provided $\lonorm{H} = \bigO{1}$, our algorithm is guaranteed to recover all of the coefficients that are larger than $\eps$.

Our algorithm even applies to $\bigO{1}$-local Hamiltonians with no locality structure and arbitrary $\lonorm{H}$, though the total time evolution $\tet = \bigO{\terms \log(\qubits)/\eps}$ becomes linear in the number of terms.
This follows from bounding $\sparse \leq \terms$ in \cref{thm:main}.

\begin{figure}
\begin{align*}
    \Qcircuit @R=.6em @C=0.65em {
    \push{\ket{0}} & \qw & \multigate{3}{U} & \multigate{3}{e^{\ii H_0 t}} & \multigate{3}{e^{-\ii H t}} & \qw & {\cdots} & & \multigate{3}{e^{\ii H_0 t}} & \multigate{3}{e^{-\ii H t}}& \multigate{3}{V} &  \qw & \meter \\
    \push{\vdots} &&&&&& \ddots &&&&&& {\vdots}\\
    &&&&&&&&&&& \\
    \push{\ket{0}} & \qw & \ghost{U} & \ghost{e^{\ii H_0 t}} & \ghost{e^{-\ii H t}} & \qw &{\cdots} & & \ghost{e^{\ii H_0 t}} & \ghost{e^{-\ii Ht}}& \ghost{V} & \qw & \meter
    \gategroup{1}{4}{4}{10}{.5em}{_\}} \\
    &&&&&& \leq 1/\eps \text{ times}
    }
\end{align*}
\caption{
    The basic subroutine of our main algorithm.
    All of the circuits we consider are of this form.
    Here, $H_0$ is a known Hamiltonian, and $U$ and $V$ are layers of single-qubit Clifford gates.
} \label{fig:the-circuit}
\end{figure}

\begin{remark}[On $\eps$ dependence]
    As discussed previously, $\tet = \bigO{1/\eps}$ and $\tres = \bigOmega{1}$ attains the optimal dependence on $\eps$.
    By a prior lower bound~\cite[Theorem 18]{dos23}, the number of times we need to interleave the unknown Hamiltonian with other operations, $\bigO{1/\eps}$, is also optimal.
    Further, $\log\frac{1}{\eps}$ is a lower bound on the number of experiments and the classical overhead, since specifying the output Hamiltonian requires $\qubits \log\frac{1}{\eps}$ bits of information, supposing that there are at least $\qubits$ non-zero Hamiltonian coefficients.
    Every experiment output gives $\qubits$ bits of information, so $\log\frac{1}{\eps}$ experiments are necessary.
    Similarly, writing down the output requires $\qubits \log\frac{1}{\eps}$ time.
    So, all of our figures of merit have optimal $\eps$ dependence up to $\log \log\frac{1}{\eps}$ terms.
\end{remark}

\begin{remark}[Other properties]
Prior algorithms on learning unitary operations to Heisenberg scaling, like that of \cite{htfs23}, typically have several other nice properties not mentioned above.
Our algorithm also matches these properties, all while supporting structure learning, long-range interactions, and constant time resolution.
We list them here.
\begin{itemize}
    \item All quantum circuits to be performed are ``prepare-apply-measure'' circuits, in that they take the form of \cref{fig:the-circuit}: initialize every qubit to some Pauli eigenvector; apply $(e^{-\ii H t}e^{\ii H_0 t})^k$ for some $k$, the alternating evolution between a known Hamiltonian $H_0$ and the unknown Hamiltonian $H$; and measure every qubit in some Pauli basis.
    \item Consequently, these circuits can be performed with $\qubits$ qubits, and no space overhead.
    \item The algorithm has $\log(1/\eps)$ rounds of adaptivity, so the quantum circuits can be significantly parallelized.
    \item The algorithm still succeeds even when the quantum circuit has up to $\bigTheta{1/\sparse}$ error per-experiment, such as that caused by SPAM (state preparation and measurement) error.
\end{itemize}
\end{remark}

\begin{remark}[On gate complexity]
    For simplicity, we do not track the gate complexity of the quantum circuits, and simply assume we can apply any unitary matrix.
    This can be computed, though: the dominant cost is the time evolution of a known Hamiltonian, $e^{\ii H_0 t}$.
    The known Hamiltonian and the unknown Hamiltonian are evolved for the same length of time, so this cost is the cost of evolving the known Hamiltonian for $\tet$ time.
    Generally, $H_0$ will always be a Hamiltonian of a similar form to the unknown one---for example, it has the same set of terms---so it is reasonable to expect that evolving with respect to $H_0$ is also efficient.

    For example, suppose that the unknown Hamiltonian is geometrically local on a constant-dimensional lattice.
    Then, the gate complexity of our algorithm is the complexity of evolving with respect to such a Hamiltonian for $\tet$ time, which is $\bigOt{\qubits \tet} = \bigOt{\qubits/\eps}$ when using QSVT-style techniques (which require ancilla qubits)~\cite{hhkl21}, or $\bigO{(\qubits \tet)^{1 + \frac{1}{2k}}} = \bigO{(\qubits/\eps)^{1 + \frac{1}{2k}}}$ when using $(2k)$th order product formulas~\cite{cs19}.
    The latter is comparable to the gate complexity of \cite{htfs23} in this setting, which is $\bigOt{\qubits \eps^{-1.5}}$, though the product formulas demand the use of constant-local gates, whereas single-qubit Cliffords suffice for \cite{htfs23}.
\end{remark}

One surprising aspect of \cref{thm:main-informal} is that we are able to do structure learning in time $\bigOt{\qubits^2}$.
This is fixed-parameter tractable (FPT), meaning that the constant in the exponent does not depend on the locality of the underlying Hamiltonian.
By contrast, classical results about learning Markov random fields state that $\qubits^{\locality}$ running time is required for structure learning with $\locality$-local interactions \cite{km17} (under standard hardness assumptions about learning sparse parities with noise).
This classical structure learning task is a special case of Hamiltonian learning from Gibbs states.
This demonstrates an interesting separation between learning a Hamiltonian from its dynamics, as in time evolution, and learning from its steady-states, as in the Gibbs state.

What allows us to get a FPT algorithm is a Goldreich--Levin-like version of the classical shadows protocol of Huang, Kueng, and Preskill~\cite{hkp20}.
For an $\qubits$-qubit state, classical shadows allows us to estimate all its $\locality$-local Pauli coefficients in time $\qubits^{\locality}$.
However, consider the ``dual'' access model where we have an unknown observable $O$ and we can efficiently estimate $\tr(O \rho)$ for input states $\rho$.
The classical shadows formalism also works in this setting~\cite{hcp23}.
We give a subroutine with improved efficiency, able to estimate all of the nontrivial $\locality$-local Pauli coefficients of $O$ in time $\bigOt{\qubits \cdot f(\locality)}$, which is FPT.
So, this observable-centered perspective gives more efficient algorithms in our setting of learning quantum dynamics.
See \cref{prop:gl-general-intro} and \cref{lem:gl-general} for more details.

As a corollary of our main theorem, we get algorithms for learning Hamiltonians that exhibit power law decay which scale better than the standard limit.
\begin{corollary}[Informal version of \cref{cor:power-law}]
    Let $H = \sum_{a=1}^\terms \lambda_a E_a$ be a $\locality$-local Hamiltonian on a $d$-dimensional lattice with $\alpha$-power law decay for $\alpha > d$ (\cref{def:power-law-Ham}).
    Let
    \[
        \kappa = \frac{d\locality}{d\locality + (\alpha - d)} \,.
    \]
    Then we can find some $\hat{\lambda}$ such that $\infnorm{\hat{\lambda} - \lambda} < \eps$ with probability $\geq 1-\delta$ using $\tet = \bigO{\frac{1}{\eps^{1 + \kappa}}\log\frac{\qubits}{\delta}}$ total time evolution.
\end{corollary}

\begin{remark} \label{rmk:power-law}
The best prior algorithm for learning power law Hamiltonians is due to Stilck França, Markovich, Dobrovitski, Werner, and Borregaard~\cite{smdwb24}.
It has a total time evolution of $1/\eps^{2 + \kappa'}$, where $\kappa'$ is a parameter which goes to zero as the decay rate $\alpha$ increases.\footnote{
    Their algorithm also learns Lindbladians, a more general class than Hamiltonians.
    Lindbladian learning requires $1/\eps^2$ evolution time, since tasks like estimating transition matrices from (continuous-time) queries reduce to it.
}
On the other hand, our algorithm always performs better than $1/\eps^2$, and tends to Heisenberg-limited scaling as the decay rate increases.
Further, their algorithm only works for $\alpha$ larger than a multiple of $d$; ours holds up to $\alpha = d$.
This is a natural barrier, as beyond it, the strength of interactions can diverge with distance and quantities like energy are no longer extensive~\cite{ddmprt23}.
\end{remark}

\subsection{Technical overview}

We now illustrate the key technical components of our algorithm.
We consider a Hamiltonian $H(\lambda) = \sum_{a=1}^\terms \lambda_a E_a$, where each $E_a$ has support size at most $\locality = \bigO{1}$.
For notational convenience, we pad the Hamiltonian such that there is a term $E_a$ for every local Pauli in $\locals_\locality$; with this change, the set of terms is now known (since it is the set of all possible interaction terms), and what is unknown is which terms have non-zero weight.
A motivating example to keep in mind is when $E_a$ is geometrically local with an unknown geometry: the terms $E_a$ are local with respect to some constant-dimensional lattice, but the position of the qubits on the lattice is unknown.
In this setting, both the local norm $\lonorm{H}$ and effective sparsity $\sparse$ are constant.
With this example, we focus on showing \emph{structure learning} (\cref{question:central-open-q}) and \emph{constant time resolution} (\cref{question:time-resolution}).
Our algorithm will naturally extend to Hamiltonians with long-range interactions (\cref{question:learning-without-locality}) with some additional care.

In this section, we focus on three main contributions.
We begin by describing a bootstrapping framework for achieving Heisenberg scaling which was introduced by Dutkiewicz, O'Brien, and Schuster~\cite{dos23}.
This will form the ``outer loop'' of our algorithm.
This prior work's instantiation of this framework ultimately does not have guarantees beyond those of \cite{htfs23}, only working with an alternative access model and still needing knowledge of the interaction terms.
We make the observation that, unlike \cite{htfs23}, this framework can be performed without locality knowledge, and so is well-suited for structure learning.
Through this framework, Heisenberg-limited structure learning reduces to a problem of estimate improvement: given a coefficient estimate $\lambda^{(0)}$ such that $\infnorm{\lambda - \lambda^{(0)}} < \eta$, find some $\lambda^{(1)}$ such that $\infnorm{\lambda - \lambda^{(1)}} < \eta/2$ using evolution time at most $1/\eta$.

If we could evolve with respect to $e^{-\ii (H - H_0)t}$ for a chosen $H_0$ instead of $e^{-\ii Ht}$, then existing results about computing expectations of local observables can be used to perform estimate improvement with the desired guarantees.
However, we are only given $e^{-\ii Ht}$, so implementing $e^{-\ii (H - H_0)t}$ requires Trotterization, alternating between time-evolutions of $H$ and $-H_0$.
Naively, this only approximates the evolution when the alternating occurs at time intervals of size $1/\qubits$, giving a time resolution with system size dependence.
Our second contribution is to show that our algorithm works when alternating happens only at constant-sized intervals, by a novel bound on a certain form of Trotter error.

We show that this suffices to obtain an algorithm for structure learning with evolution time $1/\eps$ and constant time resolution, but the classical overhead is $\qubits^{\locality}$, scaling with the locality of the unknown Hamiltonian.
Our third contribution is to give an algorithm to locate these terms in $\qubits^2$ time by implementing Goldreich--Levin-like queries on the Pauli spectrum of a Hamiltonian to efficiently identify its large coefficients.

\paragraph{A recursive framework for Hamiltonian learning with Heisenberg scaling.}
First, we describe the bootstrapping algorithm of \cite{dos23}.
The algorithm performs recursion on the residual: $\lambda$ and $\lambda^{(j)}$ are $2^{-j}$-close, and to get an improved estimate, we estimate the new Hamiltonian $H(\lambda) - H(\lambda^{(j)})$ to $2^{-j-1}$ error.
It may not be immediately clear why this recursion will be more efficient than learning $H(\lambda)$ to $\eps$ error in a single shot.
The underlying principle is that one wants to transform the task of estimating a parameter of a unitary $Z$ to $\eps$ error, which naively requires $1/\eps^2$ applications of $Z$, to the task of estimating a parameter of $Z^{1/\eps}$ to constant error, which only requires a constant number of applications of $Z^{1/\eps}$, or $1/\eps$ applications of $Z$ total.
This idea is standard in quantum metrology, including in the work on robust phase estimation~\cite{Kimmel2015}, and for more complicated types of parameters, it becomes important to use an initial estimate in the amplification subroutine to generate an improved estimate~\cite{hkot23}.
This approach is reminiscent of gradient descent, where in each iteration, we apply a linear update to the vector of our current estimate for $\lambda$.

\begin{algorithm}[Bootstrapping a Hamiltonian learning algorithm to Heisenberg scaling]\mbox{} \label{algo:intro-strap}
\begin{algorithmic}[1]
    \State Let $\lambda^{(0)} = (0,\dots,0)$;
    \For{$j$ from $0$ to $T = \floor{\log_2(1/\eps)}$}
        \Comment{We maintain that $\infnorm{\lambda - \lambda^{(j)}} \leq 2^{-j}$}
        \State Learn the Hamiltonian $H(\lambda) - H(\lambda^{(j)})$ to error $1/2^{j+1}$;
        \Comment{Recall $H(x) = \sum_{a=1}^\terms x_a E_a$}
        \State Let $\wh{\lambda}^{(j)}$ be the estimated coefficients;
        \State Let $\lambda^{(j+1)} \gets \lambda^{(j)} + \wh{\lambda}^{(j)}$;
    \EndFor
    \State \Output $\lambda^{(T+1)}$.
\end{algorithmic}
\end{algorithm}

In summary, to learn $H = H(\lambda)$ to $\eps$ error, it suffices to be able to improve an $\eta$-good estimate $H_0 = H(\lambda^{(0)})$ to an $\eta/2$-good one.
To achieve Heisenberg scaling, we want to perform this subroutine with evolution time $1/\eta$.
We make the key observation that this framework does not reduce to problems on individual terms, nor does it use any locality information about $H$: the reduction is global and agnostic to the structure of $H$.
Prior work~\cite{dos23} uses this framework with estimation improvement being done by the algorithm of \cite{hkt24} as a black box.
Since the algorithm in the subroutine requires knowledge of the terms, their full algorithm also does, and cannot perform structure learning; see \cref{sec:related} for more details.
Our goal is now to perform estimation improvement without knowledge of the terms.

\paragraph{Improving an estimate with continuous quantum control.}
We first make a simplifying assumption: suppose, for now, that we could perform the evolution $e^{-\ii \wh{H}t}$, where $\wh{H} = (H - H_0)/\eta$.
Then our goal of improving an estimate amounts to estimating the coefficients of $\wh{H}$ to $1/2$ error.
The simple algorithm of ``time derivative estimation'' does the trick here: to learn a term $\wh{\lambda}_a$ of $\wh{H}$, take $P_a$ to be a single-qubit Pauli such that $[E_a, P_a]$ is not zero, and let $Q_a = \tfrac{\ii}{2} [E_a, P_a] = \ii E_a P_a$. 
Then initialize a state with density matrix $\frac{\id + Q_a}{\dims}$ where $\dims = 2^\qubits$, apply $e^{-\ii \wh{H} t}$, and then measure the qubit associated to $P_a$ in the eigenbasis of $P_a$ to produce an unbiased estimator of $\tr\parens{P_a e^{-\ii \wh{H} t}\frac{\id + Q_a}{\dims} e^{\ii \wh{H}t}}$. 

Next, we argue that $\tr\parens{P_a e^{-\ii \wh{H} t}\frac{\id + Q_a}{\dims} e^{\ii \wh{H}t}}$ itself is a reasonable estimate of the $a$-th coefficient of $\wh{H}$. To see this, we observe that for small $t$, $e^{-\ii \wh{H}t} P_a e^{\ii \wh{H}t}$ is well-approximated by its first-order behavior (\cref{lem:first-order-approx}):
\begin{equation} \label{eq:tech-fo}
    e^{\ii \wh{H}t} P_a e^{-\ii \wh{H}t}
    = P_a + [\ii \wh{H}t, P_a] + \Delta,
    \text{ where } \fnorm{\Delta} \leq \frac{t^2}{2}\fnorm{[\wh{H}, [\wh{H}, P_a]]} = \bigO{t^2}
\end{equation}
This is a consequence of the Hadamard formula.
Using this approximation, we have
\begin{align}
    \tr\parens[\Big]{P_a e^{-\ii \wh{H} t}\tfrac{\id + Q_a}{\dims} e^{\ii \wh{H}t}}
    &= \frac{1}{\dims}\tr\parens{e^{\ii \wh{H} t}P_a e^{-\ii \wh{H} t} Q_a} \nonumber \\
    &= \frac{1}{\dims}\tr\parens{(P_a + [\ii \wh{H} t, P_a])Q_a} + \bigO{t^2} \nonumber \\
    &= \frac{\ii t}{\dims}\sum_{b=1}^\terms \wh{\lambda}_b \tr\parens{[E_b, P_a]Q_a} + \bigO{t^2} \nonumber\\
    &= 2\wh{\lambda}_a t + \bigO{t^2}.
    \label{eq:tech-pauli-query}
\end{align}
Here, we used that $\tr\parens{[E_a, P_a]Q_a} = \tr((2E_a P_a)(\ii E_a P_a)) = -2\ii\dims$ and $\tr\parens{[E_b, P_a] Q_a} = 0$ for $b \neq a$.
Taking $t$ to be a small enough constant, and with a large constant number of copies of the unbiased estimator, we can extract an estimate of $\wh{\lambda}_a$ to $\frac{1}{2}$ error, as desired.
This algorithm as stated holds for one particular term, but can be parallelized to estimate all coefficients with $\bigO{\log(\qubits)}$ applications of $e^{\ii \wh{H} t}$.
In fact, $\bigO{\log(\qubits)}$ applications suffice to get estimates to $\tr(P e^{-\ii \wh{H} t} Q e^{\ii \wh{H}t})$ for all $P$ and $Q$ with constant support size (\cref{lem:shadows}).
This follows because, as observed by \cite{hkt24}, the same experiment of preparing a density matrix, applying $e^{-\ii \wh{H} t}$, and measuring, can be used to estimate $\bigO{\qubits^2}$ many $P, Q$ pairs.
This is, in some sense, a ``process'' version of the classical shadows protocol of Huang, Kueng, and Preskill~\cite{hkp20}.

\paragraph{Implementing the algorithm without continuous quantum control.}
We have just described a Hamiltonian learning algorithm with Heisenberg-limited scaling, assuming that we could apply $e^{-\ii (H - H_0)t}$ for our unknown $H$ and a known $H_0$.
Now, we show how to modify this algorithm to work when we are only given the ability to perform $e^{-\ii Ht}$.

A naive first attempt is to replace the continuous control with a \textit{Trotter} approximation, where we alternate between applying $e^{-\ii Ht}$ and $e^{\ii H_0 t}$:
\begin{align*}
    e^{-\ii \wh{H}t} = e^{-\ii (H - H_0)t/\eta} = (e^{-\ii Ht/(\eta s)} e^{\ii H_0t/(\eta s)})^{s} + \text{error},
\end{align*}
where the error term goes to zero as $s$ goes to infinity.
When the error is small enough, this approach yields essentially the same total evolution time as the continuous control algorithm.
A large body of work is devoted to understanding the error term~\cite{CSTWZ19,blanes16,suzuki85,bacs06}, but in summary, to make its operator norm small here requires $s$ to depend on the system size, $\qubits$, so the time resolution scales inversely with system size.

However, recall that our algorithm only used $e^{\ii\wh{H}t} P_a  e^{-\ii\wh{H}t}$ up to first order, as shown in \eqref{eq:tech-fo}.
So, a much weaker bound on the error above suffices.
Formally, let $Z$ be our Trotter approximation to $e^{-\ii\wh{H}t}$.
Then $Z$ only needs to satisfy that, for a Pauli $P$ supported on one qubit,
\begin{align} \label{eq:intro-key}
    Z^\dagger P Z
    = P + [-\ii \wh{H}t, P] + \Delta,
    \text{ where } \fnorm{\Delta} \leq ct\fnorm{P}
\end{align}
for a sufficiently small constant $c$.
Notice that the error bound of $t^2\fnorm{P}$ in \eqref{eq:tech-fo} is here replaced with $ct\fnorm{P}$, which, importantly, is weaker.
Our key lemma is that, for local Hamiltonians $H$ and $H_0$, there is such a discrete approximation to the time evolution $e^{-\ii \wh{H} t}$ with \emph{constant} time resolution.

\begin{lemma}[Informal version of \cref{lem:main}] \label{lem:key-trotter-intro}
Let $H,H_0$ be geometrically local Hamiltonians and set 
\begin{align} \label{eq:tech-z}
    Z = (e^{-\ii Hc} e^{\ii H_0c})^{t/(\eta c)} \,.
\end{align}
Assume that $t < 1$, $\infnorm{\lambda - \lambda^{(0)}} \leq c'\eta$ where $c,c'$ are sufficiently small constants.  Then for any Pauli $P$ supported on a single qubit, \eqref{eq:intro-key} holds, where $\wh{H} = (H - H_0)/\eta$.
\end{lemma}

Next, we discuss this Trotter approximation in more detail.

\paragraph{Understanding constant-time Trotterization.}
\cref{lem:key-trotter-intro} is quite subtle and we focus on proving \eqref{eq:intro-key} with just one alternating layer.
The goal is to bound
\begin{equation}
    \Delta' = Z^\dagger P Z - (P + [-\ii (H - H_0)c, P]),
    \text{ where } Z = e^{-\ii Hc} e^{\ii H_0c}.
\end{equation}
Since we want an ultimate bound of $\fnorm{\Delta} = \bigO{ct \fnorm{P}}$ when we have $t/\eta c$ alternating layers, for just one layer we aim for a bound of $\fnorm{\Delta'} = \bigO{c^2\eta \fnorm{P}}$.
Since $P + [-\ii (H - H_0)c, P]$ is a first order approximation of $e^{-\ii(H - H_0)c} P e^{\ii (H-H_0)c} + \bigO{c^2\eta^2}$, i.e.  
$$P + [-\ii (H - H_0)c, P] = e^{-\ii(H - H_0)c} P e^{\ii (H-H_0)c} + \bigO{c^2\eta^2},$$ the error $\Delta'$ can be interpreted as a Baker--Campbell--Hausdorff bound, relating evolution with respect to $H - H_0$ to the alternating evolution of $H$ and $-H_0$.
Bounding $\Delta'$ amounts to understanding nested commutators through the Hadamard formula (\cref{fact:hadamard}), $e^{X} P e^{-X} = P + [X, P] + \frac12 [X, [X, P]] + \cdots$.
Expanding the exponentials in this way, we can confirm that the zeroth and first order terms of $\Delta'$ cancel, and the second order terms are
\begin{align*}
    &\frac12[\ii cH, [\ii cH, P]] + [-\ii cH_0, [\ii cH, P]] + \frac12[-\ii cH_0, [-\ii cH_0, P]] \\
    &\quad = \frac12[\ii c(H - H_0), [\ii c(H - H_0), P]] + \frac{1}{2}[[-\ii cH_0, \ii cH], P] \\
    &\quad = \frac12\underbrace{[\ii c(H - H_0), [\ii c(H - H_0), P]]}_{\bigO{\eta^2c^2}} + \frac{1}{2}\underbrace{[[-\ii c(H - H_0), \ii cH], P]}_{\bigO{\eta c^2}} .
\end{align*}
Analogous arguments show that the higher-order terms also scale linearly in $\eta$.
The full proof of \cref{lem:main} requires additional care to treat the layers of commutators: simply iterating the one-layer bound does not suffice, since one also needs to show that the error terms, i.e.\ the higher-order nested commutators, are ``locally small'' in an appropriate sense.
Technical work also goes into proving the lemma beyond the finite-range setting, to the class of Hamiltonians $H$ with bounded $\lonorm{\cdot}$ norm.
In particular, we prove an $L^2$ bound on commutator expansions which are typically bounded in $L^1$, so that the lemma holds in greater generality.

Existing Trotter error bounds control $e^{X} e^{Y} - e^{X + Y}$ in terms of commutators like $[X, Y]$~\cite{somma16,CSTWZ19}, as well as control $e^{X} P e^{-X}$ in terms of commutators like $[X, P]$~\cite[Theorem 10]{CSTWZ19}.
This lemma can be seen as a combination of the two, bounding $e^X e^Y P e^{-Y} e^{-X} - e^{X + Y}Pe^{-(X+Y)}$ in terms of commutators of the form $[[X, Y], P]$.
We believe it may be of independent interest, as it allows us to achieve constant time resolution.

\paragraph{Performing Hamiltonian learning with weaker locality guarantees.}
In summary, we can take the continuous quantum control algorithm and apply the Trotterization lemma to replace the time evolution by $H - H_0$ with alternating constant-time evolutions between the unknown Hamiltonian $H$ and a known Hamiltonian $H_0$.
This suffices to obtain our main theorem, albeit with an inefficient running time of $\qubits^{\locality}$.

\begin{algorithm}[Improving an $\eta$-good estimate $H_0 = H(\lambda^{(0)})$, informal version of Lines 8 through 15 of \cref{algo:strap}]\mbox{} \label{algo:intro-boot}
\begin{algorithmic}[1]
    \State Let $c$ be a sufficiently small constant;
    \State Perform the following circuit $\bigO{\log(\qubits/\delta)}$ times and record the measurements:
    
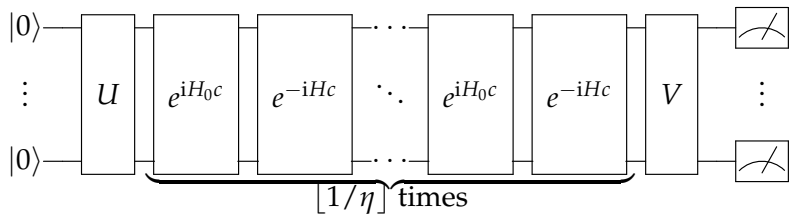
\[
        \Qcircuit @R=.6em @C=0.65em {
        \push{\ket{0}} & \qw & \multigate{3}{U} & \multigate{3}{e^{\ii H_0 c}} & \multigate{3}{e^{-\ii H c}} & \qw & {\cdots} & & \multigate{3}{e^{\ii H_0 c}} & \multigate{3}{e^{-\ii H c}}& \multigate{3}{V} &  \qw & \meter \\
        \push{\vdots} &&&&&& \ddots &&&&&& {\vdots}\\
        &&&&&&&&&&& \\
        \push{\ket{0}} & \qw & \ghost{U} & \ghost{e^{\ii H_0 t}} & \ghost{e^{-\ii H t}} & \qw &{\cdots} & & \ghost{e^{\ii H_0 t}} & \ghost{e^{-\ii Ht}}& \ghost{V} & \qw & \meter
        \gategroup{1}{4}{4}{10}{.5em}{_\}} \\
        &&&&&& \floor{1/\eta} \text{ times}
        }
    \]
    where $U$ and $V$ are random layers of single-qubit Clifford gates;
    \For{$a$ from $1$ to $\terms$}
        \State Let $P_a \in \locals_1$ be some 1-local Pauli that anticommutes with $E_a$;
        \State Let $Q_a \leftarrow \frac{\ii}{2}[P_a,E_a] = \ii P_a E_a$;
        \Comment{This choice of $P_a$ and $Q_a$ comes from the computation in \eqref{eq:tech-pauli-query}}
    \EndFor
    \State Use the circuit measurements to produce estimates $\wh{\lambda} \in [-1,1]^\terms$ to the expectation values $\mu \in [-1,1]^\terms$, where
    \[
        \mu_a = \frac{1}{2\dims}\tr\parens*{P_a \parens[\big]{e^{-\ii H c} e^{\ii H_0 c}}^{\floor{1/\eta}} Q_a \parens[\big]{e^{-\ii H_0 c} e^{\ii H c}}^{\floor{1/\eta}}};
    \]
    \Comment{The estimates $\wh{\lambda}$ will satisfy $\infnorm{\wh{\lambda} - \mu} < c/20$ with probability $\geq 1-\delta$ (\cref{lem:shadows})}
    \State Set $\lambda^{(1)} \gets \lambda^{(0)} + \frac{1}{c\floor{1/\eta}} \wh{\lambda}$;
    \State Set all coefficients of $\lambda^{(1)}$ smaller than $\eta/4$ to $0$;
    \State \Output $\lambda^{(1)}$;
\end{algorithmic}
\end{algorithm}

Throughout, we did not require knowledge of the structure of the Hamiltonian: the estimate improvement protocol proceeds by running quantum circuits with the alternating evolution to create a shadows dataset, and then using the dataset to estimate every coefficient.
This algorithm has one minor snag: the Hamiltonian of improved estimates, $H(\lambda^{(1)})$, may be much larger than $H(\lambda)$, since it can have many non-zero coefficients: our only guarantee is that every coefficient is $\eta/10$-close to the true value.
To apply \cref{lem:main} in the next iteration of estimation improvement, the estimate must be well-behaved, i.e.\ $\lonorm{H^{(1)}}$ must be bounded.
So, to maintain this, we round entries smaller than $\eta/4$ in $\lambda^{(1)}$ to zero, and then proceed.
We summarize the algorithm in \cref{algo:intro-boot}.
Repeating this estimation improvement protocol $\log_2(1/\eps)$ times and keeping track of the relevant parameters gives \cref{thm:main-informal}, except with a worse classical overhead of $\bigThetat{\terms} = \bigThetat{\qubits^\locality}$.

Next, we discuss how to make this algorithm more time-efficient.

\paragraph{Time-efficient structure learning.} 
Our goal is now to make our structure learning algorithm time-efficient, with a scaling better than $\qubits^\locality$.
Our algorithm currently proceeds by considering a Trotterized Hamiltonian evolution $Z$ \eqref{eq:tech-z}, which we apply in circuits to get estimates of the quantity $\frac{1}{\dims}\tr(P Z Q Z^\dagger)$, where $P$ and $Q$ are Paulis with small support (\cref{algo:intro-boot}).
Viewing $O = Z^\dagger P Z$ as an observable in $\C^{N \times N}$ (note $\norm{O} \leq 1$ and $\tr(O) = 0$), this estimator is of $\frac{1}{\dims}\tr(OQ)$.
This is a coefficient $c_Q$ of $O$ in its Pauli decomposition, $O = \sum_{R \in \locals} c_R R$ where $c_R = \frac{1}{\dims}\tr(OR)$.
Further, our key linear approximation \eqref{eq:intro-key} shows that $O \approx P + [-\ii \wh{H} t, P]$, so the $c_Q$'s are coefficients of $\wh{H}$ up to scaling and error.
This follows from the computation in \eqref{eq:tech-pauli-query}.
So, to estimate the coefficients of $\wh{H}$, it suffices to estimate the $c_Q$'s.
This recovers the coefficients of terms that do not commute with $P$, and ranging over all choices of single-qubit $P$, we can recover all coefficients of $\wh{H}$, thereby giving our improved estimate of $H$.

Upon establishing this connection to the Pauli decomposition of $O$, we can now abstract out the underlying Hamiltonian and consider structure learning for an arbitrary observable $O$.
The problem is as follows: suppose we are given an observable $O = \sum_{Q \in \locals} c_Q Q$, which we can access by preparing a state $\rho$ and measuring the POVM\footnote{
    For a pair of positive semi-definite matrices $O_{+1}, O_{-1}$ such that $O_{+1} + O_{-1} = \id$, its associated POVM is a quantum measurement that takes a quantum state with density matrix $\rho$ and outputs ``$+1$'' or ``$-1$'' with probability $\tr(O_{+1} \rho)$ and $\tr(O_{-1} \rho)$, respectively~\cite[Section~2.2.6]{NielsenChuang00}.
} associated with $\{\frac{\id + O}{2}, \frac{\id - O}{2}\}$, giving a $\{\pm 1\}$-valued random variable with expectation $\tr(O \rho)$.
Then the goal is to get an estimate of every $c_Q$, under the assumption that $O$ is ``low-degree'', in the sense that $c_Q = 0$ when $Q$ has support size greater than $\locality$, and ``sparse'', in the sense that all but $\bigO{1}$ of the $c_Q$'s are zero.\footnote{
    The problem as described assumes that the approximation $O \approx P + [-\ii \wh{H} t, P]$ is exact and that the Hamiltonian $\wh{H}$ is geometrically local, or more generally, finite-range on a bounded-degree graph.
}

We previously described how to query Pauli coefficients of $O$ with this access model: for the state $\rho = \frac{\id + Q}{\dims}$, $\tr(O\rho) = c_Q$, so by preparing multiple copies of this state we can get an estimate of $c_Q$.
So, to estimate all Pauli coefficients, we can query every local $Q$, costing $\bigO{\qubits^{\locality}}$ time.
This is necessary for this strategy since we are not told which $c_Q$'s are non-zero.

However, we can actually estimate the nontrivial interactions more efficiently, in $\bigO{ 2^{\locality} n^2}$ time.
This is fixed-parameter tractable.
To obtain a FPT runtime for structure learning, we develop quantum versions of classical machinery from boolean function learning.
Classically, given samples of a degree-$d$ function $f: \{-1, 1\}^n \to [-1,1]$ under the uniform distribution, we can estimate all of its Fourier coefficients in time $\bigO{n^d}$.
However, given query access to $f$, the Goldreich--Levin algorithm~\cite{gl89} can learn all of the Fourier coefficients larger than a constant in time $\bigOt{n}$, which suffices to estimate $f$ when it is Fourier-sparse.
Analogously, in the quantum setting, while naively it takes $\bigOt{\qubits^{\locality}}$ time to estimate all $\locality$-local Pauli coefficients, our ability to choose $\rho$ gives us a sort of query access, so we can estimate its non-trivial coefficients more efficiently.

We can see this with a simple example.
Suppose $O$ only has weight on Pauli terms made up of $\id$ and $\sigma_z$, so that we can write $O = \sum_{Q \in \{\id, \sigma_z\}^{\otimes \qubits}} c_Q Q$.
The observable $O$ then can be treated as a boolean function: consider querying for $\tr(O \rho)$ where $\rho$ is a computational basis vector, $\ket{b_1}\dots\ket{b_\qubits}$ for $b_i \in \{0,1\}$.
Then\footnote{
    This computation is also true without the assumption on $O$, since $\bra{b}\sigma_x\ket{b} = \bra{b}\sigma_y\ket{b} = 0$ for $b \in \{0,1\}$.
}
\begin{align*}
    \tr(O \ket{b}\bra{b})
    &= \sum_{Q \in \locals} c_Q \bra{b} Q \ket{b} \\
    &= \sum_{Q \in \{\id,\,\sigma_Z\}^{\otimes \qubits}} c_Q \prod_{i \in \supp(Q)} (-1)^{b_i},
\end{align*}
so $\{c_Q\}_Q$ are the Fourier coefficients of the boolean function $b \mapsto \tr(O\ket{b}\bra{b})$.
In this setting, we can directly run the Goldreich--Levin algorithm, with minor modifications to deal with the form of access we are given.
In particular, we can simulate its key subroutine, where instead of estimating a single coefficient $c_Q$, we estimate the weight of coefficients $\abs{c_X}^2$ over all $X$ ``containing'' $Q$.
We refer to this as a GL query.
Given products of Paulis $X,Q \in \locals$, we write $Q \subseteq X$ if $X$ matches $Q$ on its support and is an arbitrary product of Paulis outside the support of $Q$ (\cref{def:pauli-order}).
Then, following existing analyses~\cite[Proposition 3.40]{odonnell14}, for a Pauli matrix $Q$ which is $\sigma_z$ on the set of qubits $S \subset [\qubits]$ and $\id$ otherwise, we can write the weight above $Q$ as a second moment expression,
\begin{align} \label{eq:gl-expec}
    \expec[\Big]{b_{\overline{S}} \sim \{0,1\}^{\abs{\overline{S}}}}{\abs[\Big]{\underbrace{\expec[\Big]{b_S \sim \{0,1\}^{\abs{S}}}{\tr(O\ket{b}\bra{b})\prod_{i \in S}(-1)^{b_i}}}_{\text{Fourier coefficient of $\tr(O \ket{b}\bra{b})$ restricted to $S$}}}^2}
    = \sum_{X \supseteq Q} c_X^2,
\end{align}
where $b_S$ and $b_{\overline{S}}$ are both being sampled according to the uniform distribution.
This expression can be estimated using our access to $\tr(O \ket{b}\bra{b})$, allowing us to perform a GL query, and thereby run Goldreich--Levin.
A straightforward generalization gives an analogous algorithm without the assumption that $O$ is boolean.

Since Goldreich--Levin requires $\bigOmega{\qubits}$ GL queries, this naively translates to a linear number of queries to $O$, and thus a total evolution time linear in $\qubits$.
This dependence come from running a separate experiment each time we answer a GL query.
We want an evolution time logarithmic in $\qubits$.
To get this improvement, we give an algorithm to answer all small-support GL queries with only a logarithmic number of queries to $O$.
We show the following; see \cref{algo:structure-query} for more details.

\begin{lemma}[Informal version of \cref{lem:gl-general}] \label{prop:gl-general-intro}
    Let $O \in \mathbb{C}^{\dims \times \dims}$ be an unknown observable with $\norm{O} \leq 1$ and a Pauli decomposition of $O = \sum_{Q \in \locals} c_Q Q$.
    Suppose we can efficiently apply the POVM $\{\frac{\id + O}{2}, \frac{\id - O}{2}\}$, and suppose we are given a natural number parameter $\locality = \bigO{1}$.
    Then, with $\bigO{\log(\qubits)}$ queries to $O$ and $\bigO{\qubits \log(\qubits) }$ additional pre-processing time, we can output a data structure.
    This data structure, with probability $0.99$, can correctly respond to the following type of query: given $X \in \locals_{\locality}$, output an estimate of 
    \begin{align*}
        \sum_{\substack{Q \in \locals \\ Q \supseteq X}} \frac{\abs{c_Q}^2}{6^{\abs{\supp(Q)}}}
    \end{align*}
    to $0.01$ error.
    Answering the query takes $\bigO{\log(\qubits)}$ time on a classical computer.
\end{lemma}

This algorithm parallelizes the subroutine for answering GL queries, allowing the estimation of weights on multiple subsets simultaneously via a carefully chosen set of measurements with coupled randomness.
Returning to our simplified setting, to estimate the expression in \eqref{eq:gl-expec}, we need to perform nested sampling: sample several copies of $b_{\overline{S}}$, and then for every copy, sample several completions $b_S$.
Our algorithm produces a ``dataset'' of $b$'s such that, for every $S$, we can find a subset of the data which takes the above form.
In essence, we can do this by guessing a choice of $S$, and then when we are later given a GL query, restricting to the part of the dataset where we guessed correctly.
The probability of guessing correctly is exponentially small in the size of the support of the queried Pauli, so the estimate is good when this is small.
Though there is no requirement that $O$ is low-degree in \cref{prop:gl-general-intro}, this is implicitly enforced by the $6^{|\supp(Q)|}$ in the denominator.

With this, we can now explain our algorithm for efficient structure learning. 
To estimate the interaction coefficients of $\wh{H}$, we estimate the Pauli coefficients of the observable $O = Z^\dagger P Z \approx P + [-\ii \wh{H}t, P]$.
Our guarantees on the Hamiltonian imply that $O$ is degree $\locality$ and has $\bigO{1}$ non-zero coefficients.
We use GL queries to iteratively learn the non-zero interaction terms.
Specifically, we first query the Paulis of support size $1$ to learn which have non-zero weight (its weight meaning the mass of coefficients corresponding to Paulis which contain it).
Then, for each size $1$ Pauli, we query all possible size $2$ Paulis which contain it to estimate its weight.
Whenever the weight of a Pauli is smaller than some constant, we can delete it.
This way, we can ensure there are only ever $\bigO{1}$ sets remaining at each level and thus we get an algorithm runs in time $\bigOt{\qubits}$.

Finally, there is one additional modification we need to make to the algorithm.  Recall that we need to let $P$ range over all possible $P \in \locals_1$ to estimate all of the coefficients of $\wh{H}$.  In our full algorithm we need to parallelize this last step to only use $\bigO{\log \qubits}$ experiments.  Compared to the GL query algorithm, this parallelization step is straightforward, following from an analysis like those used for classical shadows.

\subsection{Related work} \label{sec:related}

\begin{figure}
    \begin{center}\begin{tabular}{r|l l l c}
        & $\tet$ & $\tres$ & $\# \textup{exp}$ & SL? \\
        \hline
        derivative estimation~\cite{zylb21} & $1/\eps^3$ & $\eps$ & $1/\eps^4$ & {\footnotesize\Checkmark} \\
        derivative estimation~\cite{caro24} & $\qubits^4/\eps^4$ & $1/\qubits$ & $1/\eps^4$ & {\footnotesize\Checkmark} \\
        unitary transformation~\cite{oksm24} & $\qubits^\locality/\eps$ & $\eps/\qubits^2$ & $\qubits^\locality\log^2(1/\eps)$ & {\footnotesize\Checkmark} \\
        cluster expansion~\cite{hkt24} & $1/\eps^2$ & $1$ & $1/\eps^2$ \\
        reshaping~\cite{htfs23} & $1/\eps$ & $\sqrt{\eps}$ & $\polylog(1/\eps)$ \\
        \textbf{Our work} & $1/\eps$ & $1$ & $\log(1/\eps)$ & {\footnotesize\Checkmark}
    \end{tabular}\end{center}
    \caption{
        A comparison of the quantum resources required by Hamiltonian learning algorithms: its total evolution time, time resolution, number of experiments used, and whether it can perform structure learning, respectively.
        For this comparison, we suppose we have a constant-local $\qubits$-qubit Hamiltonian on a constant-dimensional lattice, and we ignore $\log(\qubits)$ and $\log\log$ dependences.
    } \label{fig:prior}
\end{figure}

\paragraph{Hamiltonian learning.}
We now discuss the literature on Hamiltonian learning from real-time evolution.
This work has roots in the physics literature on quantum-enhanced sensing and quantum metrology, where the primary goal is to devise experimentally feasible Heisenberg-limited protocols for specific, simple classes of Hamiltonians~\cite{glm04,ramsey50,bb05,lkd02}.
We will focus on algorithms for general, many-body local Hamiltonians, paying attention to their total evolution time, $\tet$, and time resolution, $\tres$.
Such algorithms generally assume a greater degree of quantum control.

Early protocols for this task~\cite{smlkr11,slp11,hbcp15,BAL19,zylb21} use time derivative estimation, which does not require locality knowledge but has poor dependence on error, both in time resolution and evolution time.
Many of these works lack fully rigorous analyses, but rigorous versions are straightforward to prove with the randomized measurement analyses in e.g.\ \cite{hkt24}: \cref{lem:first-order-approx} with $t \gets \bigO{\eps}$ and \cref{lem:shadows} with $\eps \gets \eps^2$ together imply that learning the coefficients of $H$ to $\eps$ error can be done with $\bigO{\log(\qubits)/\eps^4}$ evolutions of $e^{-\ii H\eps}$.
This approach also gives structure learning, since the output of the algorithm can be used to estimate $\frac{1}{\dims}\tr(H X)$ for every $X \in \locals_{\locality}$ with $\bigO{\abs{\locals_\locality}}$ classical overhead.
With appropriate adjustments, these works can also handle non-local Hamiltonians.
These approaches can perform learning in general settings, but do not achieve optimal performance in evolution time, time resolution, and classical overhead.

There are recent works which give algorithms for these general forms of learning, but which still fall short in these figures of merit.
Caro~\cite{caro24} uses polynomial interpolation to change the length of the time evolutions used from $\eps$ to $1/\qubits$, which is an improvement for small $\eps$.
Odake, Kristjánsson, Soeda, and Murao~\cite{oksm24} give an algorithm for Hamiltonian learning based on quantum transformation which, notably, attains Heisenberg scaling with only very mild assumptions on the Hamiltonian.
However, their result is best-suited for learning only one parameter of the Hamiltonian, and they must pay a factor of $n^\locality$ for structure learning to run their circuit for every parameter to be estimated.
Others have investigated heuristic approaches to structure learning~\cite{fgw+22,gfk+21}.

A subsequent line of work gives improvements when the structure is known, with the work of Haah, Kothari, and Tang~\cite{hkt24} achieving $\tet = \bigO{1/\eps^2}$ and $\tres = \bigOmega{1}$, and the work of Huang, Tong, Fang, and Su~\cite{htfs23} achieving Heisenberg scaling, $\tet = \bigO{1/\eps}$, and $\tres = \bigOmega{\sqrt{\eps}}$.
These algorithms do not support structure learning.

Algorithms with Heisenberg scaling have also been attained for bosonic~\cite{ltngy23} and fermionic~\cite{mh24,nly23} Hamiltonians, all following the strategy of \cite{htfs23} of reshaping the Hamiltonian to decouple sites of the system, and incurring a polynomial $\eps$ dependence in the time resolution.
There is also work for learning of Lindbladians, the generalization of Hamiltonians to open quantum systems~\cite{smdwb24}.
In this work, the authors use robust polynomial interpolation to give an improvement to time derivative estimation; this still requires knowledge of the interaction structure, though they are able to use Lieb--Robinson bounds to give algorithms for Hamiltonians with power law decay.
\cref{rmk:power-law} gives a more detailed comparison to our work.

There are Hamiltonian learning algorithms which have total evolution time $\bigO{\log(\qubits)/\eps}$ without using reshaping, but these hold under alternative access models~\cite{ku18,dos23}.
Dutkiewicz, O'Brien, Schuster~\cite{dos23} give a simple algorithm with Heisenberg scaling under a different access model of quantum control, and show that some form of control is necessary.
For example, experiments of the form, prepare a state, apply $e^{\ii Ht}$, and measure, are not enough to learn $H$ with Heisenberg scaling.
Their algorithm is closely related to ours, but relies on bootstrapping the Hamiltonian learning algorithm of \cite{hkt24} which requires the Hamiltonian to be low-intersection and the structure to be known.
Further, they assume \emph{continuous quantum control}, where circuits must take the form, prepare a state, apply $e^{-\ii (H + H_0)t}$ for $H_0$ a known Hamiltonian, and measure.
This model is incompatible with the standard ``discrete'' access model, where one can only apply $e^{-\ii Ht}$ for various lengths of time, interleaved with other unitaries.
Yet, our algorithm, by nature of the types of circuit used, can be implemented in both access models, continuous and discrete.

\paragraph{Classical learning problems.}
In classical settings, one can define a joint distribution on a collection of random variables in terms of their interaction structure.
There is by now a rich and well-developed understanding of how to perform structure learning \cite{bms13, Bresler2015, VMLC16, hkm17, km17} as well as statistical \cite{sw12} and computational lower bounds \cite{km17}.
Recent work of Gaitonde and Mossel \cite{gm23} studies the problem of structure learning, not from the Gibbs distribution, but from natural dynamics that converge to it.
Notably, they show some settings where they can obtain better algorithms from learning from dynamics than are possible with i.i.d.\ samples from the Gibbs distribution.
This bears some similarities to our work, in particular since we are able to obtain structure learning from real-time evolution with algorithmic guarantees that actually surpass the natural computational lower bounds when learning from the Gibbs state.
Another conceptually related line of work studies the problem of learning linear dynamical systems from its trajectories \cite{hmr18, sbr19, blmy23a} or mixtures therein \cite{cp22, blmy23b}. 

\paragraph{Estimating the Pauli spectrum.}
There has been previous work on quantum generalizations of learning boolean functions~\cite{mo08,agy20,cny23}.
Most relevant to our work is the setting studied by Montanaro and Osborne~\cite{mo08} and, later, Angrisani~\cite{angrisani23}, where there is some unknown unitary $U$, which we think about in terms of its Pauli decomposition, $U = \sum_{Q \in \locals} c_Q Q$, and our goal is to learn the coefficients $c_Q$ from either black-box access or statistical query access to $U$.
The algorithms in these papers also draw connections to the classical Goldreich--Levin algorithm.
However, this setting differs from ours in a few ways.  The unitary that we can apply is $e^{-\ii Ht}$, which does not have low-degree Pauli spectrum: we need to first leverage \eqref{eq:intro-key} to get such a spectrum.  It is also not clear how to ``parallelize" the algorithms in prior work to only require $O(\log n)$ experiments.
Our algorithm, while tailored to Hamiltonian learning, offers several additional advantages in that it has an improved dependence on error and doesn't require entangled applications of the unitary.

\subsection{Discussion}

In this work, we introduce a framework for Hamiltonian learning with Heisenberg scaling, where our key departure from previous works is our use of \emph{term cancellation} to reshape the Hamiltonian instead of \emph{dynamical decoupling}.
Such a modification was suggested in the past~\cite{hkot23,dos23}, but we showed that this can be instantiated, and has significantly more power than was initially suggested.
The framework we present should have broad generality: the only requirement is that some version of our Trotterization lemma (\cref{lem:main}) holds, along with a protocol for Hamiltonian learning to constant error.
Neither statement requires strong locality constraints.
This flexibility makes this general strategy appealing, both for use in theoretical learning algorithms beyond qubit Hamiltonians (to fermionic or bosonic Hamiltonians), as well as for practical use.

For example, this result can perform Hamiltonian learning for the ``sparse, non-local'' settings typically considered in Hamiltonian simulation, where $H = \sum \alpha_a E_a$ for $E_a \in \locals$ not necessarily small-support, and $\sum \abs{\alpha_a}$ is bounded.
To learn this to constant error, one can use the approximation $e^{-\ii Ht} \approx \id - \ii Ht + \bigO{(t\sum \abs{\alpha_a})^2}$ (e.g.\ \cref{lem:first-order-approx}) for the Trotterization bound, and then use shadow tomography~\cite{aaronson20,bo24} to estimate expressions of the form $\tr(Pe^{-\ii Ht} Q e^{\ii Ht})$ which approximate the coefficients.
We leave determining the limits of this approach for non-local Hamiltonians to future work.

We now discuss some interesting future directions.

\begin{enumerate}
    \item Can one prove lower bounds on Hamiltonian learning?
    There is a naive lower bound of $\frac{1}{\eps}$ for estimating one parameter, and \cite{htfs23} gives an improved lower bound of $\frac{1}{\eps}\log\frac{1}{\delta}$ when the algorithm must be robust to SPAM errors.
    It is not known how the complexity of learning all parameters scales with the underlying locality or even the system size.
    Is a dependence on effective sparsity $\sparse$ necessary?
    What is the optimal dependence?
    \item Is it possible to achieve $\tet = \bigO{1/\eps}$ and $\tres = \bigOmega{1}$ for the task of learning a \emph{single} coefficient, with no dependence on system size?
    Current ``dynamical decoupling'' strategies get the time evolution but have $\tres \lesssim \sqrt{\eps}$, and our ``term cancellation'' strategies get the time resolution, but learn all parameters at once, and so have $\tet = \bigO{\log(\qubits)/\eps}$.
    It may be that the algorithm presented here still works, in the sense that with $\tet = \bigO{1/\eps}$ every estimated parameter will be accurate with constant probability, but the analysis would need to be modified to account for errors in learning parameters.
    \item Is efficient Hamiltonian learning possible with arbitrarily large time resolution?
\end{enumerate}

\newpage

%!TEX root = main.tex
\section{Background} \label{sec:prelim}

Throughout, $\log$ denotes the natural logarithm, $\ii = \sqrt{-1}$, and $[k] = \{1, 2, \dots, k\}$.
$\bigO{\cdot}$, $\bigTheta{\cdot}$, and $\bigOmega{\cdot}$ are big O notation, and we use the notation $f \lesssim g$ to mean $f = \bigO{g}$, and analogously for $\gtrsim$ and $\eqsim$.
The notation $\bigOt{f}$ denotes $\bigO{f\polylog(f)}$.
Everywhere, the binary operation $\cdot$ denotes the usual multiplication.

We use the Iverson bracket: $\iver{P} = 1$ if the proposition $P$ is true, and 0 otherwise.
The complement of a set $S \subset [n]$, $[n]\setminus S$, is denoted $\overline{S}$, and for a vector $v \in \mathbb{C}^n$, $v^S = \prod_{i \in S} v_i$.

\subsection{Linear algebra}

We work in the Hilbert space $\mathbb{C}^{\dims}$ corresponding to a system of $\qubits$ qubits, $\mathbb{C}^2 \otimes \dots \otimes \mathbb{C}^2$, so that $\dims = 2^\qubits$.
For a matrix $A$, we use $A^\dagger$ to denote its conjugate transpose and $\norm{A}$ to denote its operator norm; for a vector $v$, we use $\norm{v}$ to denote its Euclidean norm.
We will work with this Hilbert space, often considering it in the basis of (tensor products of) Pauli matrices.

\begin{definition}[Pauli matrices] \label{def:paulis}
    The Pauli matrices are the following $2 \times 2$ Hermitian matrices.
    \begin{equation*}
    \sigma_\id = \begin{pmatrix}
        1 & 0 \\ 0 & 1
    \end{pmatrix}, \qquad \sigma_x = \begin{pmatrix}
        0 & 1 \\
        1 & 0
    \end{pmatrix}, \qquad \sigma_y = \begin{pmatrix}
        0 & -\ii \\
        \ii & 0
    \end{pmatrix}, \qquad \sigma_z = \begin{pmatrix}
        1 & 0\\
        0& -1
    \end{pmatrix}.
    \end{equation*}
    These matrices are unitary and (consequently) involutory.
    Further, $\sigma_x \sigma_y = \ii \sigma_z$, $\sigma_y \sigma_z = \ii \sigma_x$, and $\sigma_z \sigma_x = \ii \sigma_y$, so the product of Pauli matrices is a Pauli matrix, possibly up to a factor of $\{\ii, -1, -\ii\}$.
    The non-identity Pauli matrices are traceless.
    We also consider tensor products of Pauli matrices, $P_1 \otimes \dots \otimes P_\qubits$ where $P_i \in \{\sigma_\id, \sigma_{x}, \sigma_{y}, \sigma_{z}\}$ for all $i \in [\qubits]$.
    The set of such products of Pauli matrices, which we denote $\locals$, form an orthogonal basis for the vector space of $2^\qubits \times 2^\qubits$ (complex) Hermitian matrices under the trace inner product.
    The product of two elements of $\locals$ is an element of $\locals$, possibly up to a factor of $\{\ii, -1, -\ii\}$.
\end{definition}

We define the following shorthand to be used in \cref{sec:shadow}.
\begin{definition}[Subset order on Pauli matrices] \label{def:pauli-order}
    For $P, Q \in \locals$, we say that $P \subseteq Q$ if $P_i \in \{\id, Q_i\}$ for every $i \in [\qubits]$.
\end{definition}

We now define the support of an operator.

\begin{definition}[Support of an operator]
    For an operator $P \in \mathbb{C}^{\dims \times \dims}$ on a system of $\qubits$ qubits, its \emph{support}, $\supp(P) \subset [\qubits]$ is the subset of qubits that $P$ acts non-trivially on.
    That is, $\supp(P)$ is the minimal set of qubits such that $P$ can be written as $P = O_{\supp(P)} \otimes \id_{[n] \setminus \supp(P)}$ for some operator $O$.
\end{definition}

So, for example, the support of a tensor product of Paulis, $P_1 \otimes \dots \otimes P_\qubits$ are the set of $i \in [\qubits]$ such that $P_i \neq \sigma_\id$.

\begin{definition}[Local Pauli operator]
    The set of Pauli matrices $P \in \locals$ such that $\abs{\supp(P)} \leq k$ is denoted $\locals_k$.
\end{definition}

A central object we consider is (nested) commutators of operators.

\begin{definition}[Commutator]\label{def:commutator}
    Given operators $A , B \in \mathbb{C}^{\dims \times \dims}$, the \emph{commutator} of $A$ and $B$ is defined as $[A,B] = AB - BA$.
    The \emph{nested commutator} of order $\ell$ is defined recursively as $[A, B]_k = [A, [A, B]_{k-1}]$, with $[A, B]_1 = [A, B]$.
\end{definition}

Pauli matrices behave straightforwardly under commutation: the commutator of two Pauli matrices is another Pauli matrix up to a scalar.

\subsection{Hamiltonians of interacting systems}
\label{subsec:hamiltonian-of-interacting-system}
We begin by defining a Hamiltonian, which encodes the interaction forces between quantum particles in a physical system.

\begin{definition}[Hamiltonian] \label{def:hamiltonian}
    A \emph{Hamiltonian} is an operator $H \in \mathbb{C}^{\dims \times \dims}$ that we consider as a linear combination of local \emph{terms} $E_a \in \locals$ with associated \emph{coefficients} $\lambda_a$, $H = H(\lambda) = \sum_{a=1}^\terms \lambda_a E_a$.
    We assume that the $E_a$'s are distinct and non-identity.
    Throughout, we assume that $\qubits = \bigO{\terms}$.\footnote{
        This can be assumed without loss by adding ``dummy'' single-qubit terms.
    }
    This Hamiltonian is \emph{$\locality$-local} if every term $E_a$ satsifies $\abs{\supp(E_a)} \leq \locality$.
\end{definition}

As is standard, we assume that the terms are distinct, non-identity Paulis, as this ensures that a Hamiltonian evolution $e^{-\ii H(\lambda) t}$ is uniquely specified by its coefficients for small $t$.
General $\locality$-local Hamiltonians $H = \sum_{S \subset [\qubits]} \lambda_S h_S$ can be written in the above form by expanding terms $h_S$ into the basis of products of Paulis, inflating the number of terms by at most $4^\locality$.
Our algorithms depend on a notion of ``local norm'', defined as follows.

\begin{definition}[Local norm of a Hamiltonian] \label{def:local-norm}
Let $H = H(\lambda) = \sum_{a=1}^\terms \lambda_a E_a$ be a Hamiltonian.
We define $\lonorm{\lambda}$ and $\ltnorm{\lambda}$ as
\begin{align*}
    \lonorm{\lambda} &= \max_{i \in [\qubits]} \sum_{\substack{a \in [\terms] \\ \supp(E_a) \ni i}} \abs{\lambda_a},
    & \ltnorm{\lambda} &= \max_{i \in [\qubits]} \parens[\Big]{\sum_{\substack{a \in [\terms] \\ \supp(E_a) \ni i}} \abs{\lambda_a}^2 }^{\frac12}. \\
\intertext{We can also write these norms directly in terms of the Hamiltonian terms:}
    \lonorm{\lambda} &= \max_{i \in [\qubits]} \sum_{\substack{a \in [\terms] \\ \supp(E_a) \ni i}} \norm{\lambda_aE_a},
    & \ltnorm{\lambda} &= \max_{i \in [\qubits]} \frac{1}{\sqrt{\dims}}\fnorm[\Big]{\sum_{\substack{a \in [\terms] \\ \supp(E_a) \ni i}} \lambda_aE_a}.
\end{align*}
Notice that $\ltnorm{\lambda} \leq \lonorm{\lambda}$.
We will sometimes abuse notation and write $\lonorm{H} = \lonorm{\lambda}$ and $\ltnorm{H} = \ltnorm{\lambda}$, though this will cause no ambiguity because $H$ and $\lambda$ are in one-to-one correspondence.
\end{definition}

The quantity $\lonorm{H}$ is commonly referred to as ``one-spin energy'', represented as $g$ or $J$ in the literature~\cite{akl16,alhambra22}.
Throughout, we assume we have a Hamiltonian where $\lonorm{\lambda}$ is bounded by, say, a constant.
This encompasses a wide range of Hamiltonians, including:
\begin{itemize}
    \item geometrically local Hamiltonians, where qubits are arranged on a lattice of constant dimension like $\mathbb{Z}^3$ and terms must be spatially local with respect to the lattice;
    \item geometrically local Hamiltonians with power law decay (see \cref{def:power-law-Ham}), where a term $E_a$ need not be spatially local but the associated coefficient $\lambda_a$ decays with the diameter of $\supp(E_a)$;
    \item and low-intersection Hamiltonians (see \cref{def:low-insersection-ham}), where the number of terms that intersect with any given term is bounded by a constant.
\end{itemize}

Our algorithm will also depend on a kind of ``local 0-norm'', which we call \emph{effective sparsity}.
\begin{definition}[Effective sparsity of a Hamiltonian] \label{def:eff-sparse}
For a Hamiltonian $H = \sum_{a = 1}^m \lambda_a E_a$, we define $H^{\leq \eps}$ to be obtained from $H$ by replacing each $\lambda_a$ with
\[
\lambda_a^{\leq \eps} = \begin{cases}
-\eps &\text{ if } \lambda_a \leq -\eps \\
\lambda_a &\text{ if } -\eps \leq \lambda_a \leq \eps \\
\eps &\text{ if } \lambda_a \geq \eps 
\end{cases}
\]
In other words, we clip all of the coefficients with magnitude larger than $\eps$.
Then the \emph{effective sparsity} of $H$ is the parameter
\begin{align*}
\sparse_\eps = \max(1 , \ltnorm{H^{\leq \eps}}^2/\eps^2) = \max\parens[\Big]{1, \sum_{a=1}^\terms \min(1, \lambda_a^2/\eps^2) } \,.
\end{align*}
\end{definition}
To see why this is a proxy for sparsity, note that when $H$ has at most $k$ terms interacting with any given site, then $\sparse_\eps \leq k$ for all $\eps$.
Further, for a Hamiltonian $H(\lambda)$ with effective sparsity $\sparse_\eps$, there must be some $\lambda'$ with at most $\sparse_\eps \qubits$ non-zero entries such that $\infnorm{\lambda' - \lambda} \leq \eps$.
More generally, $\sparse_\eps$ is non-decreasing as $\eps \to 0$, converging to the true local sparsity of $H$, and
\begin{align}
    \sparse_{\eps/C} \leq C^2 \sparse_\eps
\end{align}
for $C > 1$.
%!TEX root = main.tex

\section{Bounds on constant-resolution Trotter formulas}

In this section, we obtain the following bound on constant resolution Trotter formulas:

\begin{restatable}[Constant-time Trotterization bound]{lemma}{bch} \label{lem:main}
Let $H = \sum_{a = 1}^\terms \lambda_a E_a$ and $H_0 = \sum_{a = 1}^\terms \lambda_a^{(0)} E_a$ be $\locality$-local Hamiltonians.  Assume that $\lonorm{H}, \lonorm{H_0} \leq \degree$ and $\ltnorm{H - H_0} \leq \eta$.  Let $P \in \locals_K$.  Let $t$ be a parameter such that $t < 1/(\degree(K + \locality)^{C})$ for some sufficiently large constant $C$.  Then for any $s \in \N$ with $s\abs{t} \leq 1/\eta$, we have
\begin{align*}
    (e^{-\ii H_0t} e^{\ii Ht})^s P (e^{-\ii Ht}e^{\ii H_0t})^s &= P + \ii st[H - H_0, P] + E,
    \\
    &\qquad \text{with }\fnorm{E} \leq (\eta s t^2 \degree  + \eta^2 s^2 t^2) (K + \locality)^{C\locality}\norm{P}_F
\end{align*}
\end{restatable}

\begin{remark}[Dependence on locality] \label{rem:locality}
We did not optimize the dependence on $\locality$ and instead tried to prove a statement that holds with the weakest possible assumptions on $H,H_0$.  In particular, \cref{lem:main} only assumes that $\ltnorm{H - H_0} \leq \eta $ but if we instead assume that $\lonorm{H - H_0} \leq \eta $ then we can improve the dependence to $\exp(\locality)$.
This is because we can just apply \cref{lem:local-norm-bound1} instead of \cref{lem:local-norm-bound} (which is the only source of the $\locality^\locality$ term).
\end{remark}

\subsection{Bounds on nested commutators}

To prove this bound, we expand the expression for real-time evolution into nested commutators.
We use the following consequence of the Baker--Campbell--Hausdorff formula.
\begin{fact}[Hadamard formula] \label{fact:hadamard}
    For $\dims \times \dims$ matrices $X$ and $Y$,
    \begin{align*}
        e^{\ii X} Y e^{-\ii X} = \sum_{k = 0}^\infty \frac{1}{k!}[\ii X, Y]_k
    \end{align*}
\end{fact}

To bound the error in truncating Hadamard formula, one only needs to control the second-order nested commutator.
\begin{lemma}\label{lem:first-order-approx}
For any $H,X$,
\[
    \fnorm{e^{\ii Ht} X e^{-\ii Ht} - (X + [\ii Ht,X]) } \leq \frac{t^2}{2} \fnorm{[H,X]_2} \,. 
\]
\end{lemma}
\begin{proof}
By the Hadamard formula, $\partial_t e^{\ii Ht}Xe^{-\ii Ht} = e^{\ii Ht}[\ii H, X]e^{-\ii Ht}$.
Using this and the fundamental theorem of calculus,
\begin{align*}
    e^{\ii Ht}Xe^{-\ii Ht} - X
    &= \int_0^1 \partial_s(e^{\ii Hst}X e^{-\ii Hst}) \diff s \\
    &= \int_0^1 e^{\ii Hst}[\ii Ht, X]e^{-\ii Hst} \diff s \\
    e^{\ii Ht}Xe^{-\ii Ht} - X - [\ii Ht, X]
    &= \int_0^1 \int_0^1 \partial_r (e^{\ii Hrst}[\ii Ht, X]e^{-\ii Hrst}) \diff r \diff s \\
    &= \int_0^1 \int_0^1 e^{\ii Hrst}[\ii Hst, [\ii Ht, X]]e^{-\ii Hrst} \diff r \diff s
\end{align*}
So, taking the norm,
\begin{align*}
    \fnorm{e^{\ii Ht}Xe^{-\ii Ht} - X - [\ii Ht, X]}
    &\leq \int_0^1 \int_0^1 \fnorm{e^{\ii Hrst}[\ii Hst, [\ii Ht, X]]e^{-\ii Hrst}} \diff r \diff s \\
    &= \int_0^1 \int_0^1 s\fnorm{[Ht, [Ht, X]]} \diff r \diff s \\
    &= \tfrac{t^2}{2}\fnorm{[H, [H, X]]}. \qedhere
\end{align*}
\end{proof}

The previous lemma shows that for evolution with respect to a single Hamiltonian, we need no control over higher-order commutators, but this does not suffice for our analysis of \emph{alternating} time evolutions.
The commutators appearing in our error term are of the form
\[
    [H_\ell, [H_{\ell-1} , [ \dots [H_1, P] \dots ]]],
\]
where $H_1, \dots , H_\ell$ are local operators and $P$ has small support.
We give two lemmas showing that we can bound the local norm of $[H, G]$ in terms of the local norms of $H$ and $G$; chaining these lemmas, we get a bound on the nested commutator.

In \cref{lem:local-norm-bound}, we give a bound in terms of $\ltnorm{H}$, and in \cref{lem:local-norm-bound1}, we give one in terms of $\lonorm{H}$.
Though the former bound is too weak to prove an ``local 2-norm'' cluster expansion bound, we only use it once per nested commutator, so we only incur its associated large overhead once.

To prove these bounds, we consider writing the commutator of $H = \sum_a \lambda_a E_a$ and $G = \sum_b \kappa_b F_b$ in the Pauli basis, $[H, G] = \sum_c \xi_c X_c$, in which case bounding the norms amounts to controlling $\abs{\xi_c}^2$.
We introduce some notation for this.

\begin{definition} \label{def:cinv}
    For an $X \in \locals$, denote the set of pairs of local Paulis whose commutator is $X$ (up to a phase) as
    \begin{align*}
        \cinv(X) = \braces[\Big]{(P, Q) \in \locals^2 \,\Big|\, [P, Q] \in \{\pm 2\ii X\}}.
    \end{align*}
\end{definition}

With this, we can write
\begin{align} 
    \xi_c &= \sum_{(E_a, F_b) \in \cinv(X_c)} 2\ii \sigma_{a,b} \lambda_a \kappa_b \qquad \text{for some } \sigma_{a,b} \in \{\pm 1\}, \nonumber \\
    \abs{\xi_c} &\leq \sum_{(E_a, F_b) \in \cinv(X_c)} 2 \abs{\lambda_a \kappa_b}. \label{eq:xi-def}
\end{align}
We could give a uniform bound for every $\abs{\xi_c}^2$ in terms of the local norms of $H$ and $G$, but when summing up over $X_c$'s this approach picks up undesirable factors in the system size $\qubits$.
Instead, we will bound the full sum
\begin{align} \label{eq:xi-bound}
    \sum_c \abs{\xi_c}^2 \leq \sum_c \sum_{\substack{(E_a, F_b) \in \cinv(X_c) \\ (E_{a'}, F_{b'}) \in \cinv(X_c)}} 4\abs{\lambda_a \kappa_b \lambda_{a'} \kappa_{b'}},
\end{align}
and directly treat the quadratic terms in the sum.

\begin{lemma} \label{lem:pauli-patterns}
    Let $\mathcal{S}_X = \cinv(X) \cap (\locals_\locality \times \locals_{\locality'})$ be the set of ways to get $X$ from a commutator of a $\locality$- and $\locality'$-local Pauli.
    For some $(P, Q) \in \mathcal{S}_X$, we say its \emph{pattern} is given by the pair of sets $(\supp(P) \cap \supp(X), \supp(Q) \cap \supp(X))$.
    We can partition $\mathcal{S}_X$ into subsets based on pattern:
    \begin{align*}
        \mathcal{S}_X[S,T] = \braces[\Big]{(P, Q) \in \mathcal{S}_X \,\Big|\, \text{the pattern of } (P, Q) \text{ is } (S, T)}.
    \end{align*}
    Then, we can conclude the following:
    \begin{enumerate}[label=(\alph*)]
        \item For every $X \in \locals$, there are at most $(2(\locality+\locality'))^\locality$ distinct patterns;
        \item For every $(P, Q) \in \locals_\locality \times \locals_{\locality'}$, there are at most $(4(\locality+\locality'))^\locality$ distinct patterns $\mathcal{S}_Y[S,T]$ (across all $Y \in \locals$) that contain both a pair with $P$, $(P, R)$, and a pair with $Q$, $(R', Q)$.
        If $\supp(P)$ and $\supp(Q)$ do not intersect, there are zero such patterns.
    \end{enumerate}
\end{lemma}
\begin{proof}
First, notice that for $\mathcal{S}_X$ to be non-empty, $\abs{\supp(X)} \leq \locality + \locality'$.
So, a pattern $(S, T)$ is two subsets of $\supp(X)$, which has size at most $\locality + \locality'$.
Further, patterns satisfy that $\supp(X) \subseteq S \cup T$, since $\supp([P, Q]) \subseteq \supp(P) \cup \supp(Q)$.
Therefore, $\supp(X) = S \cup T$.

It follows that the number of distinct patterns is at most $(2(\locality + \locality'))^{\locality}$: counting the number of patterns $(S, T)$ amounts to counting $(S, S \cap T)$, using that $T = (\supp(X) \setminus S) \cup (S \cap T)$; since $\abs{\supp(S)} \leq \locality$, there are at most $(\locality + \locality')^\locality$ possibilities for $S$ and at most $2^{\abs{S}} \leq 2^\locality$ possibilities for $S \cap T$.

For part (b), fix a $(P, Q) \in \locals_\locality \times \locals_{\locality'}$ and consider a $\mathcal{S}_Y[S,T]$ which contains both some $(P, R)$ and some $(R', Q)$.
Then $\supp(Y) = S \cup T \subset \supp(P) \cup \supp(Q)$.
Further, since $[R', Q]$ is $Y$ up to scaling for some $R' \in \locals_{\locality}$, $Q$ and $Y$ only differ in at most $\locality$ qubits.
Together, this shows that all such $\mathcal{S}_Y[S,T]$ can be attained by starting from $Q$ and changing at most $\locality$ qubits in $\supp(P) \cup \supp(Q)$.
This gives a choice of $Y$, which then specifies the pattern $(S, T)$ as well.
So, there are at most $(\locality + \locality')^\locality 4^\locality$ possibilities.
\end{proof}

\begin{lemma}\label{lem:local-norm-bound}
Let $H= \sum_{a = 1}^\terms \lambda_a E_a$ and $G = \sum_{b = 1}^{\terms'} \kappa_b F_b $ be $\locality$- and $\locality'$-local Hamiltonians, respectively.
Then
\[
\begin{split}
\fnorm{[H,G]} &\lesssim (4(\locality' + \locality))^{3\locality} \ltnorm{H}\fnorm{G}  \\
\ltnorm{[H,G]} &\lesssim (4(\locality' + \locality))^{3\locality} \ltnorm{H}\ltnorm{G}
\end{split}
\]    
\end{lemma}
\begin{proof}
We first write $[H, G]$ in the Pauli basis, $[H, G] = \sum_c \xi_c X_c$, where every $X_c$ has support size at most $\locality + \locality'$.
Then, as in \eqref{eq:xi-def},
\begin{align*}
    \abs{\xi_c} &\leq \sum_{(E_a, F_b) \in \cinv(X_c)} 2 \abs{\lambda_a \kappa_b},
\end{align*}
where $\cinv(X_c)$ is the set of Paulis whose commutator evaluate to $X_c$, up to a non-zero scalar (\cref{def:cinv}).
Since $E_a$ and $F_b$ are local, we can restrict our sum to only $\mathcal{S}_{X_c} = \cinv(X_c) \cap (\locals_\locality \times \locals_{\locality'})$, and then split up the sum in terms of its associated pattern $(S, T)$, as described in \cref{lem:pauli-patterns}.
\begin{align}
    \abs{\xi_c} &\leq \sum_{\text{patterns }(S,T)}\sum_{(E_a, F_b) \in \mathcal{S}_{X_c}[S,T]} 2 \abs{\lambda_a \kappa_b} \nonumber
    \intertext{Now, we use part (a) of \cref{lem:pauli-patterns} to bound this.}
    \abs{\xi_c}^2 &\leq (\# \text{ of patterns})\sum_{\text{patterns }(S,T)}\parens[\Bigg]{\sum_{(E_a, F_b) \in \mathcal{S}_{X_c}[S,T]} 2 \abs{\lambda_a \kappa_b}}^2 \nonumber\\
    &\leq (2(\locality + \locality'))^\locality \sum_{\text{patterns }(S,T)}\parens[\Bigg]{\sum_{(E_a, F_b) \in \mathcal{S}_{X_c}[S,T]} 2 \abs{\lambda_a \kappa_b}}^2 \nonumber\\
    &\leq (2(\locality + \locality'))^\locality \sum_{\text{patterns }(S,T)}\sum_{\substack{(E_a, F_b) \in \mathcal{S}_{X_c}[S,T] \\ (E_{a'}, F_{b'}) \in \mathcal{S}_{X_c}[S,T]}} 2(\abs{\lambda_a \kappa_{b'}}^2 + \abs{\lambda_{a'}\kappa_b}^2) \label{eq:xi-f-reuse}
\end{align}
Then, we sum over $X_c$'s and use part (b) of \cref{lem:pauli-patterns} to conclude that, over all $X_c$'s, a $\abs{\lambda_a \kappa_{b'}}^2$ term appears at most $(4(\locality + \locality'))^\locality$ times.
\begin{align*}
    \fnorm{[H, G]}^2
    &= \sum_c \abs{\xi_c}^2 \\
    &\leq (2(\locality + \locality'))^\locality \sum_c\sum_{\text{patterns }(S,T)}\sum_{\substack{(E_a, F_b) \in \mathcal{S}_{X_c}[S,T] \\ (E_{a'}, F_{b'}) \in \mathcal{S}_{X_c}[S,T]}} 2(\abs{\lambda_a \kappa_{b'}}^2 + \abs{\lambda_{a'}\kappa_b}^2) \\
    &\lesssim (4(\locality + \locality'))^{2\locality} \sum_{\substack{a,b \\ \supp(E_a) \cap \supp(F_b) \neq \varnothing}} \abs{\lambda_a \kappa_{b'}}^2.
\end{align*}
Finally, using the definition of the local norm,
\[
    \sum_{\substack{a,b \\ \supp(E_a) \cap \supp(F_b) \neq \emptyset} } \abs{\lambda_a c_b}^2 \leq  \locality' \ltnorm{H}^2 \norm{G}_F^2,
\]
and combining this with the above completes the proof of the first statement.

We can deduce the second statement from the first statement.  For a fixed site $i$, note that all terms of $[H,G]$ that intersect $i$ must be of the form $[E_a, F_b]$ where $i \in \supp(E_a)$ or $i \in \supp(F_b)$.  Now we can consider these two types of terms separately and apply the first statement in the lemma, which we already we proved, to bound each of them.
\end{proof}

For the $\lonorm{\cdot}$ bound, we split up the sum in a different way.
\begin{lemma} \label{lem:pauli-types}
    Let $\mathcal{S}_X = \cinv(X) \cap (\locals_\locality \times \locals_{\locality'})$ be the set of ways to get $X$ from a commutator of a $\locality$- and $\locality'$-local Pauli.
    For some $(P, Q) \in \mathcal{S}_X$, we say its \emph{type} is given by the pair $(\abs{\supp(P)}, \abs{\supp(P) \cap \supp(Q)})$.
    We can partition $\mathcal{S}_X$ into subsets based on type:
    \begin{align*}
        \mathcal{S}_X[j,k] = \braces[\Big]{(P, Q) \in \mathcal{S}_X \,\Big|\, \text{the type of } (P, Q) \text{ is } (j, k)}.
    \end{align*}
    Then, we can conclude the following:
    \begin{enumerate}[label=(\alph*)]
        \item For every $X \in \locals$, there are at most $\locality^2$ distinct types;
        \item If $(P, Q), (P', Q') \in \mathcal{S}_X[j,k]$, then $\supp(P) \cap \supp(Q') \neq \varnothing$.
    \end{enumerate}
\end{lemma}
\begin{proof}
Part (a) follows because $\abs{\supp(P)} \leq \locality$ and $\supp(P) \cap \supp(Q)$ must be non-empty.

For part (b), consider $(P, Q), (P', Q') \in \mathcal{S}_X[j,k]$.
Because $[P, Q]$ is a non-zero scalar multiple of $X$, both $\supp(P) \setminus \supp(Q)$ and some element of $\supp(P) \cap \supp(Q)$ are in $\supp(X)$.
Further, $\supp(X) \subseteq \supp(P') \cup \supp(Q')$, so at least $\abs{\supp(P) \setminus \supp(Q)} + 1$ elements from $\supp(P)$ are in $\supp(P') \cup \supp(Q')$.
Because $(P', Q')$ has the same type as $(P, Q)$, $\abs{\supp(P') \setminus \supp(Q')} = \abs{\supp(P) \setminus \supp(Q)}$, so at least one of the $\supp(P)$ elements is not in $\supp(P') \setminus \supp(Q')$; therefore, it must be in $\supp(Q')$.
\end{proof}

\begin{lemma}\label{lem:local-norm-bound1}
Let $H= \sum_{a = 1}^\terms \lambda_a E_a$ and $G = \sum_{b = 1}^{\terms'} \kappa_b F_b $ be $\locality$- and $\locality'$-local Hamiltonians, respectively.
Then
\[
\begin{split}
\ltnorm{[H,G]} &\lesssim \locality' \locality \lonorm{H} \ltnorm{G} \\
\fnorm{[H,G]} & \lesssim \locality' \locality \lonorm{H} \fnorm{G}
\end{split}
\]
\end{lemma}
\begin{proof}
We will prove the first inequality.  The proof of the second is essentially the same. 
We first write $[H, G]$ in the Pauli basis, $[H, G] = \sum_c \xi_c X_c$, where every $X_c$ has support size at most $\locality + \locality'$.
Then, as in \eqref{eq:xi-def},
\begin{align*}
    \abs{\xi_c} &\leq \sum_{(E_a, F_b) \in \cinv(X_c)} 2 \abs{\lambda_a \kappa_b},
\end{align*}
where $\cinv(X_c)$ is the set of Paulis whose commutator evaluate to $X_c$, up to a non-zero scalar (\cref{def:cinv}).

To bound the local norm, fix a site $i \in [\qubits]$, and consider those $X_c$ which intersect $i$.
If $i \in \supp(X_c)$, then when $[E_a, F_b]$ is $X_c$ up to a non-zero scalar, either $i \in \supp(E_a)$ or $i \in \supp(F_b)$.
So, we can split up
\begin{align*}
    \sum_{\substack{c \\ i \in \supp(X_c)}} \abs{\xi_c}^2
    \lesssim \sum_{\substack{c \\ i \in \supp(X_c)}} \parens[\Big]{\sum_{\substack{(E_a, F_b) \in \cinv(X_c) \\ i \in \supp(E_a)}} \abs{\lambda_a \kappa_b}}^2 + \sum_{\substack{c \\ i \in \supp(X_c)}}\parens[\Big]{\sum_{\substack{(E_a, F_b) \in \cinv(X_c) \\ i \in \supp(F_b)}} \abs{\lambda_a \kappa_b}}^2.
\end{align*}
We can bound the first term in a straightforward manner:
\begin{align}\label{eq:type1-terms}
    \sum_{\substack{c \\ i \in \supp(X_c)}} \parens[\Big]{\sum_{\substack{(E_a, F_b) \in \cinv(X_c) \\ i \in \supp(E_a)}} \abs{\lambda_a\kappa_b}}^2
    &\leq \parens[\Bigg]{ \sum_{\substack{a \\ i \in \supp(E_a)}} \abs{\lambda_a} \parens[\Bigg]{\sum_{\substack{b \\ \supp(F_b) \cap \supp(E_a) \neq \emptyset}} \abs{\kappa_b}^2}^{1/2} }^2 \nonumber\\
    &\leq \locality^2 \cdot \lonorm{H}^2 \ltnorm{G}^2,
\end{align}
where the first inequality follows from the triangle inequality, thinking of the expression as $\norm{\sum_a \abs{\lambda_a} v^{(a)}}^2$ for $v^{(a)}$ the appropriate vector of $\abs{\kappa_b}$'s; and the second follows from the definition of local norm.

For the second term, where we take $i \in F_b$, we split the sum up into their associated types $(j,k)$, as described in \cref{lem:pauli-types}, and apply the results of that lemma.
\begin{align*}
    \sum_{\substack{c \\ i \in \supp(X_c)}}\parens[\Big]{\sum_{\substack{(E_a, F_b) \in \cinv(X_c) \\ i \in \supp(F_b)}} \abs{\lambda_a \kappa_b}}^2
    &\leq \sum_{\substack{c \\ i \in \supp(X_c)}}(\# \text{ of types}) \sum_{\text{types }(j,k)}\parens[\Big]{\sum_{\substack{(E_a, F_b) \in \mathcal{S}_{X_c}[j,k] \\ i \in \supp(F_b)}} \abs{\lambda_a \kappa_b}}^2 \\
    &\leq \locality^2 \sum_{\substack{c \\ i \in \supp(X_c)}} \sum_{\text{types }(j,k)}\sum_{\substack{(E_a, F_b) \in \mathcal{S}_{X_c}[j,k] \\ (E_{a'}, F_{b'}) \in \mathcal{S}_{X_c}[j,k] \\ i \in \supp(F_b), \supp(F_{b'})}} \abs{\lambda_a \kappa_b \lambda_{a'} \kappa_{b'}} \\
    &\leq \locality^2 \sum_{\substack{c \\ i \in \supp(X_c)}} \sum_{\text{types }(j,k)}\sum_{\substack{(E_a, F_b) \in \mathcal{S}_{X_c}[j,k] \\ (E_{a'}, F_{b'}) \in \mathcal{S}_{X_c}[j,k] \\ i \in \supp(F_b), \supp(F_{b'})}} \abs{\lambda_a \lambda_{a'}} \abs{\kappa_b}^2
\intertext{Note that the above sum has no double-counting, in that every collection of indices $(a, b, a', b') \mapsto \abs{\lambda_a \lambda_a' \kappa_b^2}$ appears at most once: the choice of $a$ and $b$ determines the corresponding $X_c$ as well as the type, and this $X_c$ along with the $a'$ determines the $F_{b'}$.
So, to bound this, we note that all terms appearing in the sum satisfy that $i \in \supp(F_b)$, $\supp(E_a)$ intersects $\supp(F_b)$, and by part (b) of \cref{lem:pauli-types}, $\supp(E_{a'})$ also intersects $\supp(F_b)$.}
    \sum_{\substack{c \\ i \in \supp(X_c)}}\parens[\Big]{\sum_{\substack{(E_a, F_b) \in \cinv(X_c) \\ i \in \supp(F_b)}} \abs{\lambda_a \kappa_b}}^2
    &\leq \locality^2 \sum_{\substack{b \\ i \in \supp(F_b)}} \abs{\kappa_b}^2 \parens[\Bigg]{ \sum_{\substack{a,a' \\ \supp(E_a) \cap \supp(F_b) \neq \varnothing \\ \supp(E_{a'}) \cap \supp(F_b) \neq \emptyset}} \abs{\lambda_a} \cdot \abs{\lambda_{a'}} } \\
    &= \locality^2 \sum_{\substack{b \\ i \in \supp(F_b)}} \abs{\kappa_b}^2 \parens[\Big]{ \sum_{\substack{a \\ \supp(E_a) \cap \supp(F_b) \neq \varnothing}} \abs{\lambda_a}}^2 \\
    &\lesssim \locality^2  \locality'^2 \cdot \lonorm{H}^2 \ltnorm{G}^2 \,.
\end{align*}
Since this holds for all $i \in [\qubits]$, $\ltnorm{[H, G]}^2 \lesssim \locality^2  \locality'^2 \cdot \lonorm{H}^2 \ltnorm{G}^2$ as desired.

To prove the second inequality, we can just use the same argument as for bounding the second term above except the sum over $c$ is over all $X_c$ instead of just those containing some site $i$.  
\end{proof}

\subsection{Bivariate nested commutators}

We will also need to ``reorder'' nested commutators, which can be done with tools developed in prior work~\cite{blmt24}.

\begin{definition}[Bivariate nested commutators, {\cite[Definition 3.3]{blmt24}}]
\label{def:bi-variate-nested-commutators}
Let $S \in \{0,1 \}^\ell$ and $X,Y,A \in \C^{\dims \times \dims}$ be matrices.  Consider a sequence $Z_1,Z_2,\dots, Z_\ell$ of length $\ell$  where each $Z_i \in \{ X,Y \}$ and $Z_i = X$ if and only if the $i$th entry of $S$ is $0$.  We define 
\[
[ (X,Y)_S ,A] = [Z_1, [Z_2, [ \dots [Z_\ell, A ] \dots ]]] \,.
\]
\end{definition}

For our bounds, we will need to analyze the difference between two bivariate nested commutators which involve the same matrices with the same multiplicities, but with a different order.
When the matrices $X$ and $Y$ are close to commuting, which in particular holds when their difference is small, reordering does not change the value of the nested commutator by too much.  When $\abs{S} = 2$, this follows from the identity below.

\begin{fact}[Jacobi identity]\label{fact:switch:HH'}
We have the identity $[X,[Y,A]] - [Y, [X, A]] = [[X,Y],A]$.
\end{fact}

We extend this to higher-order commutators by induction.
\begin{lemma}[Reordering bivariate nested commutators, {\cite[Lemma 3.5]{blmt24}}]\label{lem:polynomial-equivalence}
For any two sequences $S,S' \in \{0,1 \}^\ell$ with the same number of $0$'s and $1$'s, let $t \leq \ell^2$ be the number of adjacent swaps needed to transform $S$ to $S'$.
Then there are some coefficients $c_1, \dots , c_t \in \{-1,1 \}$, and sequences $S_1,T_1, \dots , S_t, T_t$ where $\len(S_i) + \len(T_i) = \ell - 2$ such that
\[
[(X,Y)_S,A] - [(X,Y)_{S'},A] = \sum_{i = 1}^t c_i \left[ (X,Y)_{S_i}, \left[[X,Y], \left[(X,Y)_{T_i},A \right]\right] \right]. 
\]
\end{lemma}
\begin{proof}
Consider when $S,S'$ differ exactly by a single swap of two adjacent elements.  In this case, by Fact~\ref{fact:switch:HH'}, the difference on the LHS is equal to exactly one term of the form
\[
\left[ (X,Y)_{S_i}, \left[[X,Y], \left[(X,Y)_{T_i},A \right]\right] \right] 
\]
where $S_i$ is the prefix up to the point where $S,S'$ differ and $T_i$ is the suffix.  Now we can repeatedly apply this to swap adjacent elements of $S$ until it matches $S'$.  Each of the residual terms is of the form given on the RHS so we are done.
\end{proof}

\subsection{\texorpdfstring{Proof of \cref{lem:main}}{Proof of Lemma 3.1}}

Now we prove our main result of this section.
The error we need to bound for this takes the form of nested commutators, as shown in the lemma below.

\begin{lemma}[Observable-based Trotter error] \label{one-peel}
    Let $X$, $H$, and $H_0$ be arbitrary Hermitian matrices. For any $t \in \mathbb{R}$, 
    \begin{align*}
        e^{-\ii H_0 t} e^{\ii H t} X e^{-\ii H t} e^{\ii H_0}
        - e^{\ii (H-H_0) t} X e^{-\ii (H-H_0) t}
        = \sum_{\ell \geq 2} \frac{(\ii t)^\ell}{\ell!} C_\ell,
    \end{align*}
    where $C_\ell$ is a sum of at most $2^\ell \ell^2$ order-$(\ell-1)$ nested commutators of $-H_0$, $H$, one copy of $[H, -H_0]$, and $X$ in the center.
\end{lemma}
\begin{proof}
We expand the expression as follows.
\begin{align*}
    & e^{-\ii H_0 t} e^{\ii H t} X e^{-\ii H t} e^{\ii H_0} - e^{\ii (H-H_0) t} X e^{-\ii (H-H_0) t} \\
    &= \sum_{k, k_0 \geq 0} \frac{(\ii t)^{k+k_0}}{k!k_0!}[-H_0, [H, X]_k]_{k_0} - \sum_{\ell \geq 0} \frac{t^\ell}{\ell!} [H-H_0, X]_\ell \\
    &= \sum_{\ell \geq 0} \frac{(\ii t)^\ell}{\ell!} \parens*{\parens[\Big]{\sum_{k = 0}^\ell \binom{\ell}{k}[-H_0, [H, X]_k]_{\ell-k}} - [H-H_0, X]_\ell} \\
    &= \sum_{\ell \geq 0} \frac{(\ii t)^\ell}{\ell!} \parens[\Big]{\sum_{S \in \{0,1\}^\ell} [-H_0, [H, X]_{\abs{S}}]_{\ell - \abs{S}} - [(-H_0, H)_S, X]}
\end{align*}
The first equality follows from two applications of the Hadamard formula (\cref{fact:hadamard}), the second follows from counting the number of ways $k+k_0 =\ell$, and the third expands the nested commutator using linearity, $[H-H_0, X]_\ell = \sum_{S \in \{0,1\}^\ell}[(-H_0, H)_S, X]$.
Note that this sum is zero for $\ell \in \{0, 1\}$.
We know from \cref{lem:polynomial-equivalence} that, for some $S_i, T_i$, and $c_i \in \{-1, 0, 1\}$,
\begin{align*}
    [-H_0, [H, X]_{\abs{S}}]_{\ell - \abs{S}} - [(-H_0, H)_S, X]
    = \sum_{i=1}^{\ell^2} c_i [(-H_0, H)_{S_i}, [[-H_0, H], [(-H_0, H)_{T_i}, X]]].
\end{align*}
This is the desired form.
\end{proof}

Now we move to the proof of \cref{lem:main}.
We restate it here.

\bch*

\begin{proof}[Proof of \cref{lem:main}]
We prove the claim by induction on $s$.
When $s = 0$ the statement is trivial.
Now, fix some $s$, and denote $X = P + \ii st[H - H_0, P]$.  Now by \cref{one-peel}, we can write 
\[
    e^{-\ii H_0t} e^{\ii Ht} X e^{-\ii Ht}e^{-\ii H_0t} -  e^{\ii(H - H_0)t}X e^{-\ii(H - H_0)t} = \sum_{\ell \geq 2} \frac{(\ii t)^{\ell}}{\ell!} C_{\ell}\,.
\]
Since the terms in $X$ all intersect with $\supp(P)$, we have that $\norm{X}_F \leq (1 + \eta stK) \norm{P}_F \leq 2K\norm{P}_F$ and that $X$ is $(K+\locality)$-local.  Also, we have
\[
    \ltnorm{[H_0, H]} = \ltnorm{[H_0, H - H_0]} \lesssim \locality^2 \degree \eta 
\]
by \cref{lem:local-norm-bound1}.
Our goal is to bound the nested commutators in $C_\ell$, which contain one copy of $[H, H_0]$ and $\ell-1$ copies of either $H$ or $-H_0$.
For illustration, representative commutator in $C_4$ is shown below.
\begin{align*}
    [-H_0, [[H, -H_0], [H, X]]]
\end{align*}
We can apply \cref{lem:local-norm-bound1} and \cref{lem:local-norm-bound} repeatedly to bound each of these terms in $C_{\ell}$.  In particular, we apply \cref{lem:local-norm-bound} for the one commutator involving $[H, -H_0]$ and \cref{lem:local-norm-bound1} for all other layers of the nested commutator.  For our example, we get
\begin{align*}
    \fnorm{[-H_0, [[H, -H_0], [H, X]]]}
    &\lesssim \locality(K + 4\locality) \lonorm{H_0} \fnorm{[[H, -H_0], [H, X]]} \tag*{by \cref{lem:local-norm-bound1}}\\
    \fnorm{[[H, -H_0], [H, X]]}
    &\lesssim (10(K + \locality))^{6\locality} \ltnorm{[H, -H_0]} \fnorm{[H, X]} \tag*{by \cref{lem:local-norm-bound}}\\
    \fnorm{[H, X]}
    &\lesssim \locality(K + \locality)\lonorm{H}\fnorm{X} \tag*{by \cref{lem:local-norm-bound1}}.
\end{align*}
Combining the above gives us a bound on the nested commutator.
In general, a commutator in $C_\ell$ can be bounded, paying the $\lonorm{\cdot}$ bound $\ell-2$ times and the $\ltnorm{\cdot}$ bound once, giving a final bound of
\begin{align*}
    &\ell! (c\locality(K + \locality))^{2\ell} \degree^{\ell - 2} \cdot (10(K + \ell\locality))^{6\locality} \ltnorm{[H, -H_0]} \fnorm{X} \\
    &\leq \ell! (K + \locality)^{c'(\ell + \locality)} \degree^{\ell - 1} \ell^{c' \locality} \eta \fnorm{P},
\end{align*}
where $c$ and $c'$ are sufficiently large constants.
Since $C_\ell$ is the sum of $2^\ell \ell^2$ of such commutators, we have
\begin{align}
    &\fnorm{e^{-\ii H_0t} e^{\ii Ht} X e^{-\ii Ht}e^{\ii H_0t} -  e^{\ii(H-H_0)t}X e^{-\ii(H - H_0)t}} \nonumber\\
    &\leq \sum_{\ell \geq 2} \frac{\abs{t}^\ell}{\ell!}\fnorm{C_\ell} \nonumber \\
    &\leq \sum_{\ell \geq 2} (2^\ell \ell^2) \abs{t}^\ell (K + \locality)^{c'(\ell + \locality)} \degree^{\ell - 1} \ell^{c' \locality} \eta \fnorm{P} \nonumber \\
    &\leq t^2 \degree (K + \locality)^{\bigO{\locality}} \eta \norm{P}_F, \label{eq:error1}
\end{align}
where the last inequality follows by the assumption on $t$: for $t < 1/(100\degree(K + \locality)^{c'})$, the sum over $\ell$ is dominated by the $\ell = 2$ term, up to a factor of $\locality^{\bigO{\locality}}$.
Next, we bound the difference between evolution by $H-H_0$ with its first-order approximation.
Using triangle inequality, \cref{lem:first-order-approx}, and then \cref{lem:local-norm-bound} twice,
\begin{align}
&\fnorm{e^{\ii(H - H_0)t}X e^{-\ii(H - H_0)t} - (X + \ii t[H - H_0, P])} \nonumber\\
&\leq \fnorm{e^{\ii(H - H_0)t}X e^{-\ii(H - H_0)t} - (X + \ii t[H - H_0, X])} + \fnorm{st^2[H - H_0, P]_2} \nonumber\\
&\leq \tfrac{t^2}{2} \fnorm{[H - H_0,X]_2} + \fnorm{st^2[H - H_0, P]_2} \nonumber\\
&\leq (K + \locality)^{\bigO{\locality}} t^2 \ltnorm{H_0 - H} (\fnorm{[H - H_0,X]} + s\fnorm{[H-H_0, P]}) \nonumber\\
&\leq (K + \locality)^{\bigO{\locality}}  s t^2 \eta^2 \fnorm{P}. \label{eq:error2}
\end{align}
Now we can do the inductive step.  Let
\[
\Delta(s) = \fnorm{ (e^{-\ii H_0t} e^{\ii Ht})^{s+1} P (e^{\ii Ht}e^{-\ii H_0t})^{s+1} - (P + \ii (s + 1)t[H - H_0, P] )}
\]
By triangle inequality,
\begin{align*}
    \Delta(s+1) &\leq \fnorm{ (e^{-\ii H_0t} e^{\ii Ht})^{s+1} P (e^{\ii Ht}e^{-\ii H_0t})^{s+1} - e^{-\ii H_0t} e^{\ii Ht} (P + \ii st[H - H_0, P] ) e^{\ii Ht}e^{-\ii H_0t}} \\
        & \qquad + \fnorm{e^{-\ii H_0t} e^{\ii Ht} (P + \ii st[H - H_0, P] ) e^{\ii Ht}e^{-\ii H_0t}- (P + \ii (s + 1)t[H - H_0, P] )}
\intertext{Since $e^{\ii Ht}$ and $e^{-\ii H_0 t}$ are unitary, the first term is $\Delta(s)$.
Substituting in the definition of $X$ and combining \eqref{eq:error1} and \eqref{eq:error2}, we bound the rest:}
    \Delta(s+1) &\leq \Delta(s) + \fnorm{e^{-\ii H_0t} e^{\ii Ht} X e^{\ii Ht}e^{-\ii H_0t}- (X + \ii t[H - H_0, P] )} \\
    & \leq \Delta(s) + t^2\degree (K + \locality)^{\bigO{\locality}} \eta \norm{P}_F+ st^2 \eta^2 (K + \locality)^{\bigO{\locality}} \norm{P}_F 
\end{align*}
Combining the above over all $s$ completes the proof.
\end{proof}

%!TEX root = main.tex

\section{Estimating expectations of Pauli observables} \label{sec:shadow}

Our algorithms ultimately reduce to estimating particular trace expressions of the unknown Hamiltonian; in this section, we describe how to do this.
This is the only part of our algorithm that requires the quantum computer, and is the only way in which we access our unknown Hamiltonian.

The algorithms described in this section proceed by running the same type of simple circuit many times, non-adaptively, to generate a ``dataset'', and then post-processes the dataset to generate estimates of relevant statistics.
Our workhorse circuit is denoted $\mathcal{C}(A,v,B,Z)$, where $A,v  \in (\{\sigma_x, \sigma_y, \sigma_z\} \times \{+,-\})^\qubits$ is a Pauli eigenvector, $B \in \{\sigma_x, \sigma_y, \sigma_z\}^\qubits$ is a tensored Pauli basis, and $Z$ is an unknown unitary.
It is shown below.
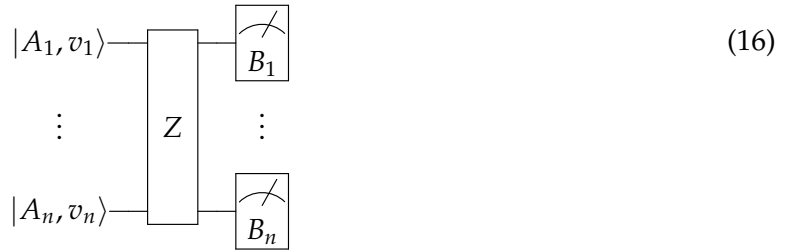
\begin{align} \label{eq:the-circuit}
    \Qcircuit @R=.6em @C=0.65em {
        \push{\ket{A_1, v_1}} & \qw & \multigate{3}{Z} & \qw & \meterB{B_1} \\
        \push{\vdots} &&&& {\vdots}\\
        &&&&&& \\
        \push{\ket{A_\qubits, v_\qubits}} & \qw & \ghost{Z} &\qw & \meterB{B_\qubits}
    }
\end{align}
In words, we prepare the eigenstate associated to $(A,v)$; apply $Z$; then measure qubit $i$ in the eigenbasis of $B_i$.
We interpret the output of the circuit as a vector $w \in \{\pm 1\}^\qubits$, where $w_i$ be the associated eigenvalue of the outcome of the $B_i$ measurement.

Note that this circuit can be implemented by initializing in $\ket{0}^{\otimes \qubits}$, applying one layer of single-qubit Clifford gates, $Z$, then one more layer of single-qubit Clifford gates, followed by measurement in the computational basis.
As such, it can be done with $\qubits$ qubits, one application of $Z$, and $\bigO{\qubits}$ additional single-qubit gates. 

To analyze this circuit, we extensively use the following fact: for an $\qubits$-qubit Pauli $B \in \locals$ and a subset $S \subset [\qubits]$,
\begin{equation} \label{eq:pauli-eigenvector-sum}
    \sum_{w \in \{\pm 1\}^\qubits} w^S \ket{B, w}\bra{B, w} = B_S \otimes \id_{\overline{S}},
\end{equation}
recalling the notation $w^S = \prod_{i \in S} w_i$ and $\overline{S} = [\qubits] \setminus S$.
Using this, we can conclude the following about our workhorse.
\begin{fact} \label{fact:circuit-helper}
    Fix $A, B \in \{\sigma_x, \sigma_y, \sigma_z\}^\qubits$ and $v \in \{\pm 1\}^\qubits$.
    Let $w \in \{\pm 1\}^\qubits$ be the output of $\circuit(A,v,B,Z)$.
    Then, over the randomness of the circuit, for every $S \subset [\qubits]$,
    \begin{align*}
        \E[w^{S}]
        &= \sum_{w \in \{\pm 1\}^\qubits} \Pr[\circuit(A,v,B,Z) \text{ outputs } w] w^{S} \\
        &= \sum_{w \in \{\pm 1\}^\qubits} \tr\parens[\Big]{\ket{B,w}\bra{B,w}Z\ket{A, v}\bra{A, v}Z^\dagger} w^{S} \\
        &= \tr\parens[\Big]{(B_S \otimes \id_{\overline{S}}) Z\ket{A, v}\bra{A, v}Z^\dagger}.
    \end{align*}
\end{fact}

In this section, we use the Iverson bracket: $\iver{P} = 1$ if $P$ is true, and $0$ otherwise.
We also use the notation $P \subseteq Q$ for $P, Q \in \locals$ to denote that $P$ can be formed from $Q$ by replacing individual qubits $Q_i$ with the identity (\cref{def:pauli-order}).

\subsection{Estimating expectations when terms are known}

First, we introduce a method for estimating expectations of the form $\frac{1}{\dims}\tr(PZQZ^\dagger)$.
This is essentially due to prior work~\cite[Lemma A.5]{hkt24}, but we reprove it here in a slightly more general form.

\begin{lemma} \label{lem:shadows}
    Given the ability to apply the unitary $Z$; locality parameters $\locality, \locality'$; and error parameters $\eps, \delta > 0$; there is a quantum algorithm which applies $Z$ at most $S = \bigO{\frac{3^{2(\locality + \locality')}}{\eps^2}\log\frac{\qubits^{\locality + \locality'}}{\delta}}$ times, uses $\bigO{S\qubits}$ additional gates, $\bigO{S\qubits}$ classical overhead, and outputs an oracle.  With probability $\geq 1-\delta$, this oracle can respond to the following type of query in $\bigO{S(\locality + \locality')}$ time on a classical computer: given $P \in \locals_{\locality}$ and $Q \in \locals_{\locality'}$, output a $\mu$ such that
    \begin{align*}
        \abs[\Big]{\mu - \tfrac{1}{\dims}\tr(PZQZ^\dagger)} < \eps.
    \end{align*}
\end{lemma}

\begin{proof}
We analyze \cref{algo:shadow}.
The pre-processing algorithm runs in time linear in the size of the dataset, which is $\bigO{S\qubits}$.
The query algorithm runs in $\bigO{\locality + \locality'}$ time per sample, giving $\bigO{S(\locality + \locality')}$ in total, since the sample only needs to be queries on the supports of $X$ and $P$.
It remains to show correctness.

Consider some input query $X \in \locals_\locality$, $P \in \locals_{\locality'}$.
Then the output, $\mu$, is an average over $S$ i.i.d.\ random variables of the form $(3v)^{\supp(X)}(3w)^{\supp(P)}\iver{X \subseteq A,\,P \subseteq B}$.
These random variables are unbiased estimators of our desired quantity:
\begin{align*}
    \E[\mu] &= \expec[\Big]{A,v,B}{\expec{w}{(3v)^{\supp(X)}(3w)^{\supp(P)}\iver{X \subseteq A,\,P \subseteq B} \mid A,v,B}} \\
    &= \expec[\Big]{A,v,B}{\expec{w}{v^{\supp(X)}w^{\supp(P)} \mid A,v,B} \,\Big|\, X \subseteq A,\,P \subseteq B} \\
    &= \expec[\Big]{A,v,B}{\tr\parens[\Big]{(B_{\supp(P)} \otimes \id_{\overline{\supp(P)}}) Z (v^{\supp(X)} \ket{A, v}\bra{A, v})Z^\dagger} \,\Big|\, X \subseteq A,\,P \subseteq B} \\
    &= \tr\parens[\Big]{P Z \expec[\Big]{A,v}{v^{\supp(X)} \ket{A, v}\bra{A, v} \,\Big|\, X \subseteq A}Z^\dagger} \\
    &= \frac{1}{\dims}\tr\parens[\Big]{P Z X Z^\dagger}.
\end{align*}
Further, each random variable is bounded in magnitude by $3^{\locality+\locality'}$.
Thus, by Hoeffding's inequality, for $S \gtrsim \frac{3^{2(\locality + \locality')}}{\eps^2}\log\frac{\qubits^{\locality + \locality'}}{\delta}$, the output $\mu$ satisfies $\abs{\mu - \frac{1}{\dims}\tr\parens{P Z X Z^\dagger}} < \eps$ with probability $\geq 1 - \delta/\qubits^{\locality+\locality'}$.
Since there are at most $\qubits^{\locality+\locality'}$ choices of $X$ and $P$, by a union bound we can conclude that the algorithm will be correct for \emph{all} choices with probability $\geq 1-\delta$.
\end{proof}

\begin{mdframed}
\begin{algorithm}[Classical shadows for operators]
    \label{algo:shadow}\mbox{}
    \begin{description}
    \item[Input:] Black-box ability to apply an unknown $\qubits$-qubit unitary $Z$; locality parameters $\locality, \locality'$; error parameters $\eps, \delta > 0$.
    \item[Output:] A oracle that, with probability $\geq 1-\delta$, can successfully respond to the following queries: given Paulis $X \in \locals_\locality$ and $P \in \locals_{\locality'}$, output an estimate $\mu$ such that
    \begin{align*}
        \abs[\Big]{\mu - \tfrac{1}{\dims}\tr(Z^\dagger PZQ)} < \eps.
    \end{align*}
    \item[Pre-processing (Generating the dataset):]\mbox{}
    \begin{algorithmic}[1]
        \State Let $S \gets \bigTheta{\frac{3^{2(\locality + \locality')}}{\eps^2}\log\frac{\qubits^{\locality + \locality'}}{\delta}}$;
        \For{$i \in [S]$}
            \State Sample a uniformly random Pauli eigenvector $A, v \sim (\{\sigma_x, \sigma_y, \sigma_z\} \times \{+, -\})^{\qubits}$;
            \State Sample a uniformly random Pauli string $B \sim \{\sigma_x, \sigma_y, \sigma_z\}^{\qubits}$;
            \State Run the circuit $\circuit(A, v, B, Z)$~\eqref{eq:the-circuit};
            \State Record the output $w \in \{\pm 1\}^n$;
        \EndFor
        \State Output the dataset $\dataset = \{(A^{(k)}, v^{(k)}, B^{(k)}, w^{(k)})\}_{k \in [S]}$.
    \end{algorithmic}
    \item[Query subroutine:] Given input $X \in \locals_\locality$, $P \in \locals_{\locality'}$, and the dataset $\dataset$, outputs estimate $\mu$. \mbox{}
    \begin{algorithmic}[1]
        \State For every $(A, v, B, w) \in \dataset$, compute the associated estimator
        \[
            \begin{cases}
                (3v)^{\supp(X)}(3w)^{\supp(P)}
                    & \text{if } X \subseteq A \text{ and } P \subseteq B \\
                0 & \text{otherwise}
            \end{cases} \vspace{1em}
        \]
        \State Output the average of this estimator over all the elements of $\dataset$;
    \end{algorithmic}
    \end{description}
\end{algorithm}
\end{mdframed}

\subsection{Finding unknown terms time-efficiently} \label{sec:structure-fast}

\begin{theorem} \label{thm:structure}
    Suppose we are given a locality parameter $\locality$, error parameters $\eps, \delta > 0$, and the ability to apply a unitary $Z$ which ``approximates the real-time evolution $e^{\ii \hat{H}}$ up to first-order'', in the sense that there is some $\hat{H} = \sum_{Q \in \locals_{\locality}} \lambda_Q Q$ such that, for every $P \in \locals_1$,
    \begin{equation} \label{eq:thm-structure}
        Z^\dagger PZ = P + [\ii\hat{H}, P] + E^{(P)}
        \qquad \text{with } \frac{1}{\sqrt{\dims}}\fnorm{E^{(P)}} < \frac{\eps}{20^{\locality+2}}
    \end{equation}
    Then there is a quantum algorithm (\cref{algo:structure}) which outputs an estimate to the coefficient vector $\tilde{\lambda}$ such that, with probability $\geq 1-\delta$, $\infnorm{\lambda - \tilde{\lambda}} < \eps$.

    This algorithm applies the unknown unitary $\bigO{\frac{e^{8\locality}}{\eps^2}\log\frac{\qubits}{\delta}}$ times (with depth-3 non-adaptive prepare-apply-measure circuits) and $\bigO{\qubits^2 \sparse_{\eps} \frac{e^{15\locality}}{\eps^2}\log\frac{\qubits}{\delta}}$ additional running time on a classical computer, where $\sparse_\eps$ is the effective sparsity of $\hat{H}$ (\cref{def:eff-sparse}).
    Further, $\tilde{\lambda}$ has sparsity $\bigO{e^{7\locality}\sparse_{\eps} \qubits}$.
\end{theorem}
\begin{remark} \label{rmk:sparsity}
Even though there are $n^{\locality}$ possible terms $Q \in \locals_{\locality}$, the output of our algorithm will be sparse and all of the coefficients that our algorithm does not ``explicitly compute"  will be set to $0$ by default.  Our algorithm will only actually compute non-zero coefficients which allows us to get a fixed $\poly(\qubits)$ runtime (instead of $\qubits^{\locality}$).
\end{remark}

Our prior mechanism for querying $Z^\dagger P Z$, the shadow oracle, does not suffice, since this only gives an estimate to a single $\frac{1}{\dims} \tr(Z^\dagger PZQ)$, corresponding to a single Pauli coefficient of $\hat{H}$.
Since there can be are $n^\locality$ coefficients which could be non-zero, structure learning requires $n^\locality$ time just using the shadow oracle.
To improve on this, we introduce a different kind of query to $Z^\dagger P Z$, where given some Pauli $X$, we detect whether some coefficient ``above'' $X$ is non-zero.
That is, the goal is to find some $Q \supseteq X$ such that $\abs{\frac{1}{\dims}\tr(Z^\dagger PZ Q)}$ is large.
To distinguish these from shadow queries, we call these new queries \emph{GL queries}.
We show that GL queries can be answered efficiently.

\begin{proposition}[Constructing an oracle for GL queries] \label{prop:structure-query}
Suppose we have black-box access to the unitary gate $Z \in \mathbb{C}^{\dims \times \dims}$, and are given as input the locality parameter $\locality$ and error parameters $\gamma, \delta > 0$. 
Then there is a quantum algorithm which applies $Z$ at most $\bigO{\frac{e^{8\locality}}{\gamma^2}\log\frac{\qubits}{\delta}}$ times and, after $\bigO{\qubits\frac{e^{8\locality}}{\gamma^2}\log\frac{\qubits}{\delta}}$ classical overhead, outputs an oracle.  With probability $\geq 1-\delta$, this oracle can respond to the following type of query in $\bigO{\frac{e^{8\locality}}{\gamma^2}\log\frac{\qubits}{\delta}}$ time on a classical computer: given $X \in \locals_\locality$ and $P \in \locals_1$, respond ``Pass'' or ``Fail'' with the following guarantees.
Let $c_Q \deq \frac{1}{\dims} \tr(Z^\dagger PZQ)$.
\begin{align}
    &\text{If}& \sum_{\substack{Q \supseteq X \\ \abs{\supp(Q)} \leq \locality}} c_Q^2 \geq \gamma^2
    \text{ and } \sum_{\substack{Q \supseteq X \\ \abs{\supp(Q)} > \locality}} c_Q^2 &< \frac{\gamma^2}{400^{\locality+2}}, &\text{ output ``Pass'';}\label{eq:structure-success}\\
    &\text{If}& \sum_{Q \supseteq X} c_Q^2 &< \frac{\gamma^2}{400^{\locality+1}}, &\text{ output ``Fail''.} \label{eq:structure-fail}
\end{align}
If neither condition holds, the algorithm may output either ``Pass'' or ``Fail''.
\end{proposition}

\begin{remark}
    This algorithm is able to determine, given some $X$, whether there is a $\locality$-local $Q$ where $Q \supseteq X$ and $\frac{1}{\dims}\tr(Z^\dagger PZQ)$ is $\gamma$-far from zero.
    Though we only prove guarantees under certain conditions, our general strategy can be used to get an estimator to
    \begin{align*}
        \sum_{Q \supseteq X} \frac{c_Q^2}{6^{\abs{\supp(Q)}}},
    \end{align*}
    which can be seen as a general measure of how much of the mass of $Z^\dagger P Z$ is placed on Pauli terms which contain $X$.
    See \cref{lem:gl-general} for a version of this statement.
\end{remark}

First, we will show how to prove \cref{thm:structure} assuming \cref{prop:structure-query}.  The algorithm for this is described in \cref{algo:structure}.  We will then prove \cref{prop:structure-query} afterwards.

\begin{mdframed}
\begin{algorithm}[Time-efficient structure learning from real-time evolution]
    \label{algo:structure}\mbox{}
    \begin{description}
    \item[Input:] Black-box ability to apply an unknown $\qubits$-qubit unitary $Z$; locality parameter $\locality$; and error parameters $\eps, \delta > 0$.
    \item[Output:] An estimate $\tilde{\lambda}$ to the Hamiltonian coefficients corresponding to $Z$.
    \item[Shadow oracle:] Given Paulis $X \in \locals_\locality$ and $P \in \locals_1$, outputs an estimate $\mu$ such that $\abs{\mu - \frac{1}{\dims}\tr(Z^\dagger PZX)} < \eps$.
    \item[GL oracle:] Given Paulis $X \in \locals_\locality$ and $P \in \locals_1$, outputs ``Pass'' or ``Fail'' with the guarantees described in \eqref{eq:structure-success} and \eqref{eq:structure-fail}, with $\gamma \gets \eps$.
    \item[Procedure:]\mbox{}
    \begin{algorithmic}[1]
        \State Perform \cref{algo:shadow} with parameters $(\locality, \locality', \eps, \delta) \leftarrow (\locality, 1, \eps, \delta/2)$ to generate the shadow oracle with probability $\geq 1-\delta/2$ (see above).
        \State Perform \cref{algo:structure-query} with parameters $(\locality, \gamma, \delta) \leftarrow (\locality, \eps, \delta/2)$ to generate the GL oracle with probability $\geq 1-\delta/2$ (see above).
       
        \State Initialize $\tilde{\lambda} \gets \vec{0}$;
        \Comment{This is our running coefficient vector estimate}
        \ForAll{1-local $P \in \locals_1$}
            \Comment{First, find the non-zero coefficients with GL queries}
            \State Let $\mathcal{Q}_1^{(P)} \gets \{Q \in \locals_1 \mid \supp(Q) = \supp(P), Q \neq P\}$;
            \For{$k$ from 2 to $\locality$}
                \State Initialize $\mathcal{Q}_k^{(P)} \gets \varnothing$;
                \ForAll{$Q \in \locals_k \setminus \locals_{k-1}$ such that $Q \supseteq Q'$ for some $Q' \in \mathcal{Q}_{k-1}^{(P)}$}
                    \State Query the GL oracle with $Q$ and $P$;
                    \State If the oracle returns ``Pass'', add $Q$ to $\mathcal{Q}_k^{(P)}$;
                \EndFor
            \EndFor
            \State Let $\mathcal{Q}^{(P)} \gets \mathcal{Q}_1^{(P)} \cup \dots \cup \mathcal{Q}_\locality^{(P)}$.
            \ForAll{$Q \in \mathcal{Q}^{(P)}$}
                \Comment{Then, estimate non-zero coefficients with shadow queries}
                \State Query $Q, P$ with the shadow oracle to get the estimate $\tilde{\lambda}_Q^{(P)}$, and set
                \[
                    \tilde{\lambda}_{R} \leftarrow \frac{(-1)^b}{2}\tilde{\lambda}_{Q}^{(P)},
                \]
                where $R \in \locals$ such that $R = (-1)^{b}\ii PQ$; \label{compute-coeff} 
            \EndFor
        \EndFor
        \State Output $\tilde{\lambda}$.
    \end{algorithmic}
    \end{description}
\end{algorithm}
\end{mdframed}

\begin{remark}
Note that in \cref{compute-coeff}, some coefficient $\tilde{\lambda}_R$ might be set multiple times as we enumerate over all $1$-local $P \in \locals_1$.  This is fine as it will be clear from the analysis that all of the estimates are consistent up to some $\eps$ error. 
\end{remark}

\begin{proof}[Proof of \cref{thm:structure}]
We analyze \cref{algo:structure}.
Let $\hat{\lambda}^{(P)}$ be the coefficient vector for $[\ii\hat{H}, P]$.
\begin{align} \label{eq:structure-local-lamb}
    \hat{\lambda}_Q^{(P)} = \begin{cases}
        2(-1)^b\hat{\lambda}_{R} & \ii PQ = (-1)^b R \\
        0 & \text{otherwise}
    \end{cases}
\end{align}
Since $\hat{\lambda}_Q^{(P)}$ is zero whenever $\supp(Q)$ does not contain $\supp(P)$, it follows that $\norm{\hat{\lambda}^{(P)}} \leq 2\ltnorm{\hat{H}}$.
Similarly, let $E^{(P)} = \sum_{Q \in \locals} \eps_Q^{(P)} Q$.
By assumption, $\norm{\eps^{(P)}} = \frac{1}{\sqrt{\dims}}\fnorm{E^{(P)}} < \eps/20^{\locality + 2}$.
Let $c_Q^{(P)} \deq \frac{1}{\dims}\tr(Z^\dagger PZQ)$.
By the assumption \eqref{eq:thm-structure},
\begin{align*}
    c_Q^{(P)} = \tfrac{1}{\dims}\tr(Z^\dagger PZQ)
    = \tfrac{1}{\dims}\tr\parens[\Big]{(P + [\ii\hat{H}, P] + E^{(P)})Q}
    = \iver{P = Q} + \hat{\lambda}_Q^{(P)} + \eps_Q^{(P)}.
\end{align*}
The first two steps of the procedure are to construct the oracles from \cref{algo:shadow} and \cref{algo:structure-query}.
The oracles give us certain kinds of access to the coefficient vector $\hat{\lambda}^{(P)} + \eps^{(P)}$.

We begin by bounding the number of calls to the oracles.
In order to do this, we bound the size of $\mathcal{Q}^{(P)}$, since the number of calls to the GL oracle and shadow oracle are $\bigO{\qubits \sum_{P \in \locals_1} \abs{\mathcal{Q}^{(P)}}}$ and $\bigO{\sum_{P \in \locals_1} \abs{\mathcal{Q}^{(P)}}}$, respectively.
Because of our guarantee of when the GL oracle outputs ``Fail'', \eqref{eq:structure-fail}, we know that the set $\mathcal{Q}^{(P)}$ is contained in the set
\begin{align*}
    \braces[\Big]{X \in \locals \setminus \{\id, P\} : \abs{\supp(X)} \leq \locality \text{ and }
    \sum_{Q \supseteq X} (c_Q^{(P)})^2 \geq \eps^2/400^{\locality+1}}.
\end{align*}
Note that we define $\mathcal{Q}^{(P)}$ such that it does not contain $P$.
Because of our bound on $\norm{\eps^{(P)}}$, $\norm{c^{(Q)}} = \norm{\hat{\lambda}^{(P)} + \eps^{(P)}}$ can only be large when $\hat{\lambda}^{(P)}$ is.
In other words, the set above, and consequently $\mathcal{Q}^{(P)}$, is contained in the set
\begin{align*} 
    \braces[\Big]{X \in \locals \setminus \{\id, P\} : \abs{\supp(X)} \leq \locality \text{ and }
    \sum_{\substack{Q \supseteq X}} (\hat{\lambda}_Q^{(P)})^2 \geq \eps^2/400^{\locality+2}}.
\end{align*}
Finally, we make one more modification to this set, noting that we can ``clip'' large coefficients of $\hat{\lambda}^{(P)}$ without changing the condition.
\begin{align} \label{eq:structure-QP-set}
    \mathcal{Q}^{(P)} \subseteq \braces[\Big]{X \in \locals \setminus \{\id, P\} : \abs{\supp(X)} \leq \locality \text{ and }
    \sum_{\substack{Q \supseteq X}} \min((\hat{\lambda}_Q^{(P)})^2, \eps^2/400^{\locality+2}) \geq \eps^2/400^{\locality+2}}.
\end{align}
The sum of all the expressions above is bounded: 
\begin{align*}
    & \sum_{X \in \locals_\locality \setminus \{\id, P\}} \sum_{Q \supseteq X} \min((\hat{\lambda}_Q^{(P)})^2, \eps^2/400^{\locality+2}) \\
    &\leq \sum_{Q \in \locals} 2^\locality \min((\hat{\lambda}_Q^{(P)})^2, \eps^2/400^{\locality+2}) \\
    &\leq 2^{\locality + 2}\ltnorm{\hat{H}^{\leq \eps/20^{\locality+2}}}^2,
\end{align*}
where the first step uses that $\hat{\lambda}_Q^{(P)}$ is only non-zero for $\locality$-local $Q$, and the second step uses \eqref{eq:structure-local-lamb} and \cref{def:eff-sparse}.
Consequently, by the pigeonhole principle, at most
\begin{align*}
    \frac{2^{\locality + 2}\ltnorm{\hat{H}^{\leq \eps/20^{\locality+2}}}^2}{\eps^2/400^{\locality+2}}
    \lesssim 2^\locality \sparse_{\eps / 20^\locality}
    \lesssim 800^\locality \sparse_\eps
\end{align*}
many $X$'s satisfy the criterion in \eqref{eq:structure-QP-set}, thus bounding $\abs{\mathcal{Q}^{(P)}}$.
We can conclude that the algorithm calls the GL oracle at most $\bigO{800^\locality \qubits^2 \sparse_{\eps}}$ times; this costs
\begin{align*}
    \bigO[\Big]{800^\locality \qubits^2 \sparse_{\eps} \cdot \frac{e^{8\locality}}{\eps^2}\log\frac{\qubits}{\delta}}
    = \bigO[\Big]{\qubits^2 \sparse_{\eps} \frac{e^{8\locality}}{\eps^2}\log\frac{\qubits}{\delta}}
\end{align*}
and dominates the running time.
The number of calls to the GL oracle also is an upper bound on the number of non-zero coefficients in the output $\tilde{\lambda}$.

As for correctness, for every $X \in \mathcal{Q}^{(P)}$, the corresponding coefficient in $\tilde{\lambda}$ must be correct, since its value was computed using the shadow oracle: it outputs an estimate $\mu$ such that
\begin{align*}
    \abs{\mu - (\hat{\lambda}_Q^{(P)} + \eps_Q^{(P)})}
    = \abs{\mu - \tfrac{1}{\dims}\tr(Z^\dagger PZQ)} < \eps.
\end{align*}
This implies that
\begin{align*}
    \abs{\mu - \hat{\lambda}_Q^{(P)}}
    \leq \abs{\mu - \hat{\lambda}_Q^{(P)} - \eps_Q^{(P)}} + \abs{\eps_Q^{(P)}}
    \leq \eps + \eps = 2\eps.
\end{align*}
For the $R \in \locals$ such that $R = (-1)^b \ii PQ$, $\hat{\lambda}_R = \frac{(-1)^b}{2} \hat{\lambda}_Q^{(P)}$, so the algorithm's estimate of $\tilde{\lambda}_R \gets \frac{(-1)^b}{2} \mu$ is $\eps$-correct.

So, the only remaining way the algorithm can fail is that some $Q \in \locals_{\locality}$ satisfies $\abs{\lambda_Q} \geq \eps$ but it was not included in any corresponding $\mathcal{Q}^{(P)}$.
This cannot happen: consider some $P \in \locals_1$ that anti-commutes with $Q$, so that the term appears in $\hat{\lambda}^{(P)}$, i.e.\ the corresponding $R = (-1)^{b+1}\ii PQ$ satisfies $\hat{\lambda}_R^{(P)} = 2(-1)^b\hat{\lambda}_R$ (following \eqref{eq:structure-local-lamb}).
Then $\mathcal{Q}$ will contain every $R'$ such that $R' \subseteq R$, since it satisfies the pass condition for the GL oracle in \eqref{eq:structure-success}:
\begin{align*}
    \sum_{Q \supseteq R'} (c_Q^{(P)})^2
    >  (c_R^{(P)})^2
    >  \eps^2
\end{align*}
and
\begin{equation*}
    \sum_{\substack{Q \supseteq R' \\ \abs{\supp(Q)} > \locality}} (c_Q^{(P)})^2
    = \sum_{\substack{Q \supseteq R' \\ \abs{\supp(Q)} > \locality}} (\eps_Q^{(P)})^2
    \leq \norm{\eps^{(P)}}^2 < \frac{\eps^2}{400^{\locality+2}}. \qedhere
\end{equation*}
\end{proof}

Now it remains to prove \cref{prop:structure-query}.  The algorithm for doing this is described in \cref{algo:structure-query}.  First we observe the following basic fact.

\begin{fact} \label{fact:pauli-eigv-helper}
    For $P \in \{\sigma_x, \sigma_y, \sigma_z\}$, $Q \in \{\id, \sigma_x, \sigma_y, \sigma_z\}$, and $v \in \{\pm1\}$, $\bra{P, v} Q \ket{P, v} = v^{\supp(Q)} \iver{Q \in \{\id, P\}}$.
    Generalizing, for $\qubits$-qubit Paulis, $P \in \{\sigma_x, \sigma_y, \sigma_z\}^{\otimes \qubits}$, $Q \in \locals$, and $v \in \{\pm 1\}^\qubits$,
    \[
        \bra{P, v} Q \ket{P, v} = v^{\supp(Q)} \iver{Q \subseteq P}.
    \]
\end{fact}

\begin{proof}[Proof of \cref{prop:structure-query}]
We analyze \cref{algo:structure-query}.
The basic idea is that, like the Goldreich--Levin algorithm with boolean functions, a GL query can be performed on input $X \in \locals_{\locality}$ by first fixing a random input state on the complement of $\supp(X)$, and then computing an estimator over random input states over $\supp(X)$.
This sampling procedure depends on $\supp(X)$, so to do this for all $X$ at once, we guess a partition of the input, which we call $T$, and then perform the procedure with respect to this partition.
See the technical overview for more intuition.

We begin by considering running time.
The pre-processing algorithm runs in $\bigO{\qubits pq} = \bigO{\frac{\qubits e^{8\locality}}{\gamma^2}\log\frac{\qubits}{\delta}}$ time.
The query algorithm runs in $\bigO{\locality pq} = \bigO{\frac{e^{8\locality}}{\gamma^2}\log\frac{\qubits}{\delta}}$ time, since every sample only needs to be queried on the support of $X$ and $P$.
The pre-processing algorithm applies the unknown unitary $pq$ times.

It remains to discuss correctness.
Fix some $\locality$-local Pauli $X$ and 1-local Pauli $P$; we want to show that our algorithm does not output the incorrect value for this choice of input.
First, we analyze a specific $\mu_k$: dropping the subscript for convenience, fix some $T \subset [\qubits]$, fix $A_{T}, v_{T}$, and consider only randomness over $[q]$ choices of $A_{\overline{T}}, v_{\overline{T}}$.

\begin{mdframed}
\begin{algorithm}[Answering GL queries]
    \label{algo:structure-query}\mbox{}
    \begin{description}
    \item[Input:] Black-box ability to apply an unknown $\qubits$-qubit unitary $Z$; locality parameter $\locality$; and error parameters $\gamma, \delta > 0$.
    \item[Output:] A oracle that, with probability $\geq 1-\delta$, can successfully respond to the following queries: given Paulis $X \in \locals_\locality$ and $P \in \locals_1$, output ``Pass'' or ``Fail'' with the guarantees described in \eqref{eq:structure-success} and \eqref{eq:structure-fail}.
    \item[Pre-processing (Generating the dataset):]\mbox{}
    \begin{algorithmic}[1]
        \State Let $p = \bigTheta{\locality \cdot 54^{\locality}\log(\qubits/\delta)}$ and $q = \bigTheta{\locality \cdot 54^{\locality}/\gamma^2}$;
        \ForAll{$k \in [p]$}
            \State Sample a random partition $T_k \subset [\qubits]$;
            \State Sample a random Pauli eigenvector $A_{T_k}, v_{T_k} \sim (\{\sigma_x, \sigma_y, \sigma_z\} \times \{+, -\})^{\abs{\supp(T_k)}}$;
            \ForAll{$\ell \in [q]$}
                \State Sample a random Pauli eigenvector $A_{\overline{T}_k}^{(\ell)}, v_{\overline{T}_k}^{(\ell)} \sim (\{\sigma_x, \sigma_y, \sigma_z\} \times \{+, -\})^{\abs{\supp(\overline{T_k})}}$;
            \EndFor
        \EndFor
        \ForAll{$(k, \ell) \in [p] \times [q]$}
            \State Sample a random Pauli basis $B \sim \{\sigma_x, \sigma_y, \sigma_z\}^\qubits$;
            \State Run the circuit $\circuit(A_{T_k}A_{\overline{T}_k}^{(\ell)}, v_{T_k}v_{\overline{T}_k}^{(\ell)}, B_{(k,\ell)},Z)$~\eqref{eq:the-circuit};
            \State Record the output $w \in \{\pm 1\}^n$;
        \EndFor
        \State Output the dataset $\dataset = \{(A_{(k,\ell)}, v_{(k,\ell)}, B_{(k,\ell)}, w_{(k,\ell)})\}_{k, \ell}$.
    \end{algorithmic}
    \item[Query subroutine:] Given input $X \in \locals_\locality$, $P \in \locals_1$, and the dataset $\dataset$, outputs estimate $\mu$. \mbox{}
    \begin{algorithmic}[1]
        \For{$k \in [p]$}
            \State Let $\dataset[k] = \{(A_{(k,\ell)}, v_{(k,\ell)}, B_{(k,\ell)}, w_{(k,\ell)})\}_{\ell \in [q]}$ be the part of the dataset associated with the partition $T_k$;
            \State For every $(A, v, B, w) \in \dataset[k]$, compute the associated estimator
            \begin{align*}
                \begin{cases}
                    (3v)^{\supp(X)}(3w)^{\supp(P)}
                        & \text{if } X \subseteq A ,\, P \subseteq B,\, T_k \cap \supp(X) = \varnothing \\
                    0 & \text{otherwise}
                \end{cases}
            \end{align*}
            and let $\mu_k$ be the average of this estimator over all the elements of $\dataset[k]$;
        \EndFor
        \State Output ``Pass'' if $\frac{1}{p}\sum_{k \in [p]} \iver{\abs{\mu_k} > 0.1\gamma/\sqrt{6}^\locality} > \frac{1}{3 \cdot 54^{\locality}}$, and ``Fail'' otherwise.
    \end{algorithmic}
    \end{description}
\end{algorithm}
\end{mdframed}

When $T \cap \supp(X) \neq \varnothing$, then $\mu_k = 0$; otherwise, $\mu_k$ can be written as
\begin{align*}
    \mu_k = \frac{1}{q} \sum_{(A,v,B,w) \in \dataset[k]} (3v)^{\supp(X)}(3w)^{\supp(P)} \iver{X \subseteq A ,\, P \subseteq B}.
\end{align*}
This is an average of $q$ estimators which are independent in the ``inner randomness'' of $A_{\overline{T}}$, $v_{\overline{T}}$, $B$, and $w$; and bounded in magnitude by $3^{\abs{\supp(X)} + \abs{\supp(P)}} \leq 3^{\locality+1}$.
For each such estimator, averaging only over the inner randomness,
\begin{align}
    & \expec{A_{\overline{T}}, v_{\overline{T}}, B, w}{(3v)^{\supp(X)}(3w)^{\supp(P)} \iver{X \subseteq A ,\, P \subseteq B}} \nonumber\\
    &= \frac{1}{3^{\abs{\supp(X)}}3^{\abs{\supp(P)}}}\expec{A_{\overline{T}}, v_{\overline{T}}, B, w}{(3v)^{\supp(X)}(3w)^{\supp(P)} \mid X \subseteq A ,\, P \subseteq B} \nonumber\\
    &= \expec{A_{\overline{T}}, v_{\overline{T}}, B, w}{v^{\supp(X)}w^{\supp(P)} \mid X \subseteq A ,\, P \subseteq B} \nonumber\\
    &= \expec[\Big]{A_{\overline{T}}, v_{\overline{T}}, B}{v^{\supp(X)} \tr((B_{\supp(P)} \otimes \id_{\overline{\supp(P)}})Z\ket{A,v}\bra{A,v}Z^\dagger) \,\Big|\, X \subseteq A ,\, P \subseteq B } \nonumber\\
    &= \expec[\Big]{A_{\overline{T}}, v_{\overline{T}}, B}{v^{\supp(X)} \tr(P Z\ket{A,v}\bra{A,v}Z^\dagger) \,\Big|\, X \subseteq A ,\, P \subseteq B } \nonumber\\
    &= \tr\parens[\Big]{P Z\expec{A_{\overline{T}}, v_{\overline{T}}}{v^{\supp(X)} \ket{A,v}\bra{A,v} \mid X \subseteq A}Z^\dagger} \nonumber\\
    &= \tr\parens[\Big]{Z^\dagger P Z \parens[\big]{\ket{A_T, v_T}\bra{A_T, v_T} \otimes X_{\overline{T}}}}2^{\abs{T}-\qubits} \nonumber
\intertext{Throughout the above computation, we use that $\overline{T}$ contains the support of $X$ as well as \cref{fact:circuit-helper}.
Writing $Z^\dagger PZ = \sum_{Q \in \locals} c_Q Q$ with $c_Q = \frac{1}{\dims}\tr(Z^\dagger PZQ)$, we can further clarify this expression:}
    &= \sum_{Q \in \locals} c_Q\tr\parens[\Big]{Q\parens[\big]{\ket{A_T, v_T}\bra{A_T, v_T} \otimes X_{\overline{T}}}}2^{\abs{T}-\qubits} \nonumber \\
    &= \sum_{Q \in \locals} c_Q v^{\supp(Q) \cap T} \iver{Q_T \subseteq A_T,\,Q_{\overline{T}} = X_{\overline{T}}}. \label{eq:structure-inner}
\end{align}
Let $\bar{\mu}_k = \E_{A_{\overline{T}}, v_{\overline{T}}, B, w}[\mu_k]$ be the expectation of $\mu_k$ over the inner randomness, which we computed above.
By Hoeffding's inequality, we know that
\begin{align}
    \Pr\bracks[\Big]{\abs{\mu_k - \bar{\mu}_k} \geq \frac{\gamma}{100\cdot \sqrt{6}^{\locality}}} \leq 2\exp\parens[\Big]{-\frac{q\gamma^2}{10000\cdot 6^\locality \cdot 2\cdot 3^{2(\locality+1)}}} \leq \frac{1}{100 \cdot 54^\locality}, \label{eq:mu-k-concentrates}
\end{align}
where the last line uses $q \gtrsim \locality \cdot 54^\locality / \gamma^2$.
This is the only place the value of $q$ is used.
Now, we consider how the ``outer randomness'' of $T$, $A_T$, and $v_T$ affects $\mu_k$.
In particular, we will bound moments associated to $\bar{\mu}_k$, which by \eqref{eq:structure-inner} takes the form
\begin{align*}
    \bar{\mu}_k = \begin{cases}
        0 & T \cap \supp(X) \neq \varnothing \\
        \sum_{Q \in \locals} c_Q v^{\supp(Q) \cap T} \iver{Q_{T} \subseteq A_T,\,Q_{\overline{T}} = X_{\overline{T}}} & \text{otherwise}
    \end{cases}
\end{align*}
We consider fixing $T$ such that $T \cap \supp(X) = \varnothing$, $A_{T}$, and interpreting $\bar{\mu}_k$ as a polynomial in $v_T \in \{\pm 1\}^{\abs{T}}$.
Then
\begin{align}
    \E_{v_T}[\bar{\mu}_k^2] = \sum_{Q \in \locals} c_Q^2 \iver{Q_{T} \subseteq A_T,\,Q_{\overline{T}} = X_{\overline{T}}}
    \leq \sum_{Q \supseteq X} c_Q^2 \label{eq:barmu}
\end{align}
We also consider decomposing $\bar{\mu}_k = \bar{\mu}_k^{\textup{(low)}} + \bar{\mu}_k^{\textup{(high)}} $ into the monomials of degree at most $\locality$ and the monomials of degree greater than $\locality$ (still viewed as a polynomial in $v_T)$.
Then, along the same lines,
\begin{align}
    \E_{v_T}[(\bar{\mu}_k^{\textup{(low)}})^2] &= \sum_{\substack{Q \in \locals \\ \abs{\supp(Q)} \leq \locality}} c_Q^2 \iver{Q_{T} \subseteq A_T,\,Q_{\overline{T}} = X_{\overline{T}}} \label{eq:barmu-low} \\
    \E_{v_T}[(\bar{\mu}_k^{\textup{(high)}})^2] &= \sum_{\substack{Q \in \locals \\ \abs{\supp(Q)} > \locality}} c_Q^2 \iver{Q_{T} \subseteq A_T,\,Q_{\overline{T}} = X_{\overline{T}}} \leq \sum_{\substack{Q \supseteq X \\ \abs{\supp(Q)} > \locality}} c_Q^2 \label{eq:barmu-high}
\end{align}
First, suppose that we are in the fail condition \eqref{eq:structure-fail}, and consider the probability that $\abs{\mu_k} > 0.1\gamma/\sqrt{6}^\locality$.
By \eqref{eq:barmu}, for every choice of $T$ and $A_T$, $\E_{v_T}[\bar{\mu}_k^2] \leq \gamma^2/400^{\locality+1}$.
By Chebyshev's inequality, $\Pr[\bar{\mu}_k^2 > \gamma^2/(400 \cdot 6^\locality)] \leq 0.1 \cdot 54^{-\locality}$.
Finally, taking
\begin{align*}
    \abs{\mu_k} \leq \abs{\bar{\mu}_k} + \abs{\mu_k - \bar{\mu}_k}
\end{align*}
and using \eqref{eq:mu-k-concentrates}, we can conclude that
\begin{align} \label{eq:structure-fail-final}
    \text{under the fail condition, } \Pr\bracks[\Big]{\abs{\mu_k} > \frac{0.1\gamma}{\sqrt{6}^\locality}} \leq 0.2 \cdot 54^{-\locality}.
\end{align}
Next, suppose that we are in the pass condition \eqref{eq:structure-success}.
Then, first, we can show
\begin{align}
    \Pr_{T, A_T}\bracks[\Big]{\E_{v_T}[(\bar{\mu}_k^{\textup{(low)}})^2] \geq \frac{\gamma^2}{6^\locality}}
    \geq \Pr_{T, A_T}\bracks[\Big]{\E_{v_T}[(\bar{\mu}_k^{\textup{(low)}})^2] \geq \frac{1}{6^\locality} \sum_{\substack{Q \supseteq X \\ \abs{\supp(Q)} \leq \locality}} c_Q^2}
    \geq \frac{1}{6^\locality}.
    \label{eq:structure-pass-1}
\end{align}
This holds since $\E_{v_T}[(\bar{\mu}_k^{\textup{(low)}})^2] \leq \sum_{\substack{Q \supseteq X; \abs{\supp(Q)} \leq \locality}} c_Q^2$ for every choice of $T$ and $A_T$ by \eqref{eq:barmu-low}, but
\begin{align*}
    & \expec{T, A_T}{\expec{v_T}{(\bar{\mu}_k^{\textup{(low)}})^2}} \nonumber\\
    &= \frac{1}{2^{\abs{\supp(X)}}}\expec{T, A_T, v_T}{(\bar{\mu}_k^{\textup{(low)}})^2 \mid T \cap \supp(X) = \varnothing} \\
    &= \frac{1}{2^{\abs{\supp(X)}}} \sum_{\substack{Q \\ \abs{\supp(Q)} \leq \locality}} c_Q^2 \expec[\Big]{T, A_T, v_T}{ \iver{Q_T \subseteq A_T,\,Q_{\overline{T}} = X_{\overline{T}}}\mid T \cap \supp(X) = \varnothing} \\
    &= \frac{1}{2^{\abs{\supp(X)}}} \sum_{\substack{Q \\ \abs{\supp(Q)} \leq \locality}} c_Q^2 \frac{1}{2^{\abs{\supp(Q)}}3^{\abs{\supp(Q) \setminus \supp(X)}}} \iver{Q \supseteq X} \\
    &= \sum_{\substack{Q \supseteq X \\ \abs{\supp(Q)} \leq \locality}} \frac{c_Q^2}{6^{\abs{\supp(Q)}}}
    \geq \frac{1}{6^\locality}\sum_{\substack{Q \supseteq X \\ \abs{\supp(Q)} \leq \locality}} c_Q^2.
\end{align*}
Suppose we have sampled such a $T$ and $A_T$ such that $\E_{v_T}[(\bar{\mu}_k^{\textup{(low)}})^2] \geq \gamma^2/6^\locality$.
Then we can lower bound the probability that $\bar{\mu}_k^{\textup{(low)}}$ is large.
\begin{align}
    \Pr_{v_T}[\abs{\bar{\mu}_k^{\textup{(low)}}} > 0.5\gamma/\sqrt{6}^\locality]
    = \Pr_{v_T}[(\bar{\mu}_k^{\textup{(low)}})^2 > 0.25\gamma^2/6^\locality]
    \geq \frac{9}{16}\frac{\expec{v_T}{(\bar{\mu}_k^{\textup{(low)}})^2}^2}{\expec{v_T}{(\bar{\mu}_k^{\textup{(low)}})^4}}
    \geq \frac{9}{16} 9^{-\locality}
    \label{eq:structure-pass-2}
\end{align}
Above, the first inequality is the Paley--Zygmund inequality, $\Pr[Z > t \E[Z]] \geq (1-t)^2\frac{\E[Z]^2}{\E[Z^2]}$, and the second inequality is the ``Bonami lemma'' consequence of hypercontractivity: $\E_x[f(x)^4] \leq 9^{\deg(f)} \E_x[f(x)^2]^2$, where $f$ is a multivariate polynomial and the expectation is over the uniform distribution on $x \sim \{\pm 1\}^k$ (see Chapter 9 of \cite{odonnell14}).

Altogether, combining \eqref{eq:structure-pass-1} with \eqref{eq:structure-pass-2}, we have shown that $\abs{\bar{\mu}_k^{\textup{(low)}}} > 0.5\gamma/\sqrt{6}^\locality$ with probability at least $\frac{9}{16}54^{-\locality}$ (over all the outer randomness).
Note that
\begin{align*}
    \mu_k &= \bar{\mu}_k^{\textup{(low)}} + \bar{\mu}_k^{\textup{(high)}} + (\mu_k - \bar{\mu}_k); \\
    \abs{\mu_k} &\geq \abs{\bar{\mu}_k^{\textup{(low)}}} - \abs{\bar{\mu}_k^{\textup{(high)}}} - \abs{\mu_k - \bar{\mu}_k}.
\end{align*}
By \eqref{eq:barmu-high}, Chebyshev's inequality, and the pass condition \eqref{eq:structure-success}, $\abs{\bar{\mu}_k^{\textup{(high)}}}$ exceeds $0.2 \gamma / \sqrt{6}^\locality \geq \sqrt{32\cdot 54^\locality(\sum_{Q \colon \abs{\supp(Q)} > \locality} c_Q^2)}$ with probability at most $\frac{1}{32}54^{-\locality}$.
By \eqref{eq:mu-k-concentrates}, $\abs{\mu - \bar{\mu}}$ exceeds $0.2 \gamma/\sqrt{6}^\locality$ with probability at most $\frac{1}{32}54^{-\locality}$.
Combining these using triangle inequality and union bound, over all randomness, we conclude that
\begin{align} \label{eq:structure-pass-final}
    \text{under the pass condition, }\Pr\bracks[\Big]{\abs{\mu_k} > \frac{0.1\gamma}{\sqrt{6}^\locality}} \geq \frac{1}{2} 54^{-\locality}.
\end{align}
Under the fail condition, the probability that $\abs{\mu_k} > \frac{0.1\gamma}{\sqrt{6}^\locality}$ is at most $0.2 \cdot 54^{-\locality}$.
To distinguish the pass and fail conditions, by a Chernoff bound, $p \gtrsim 54^\locality \log\frac{n^{\locality+1}}{\delta}$ copies of $\mu_k$ suffice to distinguish these two conditions with probability at least $1 - \frac{\delta}{n^{\locality + 1}}$.
By a union bound, for this choice of $p$, the algorithm will correctly distinguish the pass and fail conditions over all $X \in \locals_\locality$ and $P \in \locals$ with probability $\geq 1-\delta$.
\end{proof}

Finally, we remark that \cref{prop:structure-query} generalizes beyond expressions of the form $Z^\dagger P Z$.
We can prove the following.
This result is not used in later sections.
\begin{lemma} \label{lem:gl-general}
    Let $O \in \mathbb{C}^{\dims \times \dims}$ be an unknown observable with $\norm{O} \leq 1$ and a Pauli decomposition of $O = \sum_{Q \in \locals} c_Q Q$.
    Suppose we can efficiently apply the POVM $\{\frac{\id + O}{2}, \frac{\id - O}{2}\}$, and suppose we are given a natural number parameter $\locality$.
    Then, with $\bigO{82^\locality \log(\qubits/\delta) / \gamma^3}$ applications of the POVM and $\bigO{82^\locality \qubits \log(\qubits/\delta) / \gamma^3}$ additional pre-processing time, we can output a data structure, which, with probability $\geq 1-\delta$, can correctly respond to the following type of query in $\bigO{82^\locality \log(\qubits/\delta) / \gamma^3}$ time on a classical computer: given $X \in \locals_{\locality}$, output an estimate of 
    \begin{align*}
        \sum_{\substack{Q \in \locals \\ Q \supseteq X}} \frac{\abs{c_Q}^2}{6^{\abs{\supp(Q)}}}
    \end{align*}
    to $\gamma$ error.
\end{lemma}
\begin{proof}
We first describe the algorithm.
We assume understanding of \cref{algo:structure-query}.
First, we generate a dataset $\dataset = \{(A_{(k, \ell)}, v_{(k, \ell)}, b_{(k,\ell)})\}_{k \in [p], \ell \in [q]}$, following Lines 2 through 6 of \cref{algo:structure-query} with
\begin{align*}
    p &= \bigTheta{9^{\locality}\log(\qubits^\locality/\delta)/\gamma^2} \\
    q &= \bigTheta{9^{\locality}/\gamma}.
\end{align*}
This generates the Pauli eigenvectors $A_{(k, \ell)}, v_{(k, \ell)}$ and associated partitions $T_k$.
We can then prepare these eigenvectors and apply the POVM $\{\frac{\id + O}{2}, \frac{\id - O}{2}\}$ to get a random variable $b_{(k, \ell)} \in \{\pm 1\}$ with expectation $\tr(O \ket{A_{(k, \ell)}, v_{(k, \ell)}}\bra{A_{(k, \ell)}, v_{(k, \ell)}})$.
This dataset $\dataset$ is our data structure.
It can be generated with $pq = \bigTheta{81^\locality \log(\qubits^\locality/\delta) / \gamma^3}$ measurements of $O$, and $\bigO{\qubits pq}$ quantum gates and classical pre-processing.

Now, suppose we are given an $X \in \locals_{\locality}$.
Then our estimator is as follows: for every $k \in [p]$, compute the estimator
\begin{align*}
    \mu_k = \begin{cases}
        \frac{1}{q}\sum_{\ell = 1}^{q} (3v_{(k, \ell)})^{\supp(X)}b_{(k, \ell)} \iver{X \subseteq A_{(k, \ell)}} & \text{if } T_k \cap \supp(X) = \varnothing \\
        0 & \text{otherwise}
    \end{cases}
\end{align*}
and output the average, $\mu = \frac{1}{p}\sum_{k=1}^p \mu_k^2$.
Computing this estimator costs $\bigO{\abs{\supp(X)}}$ per sample in the dataset, giving a total running time of $\bigO{pq\abs{\supp(X)}}$.
To analyze this estimator, we consider a $k$ for which $T_k \cap \supp(X) = \varnothing$ and consider a $(A, v, b) \in \dataset[k]$.
Then, following \eqref{eq:structure-inner},
\begin{align}
    & \expec{A_{\overline{T}}, v_{\overline{T}}, b}{(3v)^{\supp(X)} b \iver{X \subseteq A}} \nonumber\\
    &= \expec{A_{\overline{T}}, v_{\overline{T}}, b}{v^{\supp(X)} b \mid X \subseteq A} \nonumber\\
    &= \expec{A_{\overline{T}}, v_{\overline{T}}}{v^{\supp(X)} \tr(O \ket{A, v}\bra{A, v}) \mid X \subseteq A} \nonumber\\
    &= \tr\parens[\Big]{O \expec{A_{\overline{T}}, v_{\overline{T}}}{v^{\supp(X)} \ket{A, v}\bra{A, v}\mid X \subseteq A}} \nonumber\\
    &= \tr\parens[\Big]{O \ket{A_T, v_T}\bra{A_T, v_T} \otimes X_{\overline{T}}} 2^{-\abs{\overline{T}}}\nonumber
\intertext{
    The last equality above uses \eqref{eq:pauli-eigenvector-sum}.
    We now use the Pauli expansion of $O$.
}
    &= \sum_{Q \in \locals} c_Q \tr\parens[\Big]{Q \ket{A_T, v_T}\bra{A_T, v_T} \otimes X_{\overline{T}}} 2^{-\abs{\overline{T}}} \nonumber\\
    &= \sum_{Q \in \locals} c_Q v^{\supp(Q) \cap T} \iver{Q_T \subseteq A_T,\, Q_{\overline{T}} = X_{\overline{T}}} \label{eq:gen-structure-inner}
\end{align}
Let $\bar{\mu}_k$ denote the above expression, the expectation over the inner randomness.
With $q$ samples over the inner randomness, the average will in fact converge to its expectation, so that $\mu_k \approx \bar{\mu}_k$.
In this case, the final estimator is approximately $\frac{1}{s}\sum_{k=1}^s \bar{\mu}_k^2$, where the $\bar{\mu}_k$'s are independent and identically distributed.
So, we consider a single $\bar{\mu}_k$, and compute its second moment.
Since $\bar{\mu}_k$ already averages over the inner randomness, this expectation is over the outer randomness of $T$, $A_{T}$, and $v_T$.
\begin{align*}
    \E[\bar{\mu}_k^2]
    &= \frac{1}{2^{\abs{\supp(X)}}}\expec{T, A_T, v_T}{\bar{\mu}_k^2 \mid T \cap \supp(X) = \varnothing} \\
    &= \frac{1}{2^{\abs{\supp(X)}}}\sum_{Q, R \in \locals} c_Q c_R \E_{T, A_T, v_T}\Big[ v^{\supp(Q) \cap T}v^{\supp(R) \cap T} \iver{Q_T \subseteq A_T,\,Q_{\overline{T}} = X_{\overline{T}}}\\
    & \hspace{16em} \iver{R_T \subseteq A_T,\,R_{\overline{T}} = X_{\overline{T}}} \;\Big|\; T \cap \supp(X) = \varnothing \Big] \\
    &= \frac{1}{2^{\abs{\supp(X)}}} \sum_{Q} c_Q^2 \expec[\Big]{T, A_T}{ \iver{Q_T \subseteq A_T,\,Q_{\overline{T}} = X_{\overline{T}}}\mid T \cap \supp(X) = \varnothing} \\
    &= \frac{1}{2^{\abs{\supp(X)}}} \sum_{Q} c_Q^2 \frac{1}{2^{\abs{\supp(Q)}}3^{\abs{\supp(Q) \setminus \supp(X)}}} \iver{Q \supseteq X} \\
    &= \sum_{Q \supseteq X} \frac{c_Q^2}{6^{\abs{\supp(Q)}}},
\end{align*}
using \eqref{eq:gen-structure-inner} and that the distributions over $T$, $A_T$, and $v_T$ are uniform.
Altogether, we have shown that, at least morally, $\mu$ is an estimator of $\sum_{Q \supseteq X} \frac{c_Q^2}{6^{\abs{\supp(Q)}}}$.
We now show that, for our choices of $p$ and $q$, $\mu$ is well-concentrated.
Let us now bound the deviation.
\begin{align}
    \mu - \sum_{Q \supseteq X} \frac{c_Q^2}{6^{\abs{\supp(Q)}}}
    &= \frac{1}{p} \sum_{k=1}^p (\mu_k^2 - \E[\bar{\mu}_k^2 ]) \nonumber\\
    &= \frac{1}{p} \sum_{k=1}^p \parens[\Big]{(\mu_k^2 - \E[\mu_k^2]) + (\E[\mu_k^2 - \bar{\mu}_k^2])} \nonumber\\
    &= \frac{1}{p} \sum_{k=1}^p \parens[\Big]{\mu_k^2 - \E[\mu_k^2]} + \E[\mu_1^2 - \bar{\mu}_1^2] \label{eq:structure-conc1}
\end{align}
Above, all the expectations are over all the randomness, though $\bar{\mu}_k$ is determined by only the outer randomness.
Since $q \geq 3^{2\locality + 2}\frac{2}{\gamma}$,
\begin{align}
    \E[\mu_1^2 - \bar{\mu}_1^2] = \E[(\mu_1 - \bar{\mu}_1)^2] \leq \frac{1}{q}3^{2\locality + 2} \leq \frac{\gamma}{2}, \label{eq:structure-conc2}
\end{align}
so that
\begin{align*}
    \Pr\bracks[\Big]{\abs[\Big]{\mu - \sum_{Q \supseteq X} \frac{c_Q^2}{6^{\abs{\supp(Q)}}}} \geq \gamma}
    &= \Pr\bracks[\Big]{\abs[\Big]{\frac{1}{p} \sum_{k=1}^p \parens[\Big]{\mu_k^2 - \E[\mu_k^2]} + \E[\mu_1^2 - \bar{\mu}_1^2]} \geq \gamma} \tag*{by \eqref{eq:structure-conc1}}\\
    &\leq \Pr\bracks[\Big]{\abs[\Big]{\frac{1}{p} \sum_{k=1}^p \parens[\Big]{\mu_k^2 - \E[\mu_k^2]}} \geq \frac{\gamma}{2}} \tag*{by \eqref{eq:structure-conc2}} \\
    &\leq 2\exp\parens[\Big]{-\frac{p\gamma^2}{8\cdot 3^{2\abs{\supp(X)}}}}. \tag*{by Hoeffding's inequality}
\end{align*}
Since $p \geq \frac{3^{2\locality+2}}{\gamma^2} \log\frac{2(3\qubits)^{\locality}}{\delta}$, this probability is at most $\frac{\delta}{(3\qubits)^\locality}$, and so with probability at most $\delta$, $\mu$ is correct for all choices of $X \in \locals_{\locality}$.
\end{proof}
%!TEX root = main.tex

\section{Algorithm}

Now we can state our main theorem for learning an unknown Hamiltonian from its time evolution.

\begin{theorem} \label{thm:main}
    Let $H = \sum_{a = 1}^m \lambda_a E_a$ be some $\locality$-local Hamiltonian where the terms $E_a \in \locals_\locality$ are known but the coefficients $\lambda_a$ are unknown.\footnote{
        In the setting of structure learning, we think of the $E_a$'s as the set of terms which could possibly appear, e.g.\ all Paulis with support size at most $\locality$.
        With this choice of terms, most of the corresponding $\lambda_a$'s will be zero.
    }
    Let $\eps, \delta \in (0, 1)$ be error parameters.
    Suppose we know a bound on local one-norm $\degree$ (\cref{def:local-norm}), $\lonorm{H} \leq \degree$, and a bound on effective sparsity $\sparse$ (\cref{def:eff-sparse}), $\sparse_\eps \leq \sparse$. 
    Then we can run \cref{algo:strap}, and it satisfies the following properties for $c$ a sufficiently large universal constant:
    \begin{enumerate}
        \item \textup{(Accuracy)} with probability $1 - \delta$, it returns estimates $\wh{\lambda}$ such that $\norm{\lambda - \wh{\lambda}}_{\infty} \leq \eps$;
        \item \textup{(Evolution time)} it applies $e^{-\ii Ht}$ for a total evolution time of $\bigO{\frac{\locality^{c \locality}\sparse }{\eps}\log \frac{n}{\delta}}$;
        \item \textup{(Time resolution)} it only ever applies $e^{-\ii Ht}$ for $t \geq (\max(1, \frac{\degree}{\sqrt{\sparse}}) \sparse \locality^{c \locality})^{-1}$;
        \item \textup{(Classical overhead)} it has a classical overhead of $\bigO{\terms \sparse^2 \locality^{c \locality}\log\frac{1}{\eps}\log(\frac{\qubits}{\delta}\log\frac{1}{\eps})}$ (naive version) or
        $\bigO{\qubits^2 \sparse^3 \locality^{c \locality} \log\frac{1}{\eps}\log\parens[\big]{\frac{\qubits}{\delta}\log\frac{1}{\eps}}}$ (structure learning version);
        \item \textup{(Number of experiments)} it runs $\bigO{\locality^{c \locality}\sparse^2\log\frac{1}{\eps}\log(\frac{\qubits}{\delta}\log\frac{1}{\eps})}$ quantum circuits, all of which take the form given in \cref{fig:the-circuit}.
    \end{enumerate}
\end{theorem}

As discussed in \cref{rmk:sparsity}, the output is given in the form of a list of non-zero coefficients $\wh{\lambda}$ and we will show that this list has at most $\bigO{ne^{7\locality}\sparse}$ elements.
We can further guarantee that the support of $\wh{\lambda}$ is contained in the support of $\lambda$ by rounding all terms of $\wh{\lambda}$ smaller than $\eps$ to zero; this worsens the accuracy guarantee by only a constant factor.
This gives the version of the accuracy guarantee in \cref{thm:main-informal}.

\begin{remark}[Access to time evolutions]
    For simplicity, we assume that we are given the ability to apply $e^{-\ii Ht}$ for any $t \geq 0$.
    However, we only need it for one choice of $t$, the one given in \cref{algo:strap}.
    Further, our algorithm works for any $t' < t$: instead of applying $e^{-\ii Ht}$, one can apply $e^{-\ii H t'}$ $\floor{t/t'}$ times in a row, only affecting the analysis by constant factors.
\end{remark}

\begin{remark}[Dependence on locality]
    We did not attempt to optimize the dependence on $\locality$ in the above theorem, and instead aimed to give a statement in the broadest generality possible.  \cref{lem:main} is the only source of the $\locality^\locality$ dependence.  As discussed in \cref{rem:locality}, we believe the dependence on $\locality^\locality$ can be improved to $\exp(\locality)$ if we define $\sparse_{\eps}$ to bound the $L^1$ effective sparsity instead of the $L^2$ effective sparsity and throughout the algorithm, we measure the closeness of our intermediate estimates in $L^1$-local norm instead of $L^2$-local norm.
\end{remark}

\begin{remark}[Tolerance to SPAM error]
    Our algorithm still works, even when the measurement probabilities of every quantum circuit is perturbed by error, e.g.\ error in state preparation and measurement.
    In \cref{algo:strap}, because of the requirements of \cref{algo:shadow} and \cref{algo:structure-query}, we only require estimates to be correct to $\eta_j s_j t e^{-3\locality} \gtrsim \sparse^2 \locality^{c\locality}$ in iteration $j$.
    This is the error we can tolerate.
\end{remark}

\subsection{Proof of the main theorem}

We prove the main theorem by bootstrapping \cref{thm:structure} to iteratively refine our estimate for the coefficients.

\begin{mdframed}
\begin{algorithm}[Learning a Hamiltonian from real-time evolution]
    \label{algo:strap}\mbox{}
    \begin{description}
    \item[Input:] Accuracy and failure probability parameters $\eps, \delta \in (0, 1)$, oracle access to $e^{-\ii Ht}$
    \item[Input:]  Local $L^1$ bound $\degree$, effective sparsity bound $\sparse$, locality parameter $\locality$
    \item[Output:] Estimate $\wh{\lambda}$ for $\lambda$ such that with probability at least $1-\delta$, $\norm{\lambda - \wh{\lambda}}_{\infty} \leq \eps$.
    \item[Operation:]\mbox{}
    \begin{algorithmic}[1]
        \State Let $T \leftarrow \floor{\log_2(1/\eps)}$;
        \State Let $t \gets (500 \max(1, \frac{\degree}{\sqrt{\sparse}}) \sparse \locality^{2C \locality})^{-1}$ for $C$ a sufficiently large constant;
        \State Let $\lambda^{(0)} \gets (0,0,\dots,0)$;
        \For{$j$ from $0$ to $T$}
            \Comment{We will maintain a $\lambda^{(j)}$ such that $\infnorm{\lambda^{(j)} - \lambda} \leq 2^{-j}$}
            \State Let $\delta_j \leftarrow \delta/(2(T+1-j))^2$;
            \Comment{$\delta_j$ is the failure probability of iteration $j$}
            \State Let $\eta_j \leftarrow 2^{-j}$;
            \Comment{$\eta_j$ is the desired error in iteration $j$}
            \State Let $s_j \gets \floor{\frac{1}{\eta_j}\max(1, \frac{\degree}{\sqrt{\sparse}})}$;
            \State Let $H_j = \sum_{a} \lambda_a^{(j)} E_a$;
            \For{$a$ from $1$ to $\terms$} \label{line:shadow-version}
                \State Let $P_a \in \locals_1$ be some 1-local Pauli that anticommutes with $E_a$;
                \State Let $Q_a \leftarrow \frac{\ii}{2}[P_a,E_a] = \ii P_a E_a$;
            \EndFor
            \State Use a shadow oracle (\cref{lem:shadows}) to produce estimates $\wh{\lambda}^{(j+1)} \in [-1,1]^\terms$ to $\mu^{(j)} \in [-1,1]^\terms$, where
            \[
                \mu_a^{(j)} = \frac{1}{2\dims}\tr\parens*{P_a \parens[\big]{e^{-\ii H t} e^{\ii H_j t}}^{s_j} Q_a \parens[\big]{e^{-\ii H_j t} e^{\ii H t}}^{s_j}},
            \]
            satisfying $\infnorm{\wh{\lambda}^{(j+1)} - \mu^{(j)}} \leq \frac{\eta_j s_jt}{20}$ with probability $\geq 1-\delta_j$;
            \State Set $\lambda^{(j+1)} \gets \lambda^{(j)} + \frac{1}{s_jt}\wh{\lambda}^{(j+1)}$;
            \ForAll{coefficients $\lambda_a^{(j+1)}$ with $|\lambda_a^{(j+1)}| \leq \frac{\eta_{j}}{4}$}
                \State Set $\lambda_a^{(j+1)} \gets 0$;
                \Comment{Round small coefficients to zero to reduce $\lonorm{\lambda^{(j+1)}}$}
            \EndFor
        \EndFor
        \State \Output $\wh{\lambda} = \lambda^{(T+1)}$.
    \end{algorithmic}
    \item[Modification for structure learning:] Instead of \cref{line:shadow-version} and the preceding for loop, run \cref{algo:structure} with parameters 
            \begin{align*}
                Z \leftarrow \parens[\big]{e^{-\ii H t} e^{\ii H_j t}}^{s_j}, \locality, \eps \leftarrow \frac{\eta_{j} s_j t}{10}, \delta \leftarrow \delta_j
            \end{align*}
        and let $\wh{\lambda}^{(j+1)}$ be the resulting output.
    \end{description}
\end{algorithm}
\end{mdframed}

\begin{proof}
We analyze \cref{algo:strap}.
We maintain the invariants that, after every iteration,
\begin{enumerate}[label=(\alph*)]
    \item $\norm{\lambda^{(j)} - \lambda}_{\infty} \leq \eta_{j}$;
    \item $\lonorm{H_{j}} \leq 2\degree$;
    \item $\ltnorm{H_{j} - H} \leq 4 \eta_{j} \sqrt{\sparse}$;
\end{enumerate}
where $H_{j} = \sum_a \lambda_a^{(j)} E_a$.
Note that (b) and (c) hold for $j = 0$ by assumption, with (c) holding because $\ltnorm{H} = \sqrt{\sparse_1} \leq \sqrt{\sparse}$.
To show that the invariants hold for larger $j$, we observe that the following statement about the $j$th iteration implies all three of them for $j+1$.
\begin{align} \label{eq:boot-master}
    \infnorm[\Big]{\frac{1}{s_jt}\wh{\lambda}^{(j+1)} - (\lambda - \lambda^{(j)})} \leq \frac{\eta_j}{10} = \frac{\eta_{j+1}}{5}.
\end{align}
This equation implies (a), $\infnorm{\lambda^{(j+1)} - \lambda} \leq \eta_{j+1}$.
Specifically, $\lambda^{(j+1)}$ is first set to $\lambda^{(j)} + \frac{1}{s_jt}\wh{\lambda}^{(j+1)}$, so that $\infnorm{\lambda^{(j+1)} - \lambda} \leq \eta_{j+1}/5$ before rounding, and the rounding only changes any coefficient by $\eta_j / 4 = \eta_{j+1}/2$.
It also implies (b), $\lonorm{H_{j+1}} \leq 2\degree$, since the rounding step ensures $\abs{\lambda_a^{(j+1)}} \leq 2\abs{\lambda_a} $ entrywise: if $\lambda_a^{(j+1)}$ was not rounded, then $\abs{\lambda_a} \geq \abs{\lambda_a^{(j+1)}} - \eta_{j+1}/5 \geq \eta_{j+1}(1/2 - 1/5)$.
Consequently,
\begin{align*}
    \abs{\lambda_a^{(j+1)}} \leq \abs{\lambda_a} + \abs{\lambda_a^{(j+1)} - \lambda_a} \leq \abs{\lambda_a} + \eta_{j+1}/5 \leq 2\abs{\lambda_a}.
\end{align*}
Finally, \eqref{eq:boot-master} implies (c).
We observe that the coefficients of $\frac{1}{\eta_{j+1}}(H_{j+1} - H)$ are smaller in magnitude than the corresponding coefficients of $\frac{4}{\eps}H^{\leq \eps/4}$:
\[
    \frac{1}{\eta_{j+1}}\abs{\lambda_a^{(j+1)} - \lambda_a} \leq \min(1, 4\abs{\lambda_a}/\eps).
\]
The bound of 1 follows from (a) and the other bound uses $\eps/2 \leq \eta_{j+1}$ and $\abs{\lambda_a^{(j+1)}} \leq 2\abs{\lambda_a}$.
Thus,
\[
\ltnorm{H_{j+1} - H} \leq \frac{4\eta_{j+1}}{\eps} \ltnorm{H^{\leq \eps/4}} \leq \eta_{j+1}\sqrt{\sparse_{\eps/4}} \leq 4\eta_{j+1}\sqrt{\sparse}.
\]
We have shown that we can complete the proof, provided \eqref{eq:boot-master} holds.
It remains to analyze the $\wh{\lambda}^{(j+1)}$ we receive from \cref{line:shadow-version} and from the structure learning modification, to show \eqref{eq:boot-master} from the inductive hypotheses (in fact we will only need to use (b) and (c) from the inductive hypotheses).

For both, we use \cref{lem:main}.
We first verify that the assumptions of \cref{lem:main} are satisfied for $H \gets H$, $H_0 \leftarrow H_j$, $P \in \locals_1$, $\degree \leftarrow 2\degree$, and $\eta \leftarrow 4 \eta_j \sqrt{\sparse}$.
Invariants (b), $\lonorm{H_j} \leq 2\degree$, and (c), $\ltnorm{H - H_j} \leq  4\eta_j \sqrt{\sparse}$, imply the desired norm assumptions.
Also, by definition,
\[
    t = \frac{1}{500 \max(1, \frac{\degree}{\sqrt{\sparse}}) \sparse \locality^{2C \locality}} < \frac{1}{500\degree \locality^C} \text{ and }
    s_jt \leq \frac{1}{500 \sparse \locality^{2C \locality} \eta_j} \leq \frac{1}{ 500 \sqrt{\sparse} \eta_j},
\]
so all assumptions are satisfied, and thus we conclude that for $P \in P_1$,
\[
    \parens[\big]{e^{-\ii H_j t} e^{\ii H t}}^{s_j} P \parens[\big]{e^{-\ii H t} e^{\ii H_j t}}^{s_j}
    = P + \ii s_j t[H -H_j, P] + E,
\]
where 
\begin{align} \label{eq:main-e-bd}
    \fnorm{E} &\leq (8\eta_j \sqrt{\sparse} s_j t^2 \degree + 16\eta_j^2 \sparse s_j^2t^2)\locality^{C \locality} \fnorm{P}
    \leq (\eta_j s_jt)\locality^{-C \locality} \fnorm{P}/20 \,.
\end{align}
First, we consider the un-modified algorithm.
In this version, the $\mu_a^{(j)}$ defined satisfies
\begin{align*}
    \mu_a^{(j)} &= \frac{1}{2\dims}\tr\parens*{P_a \parens[\big]{e^{-\ii H t} e^{\ii H_j t}}^{s_j} Q_a \parens[\big]{e^{-\ii H_j t} e^{\ii H t}}^{s_j}}, \\
    &= \frac{1}{2\dims}\tr\parens*{Q_a (P_a + \ii s_j t[H-H_j, P_a] + E)} \\
    &= \frac{1}{2\dims}\parens[\Big]{\tr\parens{Q_a P_a} + \ii s_j t \sum_{b=1}^\terms(\lambda_b - \lambda_b^{(j)})\tr\parens{Q_a[E_b, P_a]} + \tr\parens{Q_a E}} \\
    &= s_jt(\lambda_a - \lambda_a^{(j)}) + \frac{1}{\dims}\tr\parens{Q_a E},
\end{align*}
where the last line uses that $Q_a = \ii E_a P_a$, so that $\frac{1}{2\dims}\tr(\ii Q_a [E_b, P_a])$ is 1 when $a = b$ and 0 otherwise.
This gives us \eqref{eq:boot-master}:
\begin{align*}
    \abs[\Big]{\frac{1}{s_j t}\wh{\lambda}_a^{j+1} - (\lambda_a - \lambda_a^{(j)})}
    \leq \frac{\eta_j}{20} + \abs[\Big]{\frac{1}{s_jt}\mu_a^{(j)} - (\lambda - \lambda^{(j)})}
    \leq \frac{\eta_j}{20} + \frac{1}{s_jt\dims}\fnorm{Q_a}\fnorm{E}
    \leq \frac{\eta_j}{10}.
\end{align*}
As for the modified, structure learning algorithm, \eqref{eq:boot-master} holds directly by the consequence of \cref{thm:structure}, where $\hat{H} \gets s_jt(H - H_j)$ and we take $\eps \leftarrow 0.1\eta_{j+1}s_jt$ and $\delta \gets \delta_j$.
The assumptions for \cref{thm:structure} are satisfied by the above application of \cref{lem:main}, in particular \eqref{eq:main-e-bd}, so it returns a $\wh{\lambda}^{(j+1)}$ such that, with probability $1 - \delta_j$,
\[
\infnorm[\Big]{\frac{1}{s_jt}\wh{\lambda}^{(j+1)} - (\lambda - \lambda^{(j)})} \leq \frac{0.1\eta_{j+1}s_jt}{s_jt} = \frac{\eta_j}{10} \,,
\]
as desired.
This shows that for both versions of the algorithm, the output is correct, in that the output $\wh{\lambda}$ satisfies $\infnorm{\wh{\lambda} - \lambda} \leq \eta_{T+1} \leq \eps$, provided the shadow oracle or structure learning oracle never fails.
The total probability of failure is $\sum_{j=0}^T \delta_j \leq \delta$, by union bound.

Finally, we analyze the running time and resources used by the algorithm.
In both versions, the unknown Hamiltonian is always applied for $t$ time, making the time resolution $t = (500 \max(1, \frac{\degree}{\sqrt{\sparse}}) \sparse \locality^{2C \locality})^{-1}$.
Further, quantum resources are only consumed by the shadow oracle and structure learning algorithm, respectively, and the classical running time is also dominated by the classical costs of these subroutines.

The original version of the algorithm runs \cref{lem:shadows} with $\eps \gets \frac{\eta_j s_jt}{20}$, $\delta \gets \delta_j$, $\locality \gets 1$, and $\locality' \gets \locality$, requiring
\begin{align} \label{eq:alg-application-count}
    \bigO[\Big]{
        \frac{e^{3\locality}}{(\eta_j s_jt)^2}\log\frac{\qubits}{\delta_j}
    } = \bigO[\Big]{
        \sparse^2 \locality^{5C\locality} \log\frac{\qubits(T+1-j)}{\delta}
    }
\end{align}
applications of $Z \gets (e^{-\ii Ht}e^{\ii H_jt})^{s_j}$, which requires $ts_j$ evolution time of the unknown Hamiltonian $H$, interleaved with evolution by the known Hamiltonian $H_j$.
Summing over $j$, this gives a total evolution time of
\begin{align*}
    \bigO[\Big]{
        \sum_{j=0}^T s_j t\frac{e^{3\locality}}{(\eta_j s_jt)^2}\log\frac{\qubits}{\delta_j}
    } = \bigO[\Big]{
        \sum_{j=0}^T \frac{\sparse \locality^{3C\locality}}{\eta_j} \log\frac{\qubits}{\delta_j}
    } = \bigO[\Big]{
        \frac{\sparse \locality^{3C\locality}}{\eps} \log\frac{\qubits}{\delta}
    }.
\end{align*}
The number of experiments in iteration $j$ is \eqref{eq:alg-application-count}, since $Z$ is applied once per experiment.
Summing over $j$, this gives a bound of $\bigO{\locality^{5 C\locality}\sparse^2\log\frac{1}{\eps}\log(\frac{\qubits}{\delta}\log\frac{1}{\eps})}$ on the number of experiments.
The classical computation time is dominated by the cost of the queries, of which there are $\terms$ per iteration, with each one costing $\locality$ times the number of experiments.
This gives a running time bound of $\bigO{\terms \sparse^2 \locality^{6 C\locality}\log\frac{1}{\eps}\log(\frac{\qubits}{\delta}\log\frac{1}{\eps})}$.

With the structure learning modification, the quantum resources are almost entirely analogous.
By \cref{thm:structure}, the number of applications of $Z$, which is still $(e^{-\ii Ht}e^{\ii H_jt})^{s_j}$, is
\begin{align*}
    \bigO[\Big]{
        \frac{e^{8\locality}}{(\eta_j s_jt)^2}\log\frac{\qubits}{\delta_j}
    },
\end{align*}
giving identical bounds of $\bigO{\frac{\sparse \locality^{3C\locality}}{\eps} \log\frac{\qubits}{\delta}}$ and $\bigO{\locality^{5 C\locality}\sparse^2\log\frac{1}{\eps}\log(\frac{\qubits}{\delta}\log\frac{1}{\eps})}$ on the total time evolution and number of experiments, respectively.
The only thing that changes is the classical running time, which in iteration $j$ is now $\bigO{\qubits^2 \sparse' e^{7\locality}}$ times the number of experiments where $\sparse'$ is the effective $\frac{\eta_js_jt}{20}$-sparsity of $s_j t(H_j - H)$ (which is the Hamiltonian we are applying \cref{thm:structure} to).  Since $\ltnorm{H_{j} - H} \leq 4 \eta_{j} \sqrt{\sparse}$, this implies $\sparse' \leq \bigO{\sparse}$.  This gives us a final bound of
\begin{equation*}
    \bigO[\Big]{\qubits^2 \sparse^3 \locality^{6C \locality} \log\frac{1}{\eps}\log\parens[\big]{\frac{\qubits}{\delta}\log\frac{1}{\eps}}}.
\end{equation*}
Finally, by \cref{thm:structure}, the sparsity of the output is $\bigO{ne^{7\locality}\sparse}$.
\end{proof}

We now apply our result to learn Hamiltonians drawn from a particular class.

\subsection{Application to local Hamiltonians on bounded-degree graphs}

We consider the following class of Hamiltonians.

\begin{definition}[Low-intersection Hamiltonian~\cite{hkt24}]
\label{def:low-insersection-ham}
    For a $\locality$-local Hamiltonian $H = \sum_{a=1}^\terms \lambda_a E_a$ with $\terms$ terms on a system of $\qubits$ qubits, fix a parameter $\realdeg$.
    Let the \emph{dual interaction graph} $\graph$ be an undirected graph with vertices labeled by the $a \in [\terms]$ such that $\lambda_a \neq 0$, with an edge between $a, b \in [\terms]$ if and only if
    \begin{align*}
        \supp(E_a) \cap \supp(E_b) \neq \varnothing.
    \end{align*}
    We say that $H$ is \emph{low-intersection}\footnote{Also called a ``low-interaction''~\cite{htfs23} or ``sparsely interacting''~\cite{gcc24} Hamiltonian.} with respect to parameter $\realdeg$ if every vertex in $\graph$ has degree at most $\realdeg$.
\end{definition}

\begin{corollary}
    Let $H = \sum_{a=1}^\terms \lambda_a E_a$ be a $\locality$-local low-intersection Hamiltonian, possibly with unknown graph.
    Suppose we know $\locality$, $\realdeg$ and suppose $\locality = \bigO{1}$.
    Then we can find some $\hat{\lambda}$ such that $\infnorm{\hat{\lambda} - \lambda} < \eps$ with probability $\geq 1-\delta$ using 
    \begin{enumerate}[label=\textup{(\alph*)}]
        \item $\bigO{\frac{\realdeg}{\eps}\log\frac{\qubits}{\delta}}$ total time evolution;
        \item $\bigOmega{1/\realdeg^{1.5}}$ time resolution;
        \item $\bigOt{\realdeg^2\terms\log\frac{1}{\eps}\log\frac{\qubits}{\delta}}$ classical runtime if we know the terms $E_a$ in advance, and $\bigOt{\qubits^2 \realdeg^3\log\frac{1}{\eps}\log\frac{\qubits}{\delta}}$ if we don't;
        \item $\bigOt{\realdeg^2\log\frac{1}{\eps}\log\frac{\qubits}{\delta}}$ many experiments.
    \end{enumerate}
\end{corollary}
\begin{proof}
A low-intersection Hamiltonian $H$ satisfies $\lonorm{H} \leq \realdeg$.
Furthermore, there are at most $\realdeg$ non-zero terms supported on any given site $i \in [\qubits]$, so the effective sparsity for any $\eps$ is at most $\sparse_{\eps} \leq \realdeg$.
The desired statement now follows immediately from \cref{thm:main}.
\end{proof}

\subsection{Application to Hamiltonians exhibiting power law decay}

We now consider Hamiltonians with long-range interactions.
In particular, unlike the previous setting, we allow all interactions to have non-zero strength, provided that this strength decays polynomially with the distance of the interaction.
\begin{definition}[Hamiltonians with power law interactions]
\label{def:power-law-Ham}
    Consider a system of $\qubits$ qubits on a $d$-dimensional lattice.
    For a $\locality$-local Hamiltonian $H = \sum_{a=1}^\terms \lambda_a E_a$ with $\terms$ terms, fix a parameter $\alpha > 0$.
    We say that $H$ exhibits \emph{$\alpha$-power law decay} if, for every $i, j \in [\qubits]$,
    \begin{align*}
        \sum_{\substack{a \in [\terms] \\ \{i,j\} \in \supp(E_a)}} \abs{\lambda_a} \leq \frac{1}{\max(1,\operatorname{dist}(i,j))^\alpha},
    \end{align*}
    where $\operatorname{dist}(i,j)$ is the length of the shortest path between $i$ and $j$ on the lattice.
\end{definition}
We specialize to lattices here for simplicity, but our results easily extend to the generalization of power law decay to other graphs~\cite{hastings10a,Hastings10}.

\begin{lemma}\label{lem:power-law-sparsity}
    Let $H$ be a Hamiltonian on a $d$-dimensional lattice with $\alpha$-power law decay for $\alpha > d$.
    Then its effective sparsity satisfies
    \begin{align*}
        \sparse_\eps \leq 2^{d\locality + 1}/(\eps(\alpha - d))^{d\locality/(d\locality+(\alpha-d))} \,.
    \end{align*}
\end{lemma}
\begin{proof}
For a set $S \subset [\qubits]$, let $\text{diam}(S)$ denote the diameter of $S$ on the lattice graph.
Fix a qubit $i \in [\qubits]$, and let $R$ be a number that we fix later.
Consider the set of terms supported on $i$, and partition them into two sets: $\mathcal{A}_i$ denotes the terms $E_a$ such that $i \in \supp(E_a)$ and $\diam(\supp(E_a)) > R$, and $\mathcal{B}_i$ denotes the terms $E_a$ such that $i \in \supp(E_a)$ and $\text{diam}(\supp(E_a)) \leq R$.
Then
\begin{align*}
    \sum_{\substack{a \in [\terms] \\ i \in \supp(E_a)}} \min(1, \abs{\lambda_a}/\eps)
    &= \sum_{E_a \in \mathcal{A}_i} \min(1, \abs{\lambda_a}/\eps) + \sum_{E_a \in \mathcal{B}_i}  \min(1, \abs{\lambda_a}/\eps) \\
    &\leq \sum_{E_a \in \mathcal{A}_i} \abs{\lambda_a}/\eps + \abs{\mathcal{B}_i}\,.
\end{align*}
We now bound the two parts, using that the number of $j$ such that $\operatorname{dist}(i, j) = \ell$ is at most $2^d \binom{d+\ell-1}{d-1} \leq 2^d\ell^{d-1}$.
First,
\[
    \sum_{E_a \in \mathcal{A}_i} \abs{\lambda_a}
    \leq \sum_{j : \operatorname{dist}(i,j) > R}\frac{1}{\operatorname{dist}(i,j)^{\alpha}}
    \leq \sum_{\ell > R}^{\infty} \frac{2^d\ell^{d-1}}{\ell^{\alpha}}
    \leq 2^d\int_{\floor{R}}^{\infty} x^{d-1-\alpha} \diff x
    \leq \frac{2^dR^{d-\alpha}}{\alpha - d} \,. 
\]
Second,
\[
    \abs{\mathcal{B}_i} \leq (2R)^{d \locality} \,.
\]
Thus,
\begin{align*}
    \sum_{\substack{a \in [\terms] \\ i \in \supp(E_a)}} \min(1, \abs{\lambda_a}/\eps)
    &\leq \frac{2^dR^{d-\alpha}}{\eps(\alpha - d)} + (2R)^{d \locality}\,.
\end{align*}
We choose $R$ to balance these terms.
\begin{align*}
    R = \parens[\big]{2^{d(\locality-1)}\eps(\alpha - d)}^{-1/(d(\locality-1)+\alpha)}
\end{align*}
This gives a final bound of
\begin{multline*}
    \sparse_{\eps}
    = \max_{i \in [\qubits]} \sum_{\substack{a \in [\terms] \\ i \in \supp(E_a)}} \min(1, \abs{\lambda_a}^2/\eps^2)
    \leq \max_{i \in [\qubits]} \sum_{\substack{a \in [\terms] \\ i \in \supp(E_a)}} \min(1, \abs{\lambda_a}/\eps) \\
    \leq 2 \cdot 2^{d\locality} \parens[\big]{2^{d(\locality-1)}\eps(\alpha - d)}^{-d\locality/(d(\locality-1)+\alpha)}
    \leq 2^{d\locality + 1}/(\eps(\alpha - d))^{-d\locality/(d\locality+(\alpha-d))} \qedhere
\end{multline*}
\end{proof}

\begin{corollary} \label{cor:power-law}
    Let $H = \sum_{a=1}^\terms \lambda_a E_a$ be a $\locality$-local Hamiltonian on a $d$-dimensional lattice with $\alpha$-power law decay.    Suppose we know $d, \locality, \alpha$, and suppose $d,\locality = \bigO{1}$ and $\alpha > d$ with $\alpha - d = \bigOmega{1}$.
    Let
    \[
        \kappa = \frac{d\locality}{d\locality + (\alpha-d)} \,.
    \]
    Then we can find some $\hat{\lambda}$ such that $\infnorm{\hat{\lambda} - \lambda} < \eps$ with probability $\geq 1-\delta$ using 
    \begin{enumerate}[label=\textup{(\alph*)}]
    \item $\bigO{\frac{1}{\eps^{1 + \kappa}}\log\frac{\qubits}{\delta}}$ total time evolution;
    \item $\bigOmega{\eps^{\kappa}}$ time resolution;
    \item $\bigOt{\frac{\qubits^2}{\eps^{3\kappa}}\log\frac{1}{\delta}}$ classical runtime;
    \item $\bigOt{\frac{1}{\eps^{2\kappa}}\log\frac{\qubits}{\delta}}$ many experiments.
    \end{enumerate}
\end{corollary}
\begin{remark}
The exponent, $1 + \kappa$ is at most 2 for all choices of $\alpha$.
As $\alpha$ increases, and the power law decay becomes stronger, $\kappa$ goes to $0$.
This recovers Heisenberg scaling in the limit.
\end{remark}
\begin{proof}
By assumption, we have $\lonorm{H} \leq 1$ and \cref{lem:power-law-sparsity} tells us that $\sparse_{\eps} \lesssim 1/\eps^{\kappa}$.  Now we can apply \cref{thm:main} and immediately get the desired bounds.
\end{proof}

\section*{Acknowledgments}
\addcontentsline{toc}{section}{Acknowledgments}

ET thanks Robin Kothari and Jeongwan Haah for discussions which sowed initial ideas that would later develop into this work.

AB is supported by Ankur Moitra's ONR grant and the NSF TRIPODS program (award DMS-2022448). AL is supported in part by an NSF GRFP and a Hertz Fellowship.  AM is supported in part by a Microsoft Trustworthy AI Grant, an ONR grant and a David and Lucile Packard Fellowship.  ET is supported by the Miller Institute for Basic Research in Science, University of California Berkeley. 

\printbibliography

\appendix
\end{document}